\documentclass[12pt,a4paper]{article}
\usepackage[utf8]{inputenc}
\usepackage[english]{babel}
\usepackage{graphicx}
\usepackage{amsfonts,amsmath,amssymb,amsthm}
\usepackage[dvipsnames]{xcolor}
\usepackage[toc,page]{appendix}
\usepackage{authblk}
\usepackage[bookmarksnumbered=true]{hyperref}
\usepackage{tikz}
\usetikzlibrary{decorations.pathreplacing, patterns}
\usepackage{capt-of, caption}  
\usepackage[capitalise]{cleveref}
\crefname{equation}{}{}

\numberwithin{equation}{section}
\title{Momentum Distribution of a Fermi Gas in the Random Phase Approximation}

\author[1,*]{Niels Benedikter}
\author[2,*]{Sascha Lill}
\affil[1]{ORCID: \href{https://orcid.org/0000-0002-1071-6091}{0000-0002-1071-6091}, e--mail: \href{mailto:niels.benedikter@unimi.it}{niels.benedikter@unimi.it}}
\affil[2]{ORCID: \href{https://orcid.org/0000-0002-9474-9914}{0000-0002-9474-9914}, e--mail: \href{mailto:sascha.lill@unimi.it}{sascha.lill@unimi.it}}
\affil[*]{Università degli Studi di Milano, Via Cesare Saldini 50, 20133 Milano, Italy}

\newcommand{\bA}{\boldsymbol{A}}
\newcommand{\bB}{\boldsymbol{B}}
\newcommand{\bC}{\boldsymbol{C}}
\newcommand{\bD}{\boldsymbol{D}}
\newcommand{\bE}{\boldsymbol{E}}
\newcommand{\bF}{\boldsymbol{F}}
\newcommand{\cA}{\mathcal{A}}
\newcommand{\cC}{\mathcal{C}}
\newcommand{\cD}{\mathcal{D}}
\newcommand{\cE}{\mathcal{E}}
\newcommand{\cF}{\mathcal{F}}
\newcommand{\cI}{\mathcal{I}}
\newcommand{\cN}{\mathcal{N}}
\newcommand{\cO}{\mathcal{O}}
\newcommand{\cS}{\mathcal{S}}
\newcommand{\fn}{\mathfrak{n}}
\newcommand{\fC}{\mathfrak{C}}
\newcommand{\CCC}{\mathbb{C}}
\newcommand{\NNN}{\mathbb{N}}
\newcommand{\RRR}{\mathbb{R}}
\newcommand{\TTT}{\mathbb{T}}
\newcommand{\ZZZ}{\mathbb{Z}}
\newcommand{\Zbb}{\mathbb{Z}}
\newcommand{\1}{\mathbb{I}}
\renewcommand{\a}{\textnormal{a}}
\newcommand{\ad}{\mathrm{ad}}
\renewcommand{\b}{\textnormal{b}}
\newcommand{\Coul}{\textnormal{Coul}}
\renewcommand{\d}{\textnormal{d}}
\newcommand{\di}{\textnormal{d}}
\newcommand{\DV}{\mathrm{DV}}
\newcommand{\diam}{\mathrm{diam}}
\newcommand{\eff}{\mathrm{eff}}
\newcommand{\F}{\mathrm{F}}
\newcommand{\HF}{\mathrm{HF}}
\newcommand{\I}{\mathrm{I}}
\newcommand{\II}{\mathrm{II}}
\newcommand{\III}{\mathrm{III}}
\newcommand{\IV}{\mathrm{IV}}
\newcommand{\V}{\mathrm{V}}

\newcommand{\nor}{\mathrm{nor}}
\newcommand{\op}{\mathrm{op}}

\newcommand{\RPA}{\mathrm{RPA}}
\newcommand{\SR}{\mathrm{SR}}
\newcommand{\supp}{\mathrm{supp}}

\newcommand{\norm}[1]{\lVert #1 \rVert}
\newcommand{\kF}{k_\F}
\newcommand{\BF}{B_\F}
\newcommand{\BFc}{B_\F^c}
\newcommand{\Ik}{\mathcal{I}_k}
\newcommand{\north}{\Gamma^{\textnormal{nor}}}
\newcommand{\fock}{\mathcal{F}}
\newcommand{\Ncal}{\mathcal{N}}
\newcommand{\Ecal}{\mathcal{E}}
\newcommand{\Nbb}{\mathbb{N}}
\newcommand{\Ical}{\mathcal{I}}
\newcommand{\Ccal}{\mathcal{C}}
\newcommand{\Cbb}{\mathbb{C}}
\newcommand{\tagg}[1]{ \stepcounter{equation} \tag{\theequation}
\label{#1} } 

\newtheorem{theorem}{Theorem}[section]

\newtheorem{proposition}[theorem]{Proposition}
\newtheorem{lemma}[theorem]{Lemma}
\theoremstyle{definition}

\theoremstyle{definition}

\theoremstyle{definition}

\theoremstyle{definition}

\begin{document}
\maketitle
\begin{abstract}
We consider a system of interacting fermions on the three-dimensional torus in a mean-field scaling limit. Our objective is computing the occupation number of the Fourier modes in a trial state obtained through the random phase approximation (in its collective bosonization formulation) for the ground state. We prove that the trial state's momentum distribution has a jump discontinuity, i.\,e., a well-defined Fermi surface. Moreover the Fermi momentum does not depend on the interaction potential (it is universal). Our result shows that the random phase approximation in the mean-field scaling limit is in principle sufficiently precise to identify a non-trivial Fermi liquid phase.
\end{abstract}


\section{Introduction and Main Result}
\label{sec:intro}
\label{subsec:mainresult}

We consider a quantum system of $ N $ spinless fermionic particles on the torus $ \TTT^3 := [0, 2 \pi]^3 $, which may be understood as a simple model of a metal. This system is described by the Hamilton operator
\begin{equation}
	H_N := \sum_{j = 1}^N - \hbar^2 \Delta_{x_j} + \lambda \sum_{i < j}^N V(x_i - x_j)
\label{eq:HN}
\end{equation}
acting on wave functions in the antisymmetric tensor product $ L_{\a}^2(\TTT^{3N}) = \bigwedge_{i=1}^N L^2(\TTT^3) $.
As the general case with $N \simeq 10^{23}$ is too difficult to analyze, we will consider the asymptotics for particle number $ N \to \infty $ in the \emph{mean-field scaling limit} introduced by \cite{NS81}, i.\,e., we set
\begin{equation}
	\hbar := N^{-\frac 13} \;, \qquad \text{and} \qquad \lambda := N^{-1}\;.
\end{equation}
The ground state energy of the system is defined as the infimum of the spectrum
\[E_N := \inf \sigma(H_N) = \inf_{\substack{\psi \in L^2_\textnormal{a}(\mathbb{T}^{3N})\\\norm{\psi}=1}} \langle \psi,H_N \psi\rangle\;.\]
Any eigenvector of $H_N$ with eigenvalue $E_N$ is called a ground state. In the present paper we analyze the momentum distribution (i.\,e., the Fourier transform of the one--particle reduced density matrix) of the random phase approximation of the ground state.

\smallskip

In the non-interacting case of interaction potential $ V = 0 $, the ground states are given by Slater determinants comprising $N$ plane waves with different momenta $k_j \in \ZZZ^3$ of minimal kinetic energy $ |k_j|^2 $, i.\,e.,
\begin{equation}
\label{eq:planewaveslater}
\psi(x_1,x_2,\ldots,x_N) = \frac{1}{\sqrt{N!}} \det\left(\frac{1}{(2\pi)^{3/2}} e^{i k_j \cdot x_i}\right)_{j,i=1}^N \;.
\end{equation}
This is (up to a phase) unique if we impose that the number of particles exactly fills a ball in momentum space; i.\,e., if
\begin{equation}
\label{eq:fermiball}
 N = |B_{\F}| \quad \textnormal{for} \quad	 B_{\F} := \{k \in \ZZZ^3 : |k| \le k_{\F} \} \quad \textnormal{with some} \quad k_{\F} > 0\;.
\end{equation}
This means that the Fermi momentum $\kF$ scales like\footnote{In \cite{BNPSS21}, $ \kappa $ is defined as $ (\frac 34 \pi)^{\frac 13} $, so in that notation $ \kappa = \hbar k_{\F} (1 + \cO(\hbar)) $.}
\begin{equation}
k_{\F} = \kappa N^{\frac 13} \qquad \text{with} \qquad \kappa = \left( \frac{3}{4 \pi} \right)^{\frac 13} + \cO(N^{-\frac 13})\;.
\label{eq:kappa}
\end{equation}
The set of momenta $\BF$ is called the Fermi ball. We also define its complement
\[\BFc := \ZZZ^3 \setminus \BF \;.\]

As a first step towards including the effects of a non-vanishing interaction potential $V$ one may consider the Hartree--Fock approximation. In this approximation, the expectation value $\langle \psi, H_N \psi \rangle$ is minimized over the choice of $N$ orthonormal orbitals $\{\varphi_j: j = 1,\ldots, N\} \subset L^2(\TTT^3)$ in the Slater determinant
\[\psi(x_1,x_2,\ldots,x_N) = \frac{1}{\sqrt{N!}} \det\left(\varphi_j(x_i)\right)_{j,i=1}^N \;.\]
(This is to be compared to the general quantum many-body problem, in which also linear combinations of Slater determinants are permitted.)
In general, the Hartree-Fock minimizer will not have plane waves as orbitals. However, with our particular assumptions on the potential, the scaling limit, and the particle number, one can show \cite[Appendix A]{BNPSS21} that \cref{eq:planewaveslater} is also the (unique up to a phase) Hartree--Fock minimizer. Thus in the Hartree--Fock approximation the momentum distribution remains the trivial
\begin{equation}
\label{eq:nkhf}
\langle \psi, a^*_q a_q \psi\rangle = \begin{cases} 0 \quad &\text{for } q \in B_{\F}^c \\
		1 \quad &\text{for } q \in B_{\F} \;. \end{cases}
\end{equation}

It is highly non-trivial to understand if this jump in the momentum distribution survives the presence of an interaction potential when going beyond Hartree--Fock theory, and if it does, how its location and height are affected by the interaction. In physics, it has become known as \emph{Luttinger's theorem} that the ``interaction may
deform the FS [Fermi surface], but it cannot change its volume. In the
isotropic case, where symmetry requires the FS to
remain a sphere, its radius must then remain $k_\F$ (the
Fermi momentum of the unperturbed system)'' \cite{Lut60}. In other words, the Fermi momentum is conjectured to be universal, i.\,e., independent of the interaction potential $V$. This is in contrast to the height of the jump, called \emph{quasiparticle weight} $Z$, which generally depends on $V$. As discussed, the Hartree--Fock approximation predicts that $k_\F$ is independent of $V$, but also that the quasiparticle weight is $Z=1$ independent of $V$. So to observe any effect of the interaction, we have to employ a more precise approximation for the ground state. In \cite{BNPSS20,BNPSS21,BPSS22,CHN21,CHN22,CHN23,Chr23PhD,CHN24} it has been shown that the ground state energy can be approximated to higher precision using the random phase approximation, in its formulation as bosonization of particle--hole pair excitations. In fact, there it was shown that
the ground state energy in the mean-field scaling limit has an expansion as
\begin{equation}
\label{eq:rpaprecision}
 E_N = E^{\textnormal{HF}}_N + E^{\textnormal{RPA}}_N + \mathcal{O}(N^{-1/3-\alpha})
\end{equation}
for some $\alpha > 0$. Here $E^{\textnormal{HF}}_N$ is the expectation value of the Hamiltonian in the Slater determinant of plane waves \cref{eq:planewaveslater}; it is a sum of three terms called the kinetic, direct, and exchange term, of orders $N$, $N$, and $N^{0}$, respectively (unless $V$ is taken as the Coulomb potential). The correction $E^{\textnormal{RPA}}_N$ to the Hartree--Fock energy is given by an explicit integral formula \cite{BNPSS20,BNPSS21,BPSS22} of order $N^{-1/3}$. (The validity of the expansion to the order of $E^{\textnormal{HF}}_N$ was proven much earlier by \cite{GS94}; Hartree--Fock theory has moreover been derived as the leading-order approximation of the dynamics in \cite{BPS14,BD23} and with mixed states as initial data in \cite{BJPSS16}.)

Our goal in this paper is to exhibit the prediction of the random phase approximation for the momentum distribution. We will take a trial state constructed by bosonization and compute the deviation of its momentum distribution from \cref{eq:nkhf},
\begin{equation}
	n_q := \begin{cases}
		\langle \psi, a_q^* a_q \psi \rangle		\quad &\text{for } q \in B_{\F}^c \\
		1 - \langle \psi, a_q^* a_q \psi \rangle	\quad &\text{for } q \in B_{\F}	\;.
	\end{cases}
\label{eq:nqEbar}
\end{equation}
Our trial state is the same as in \cite{BNPSS20,BNPSS21,BPSS22}. Trial states of a similar form have been used in \cite{CHN23,Chr23PhD} for mean-field Fermi gases, as well as for dilute Fermi gases in \cite{FGHP21,Gia22,Gia23,GHNS24}.

Since to obtain sufficiently sharp estimates we use a technically complicated bootstrap, in this paper we only consider interaction potentials $V$ with compactly supported Fourier transform\footnote{The interaction $ V(x-y) $ being a two-particle multiplication operator, we use the convention $ V(x) = \sum_{k \in \ZZZ^3} \hat{V}_k e^{ik \cdot x} $ for its Fourier transform. This is in contrast to wave functions, whose Fourier transform is defined to be $L^2$-unitary, $ \psi(x_1, \ldots, x_N) = (2 \pi)^{-\frac{3N}{2}} \sum_{k_1, \ldots, k_N \in \ZZZ^3} \hat{\psi}(k_1, \ldots, k_N) e^{i(k_1 \cdot x_1 + \ldots + k_N \cdot x_N)} $.}. This should be generalizable (for the ground state energy even the Coulomb case has been covered \cite{CHN24}) without fundamental changes but at the cost of readability. So to state our main theorem, we assume that $ \supp(\hat{V}) \subset B_R(0) $ for some $ R > 0 $. Further, we adopt the convention that $ C $ is a positive constant (in particular not depending on $ N$, $V $, or $ q $) but whose value may change from line to line.

\begin{theorem}[Main Result]
\label{thm:main}
Assume that the Fourier transform $ \hat{V} $ of the interaction potential is non-negative and compactly supported.
Then, there exists a sequence of trial states $ \psi_N \in L_{\a}^2(\TTT^{3N}) $ with particle numbers $N$ corresponding to completely filled Fermi balls as in \cref{eq:fermiball} such that
\begin{itemize}
\item the sequence of trial states is energetically close to the ground state up to the precision of the random phase approximation (see \cref{eq:rpaprecision}) in the sense that there exists some $ \alpha > 0 $ such that
\begin{equation}
\label{eq:mainenergy}
	\langle \psi_N, H_N \psi_N \rangle - E_N \le C N^{-\frac{1}{3}-\alpha}\;;
\end{equation}
\item and for any $ \epsilon > 0 $ and all momenta $q \in \Zbb^3$ such that\footnote{This restriction is a technical assumption needed for the step from the bosonized momentum distribution \cref{eq:nqb} to the explicit integral in \cref{eq:main}, compare to the estimates of \cref{eq:threesuberrors}. The justification of the bosonization, \cref{thm:main2}, is valid independently.}
\begin{equation}
\label{eq:cQepsilon}
	q \notin \left\{ p \in \ZZZ^3 \mid \exists k \in B_R(0) : \frac{|k \cdot p|}{\lvert k \rvert \lvert p \rvert} \in (0, \epsilon) \right\}
\end{equation}
the trial states' momentum distribution $ n_q $ can be estimated as
\begin{equation}
\label{eq:main}
	0
	\le n_q
	\le N^{-\frac{2}{3}} \sum_{k \in \cD^q \cap \ZZZ^3}  \frac{\hat{V}_k}{2 \kappa |k|} \frac{1}{\pi} \int_0^\infty \frac{(\mu^2 - \lambda_{q,k}^2)(\mu^2 + \lambda_{q,k}^2)^{-2}}{1 + Q_k^{(0)}(\mu)} \; \d \mu + \cE\;,
\end{equation}
where
\begin{equation}
\label{eq:abbreviations1}
\begin{aligned}
	\lambda_{q,k} &:= \frac{|k \cdot q|}{\lvert k \rvert \lvert q \rvert}\;, \qquad
	Q_k^{(0)}(\mu) := 2\pi \kappa \hat{V}_k \left( 1 - \mu \arctan \left( \frac{1}{\mu} \right) \right) \;, \\
	\cD^q &:= \begin{cases}
		\{ k \in B_R(0) : q+k \in B_{\F}^c \} \quad & \textrm{if } q \in B_{\F} \\
		\{ k \in B_R(0) : q-k \in B_{\F}\} \quad & \textrm{if } q \in B_{\F}^c \;.\\
	\end{cases}
\end{aligned}
\end{equation}
\end{itemize}
The error term $\mathcal{E}$ is bounded by
\begin{equation}
\label{eq:mainerror}
	\lvert \cE \rvert
	\le C \epsilon^{-1} N^{-\frac{2}{3}- \frac{1}{12}}\;.
\end{equation}
In the trial state $ \psi_N $ which we construct in \cref{sec:trialstate}, for ``most'' $ q \in \ZZZ^3 $, the upper bound in \cref{eq:main} is an equality. Due to the technical details of the trial state construction, we defer the precise definition of ``most'' to Proposition \ref{prop:optimality}.
\end{theorem}
The theorem is proven in \cref{sec:proofmain}, based on the sharpened bosonization strategy explained in \cref{sec:strategyofproof}. A key role if played by the bootstrap of the bosonization approximation, which justifies the use of the quasibosonic picture for $n_q$ as a ``pointwise'' observable compared to the earlier results on the random phase approximation concerning only the ground state energy in which the properties of the state enter only ``averaged'' over the entire range of momenta.

%
One can also obtain a series expansion for $ n_q $ in terms of Friedrichs diagrams \cite{Lil23,BL23}. The bosonized momentum distribution $ n_q^{(\b)} $ arises by restricting to a subset of diagrams. The next-smaller diagrams render contributions of order $N^{-1} $, to be compared to the error bound \eqref{eq:mainerror} of order $N^{-\frac{2}{3} - \frac{1}{12}}$. However, no convergence of the diagrammatic expansion was established. Establishing convergence is a common difficulty with perturbative expansions requiring significant effort, whereas in our present analysis it is proven rather easily (see the proof of \cref{lem:bootstrap}).

As as a measure for the height of the jump at the Fermi surface we define the \emph{quasiparticle weight} as
\[Z := 1 - \sup_{q \in \BF} n_q - \sup_{q \in \BFc} n_q \;.\]
(Note that in our convention \eqref{eq:nqEbar}, $n_q$ only represents the excitations with respect to the non-interacting Fermi ball.) As a corollary of our main theorem, we obtain an estimate for $Z$.
\begin{theorem}[Jump at the Fermi Surface]
\label{thm:jump}
Under the assumptions of \cref{thm:main}, the trial states $ \psi_N \in L_{\a}^2(\TTT^{3N}) $ exhibit a jump at the Fermi surface, in the sense that
\begin{equation}\label{eq:jump}
	 Z \geq 1 - C N^{- \frac{2}{3}+ \frac{1}{12}}\;.
\end{equation}
\end{theorem}
This does not depend on $ \epsilon $, nor is there any restriction like in \eqref{eq:cQepsilon}.
The proof is given in \cref{sec:proofmain}. We expect that the sharp bound is of the form $Z \geq 1 - C N^{- \frac{2}{3}}$, as proposed in the physics literature \cite{DV60} extrapolated to the mean-field scaling limit.

\smallskip

The presence of a jump in the momentum distribution is characteristic of the Fermi liquid phase. Fermi liquid theory was phenomenologically introduced by Landau \cite{Lan56}, who argued that the interaction becomes suppressed due to correlations between particles, resulting in a system that on mesoscopic scales appears to be composed of extremely weakly interacting fermions (which share the quantum numbers of the electrons but have, e.\,g., a renormalized mass). Bosonization as a microscopic justification of Fermi liquid theory was suggested by \cite{Hal94,CF94}.

A rigorous proof of Fermi liquid theory has been undertaken in a series of ten papers surveyed in \cite{FKT04} in spatial dimension $d=2$. This program used multiscale methods of constructive field theory to construct a convergent perturbation series. To suppress the superconducting instability, an asymmetric Fermi surface was assumed; an example was constructed in \cite{FKT00}. The case $d=3$ was partially analyzed by \cite{DMR01}. In this framework, Salmhofer's criterion \cite{Sal98} was formulated as a more precise characterization of the Fermi liquid phase; however, it makes reference to positive temperature, which we believe unnecessary in the mean-field scaling limit.

Our results concern a trial state, not the actual ground state. It is one of the fundamental problems of mathematical condensed matter theory to understand if similar statements hold for the ground state. This is a very subtle question. It is expected that due to the Kohn--Luttinger instability the ground state always has a superconducting part which smoothens out the discontinuity. However, we conjecture that the Kohn--Luttinger effect changes the momentum distribution only on a scale which is extremely close to the Fermi surface, well separated from the scale on which we observe the characteristic Fermi liquid behavior. A rigorous proof is challenging because a-priori bounds through the ground state energy would have to be of extremely high precision; in fact, a single pair excitation $a^*_p a^*_h$ may change the momentum distribution from $\langle \psi_N, a^*_p a_p \psi_N \rangle =0 $ to $\langle \psi_N, a^*_p a_p \psi_N \rangle =1 $ at a kinetic energy cost as small as order $\hbar^2 = N^{-2/3}$; this has to be compared to the resolution of the ground state energy \cref{eq:mainenergy} that is only $\hbar N^{-\alpha} = N^{-\frac{1}{3}-\alpha}$. Nevertheless, we believe our result is interesting for two main reasons. First, it shows that the random phase approximation in the mean-field scaling limit is sufficient to identify a Fermi liquid: it is neither trivial (as the Hartree--Fock approximation), nor do we need higher orders of the expansion in $N^{-1/3}$. Second, our trial state \cref{eq:psi} being given in terms of unitary transformations $\mathfrak{R}$ and $T$ as $\mathfrak{R}T\Omega$, it is natural to study the transformed Hamiltonian $T^* \mathfrak{R}^* H_N \mathfrak{R} T$ to obtain a-priori estimates on the deviation of the momentum distribution from its random phase approximation formula. It remains very difficult to push this approach to the precision obtained for the trial state and which would prove Fermi liquid behavior, but some rough statements for sufficiently wide averages in momentum are possible; a discussion of the obtainable estimates shall appear elsewhere.

\subsection{Comparison with the Physics Literature}
\label{subsec:literaturecomp}
The physics literature \cite{DV60, Lam71a, Lam71b} considers the system in the thermodynamic limit, where sums over momenta become integrals. Our estimates are not uniform in the system's volume, but we can formally extrapolate \cref{eq:main} to the thermodynamic limit.
 To do so, we rescale the torus $ \TTT^3 $ to the torus $ L \TTT^3 = [0, 2 \pi L]^3 $. The corresponding momentum space is $ L^{-1} \ZZZ^3 $ and we can replace sums over $ \ZZZ^3 $ by sums over $ L^{-1} \ZZZ^3 $. The number of momenta in the Fermi ball $ B_{\F} := \{ k \in L^{-1} \ZZZ^3 \; : \; |k| \le k_{\F} \} $ is now
\begin{equation}
	N = |B_{\F}| \approx \frac{4 \pi}{3} k_{\F}^3 L^3 \;.
\end{equation}
We consider $L \to \infty$ followed by the high density limit $\kF \to \infty$. The density is
\begin{equation}
	\rho := \frac{N}{(2 \pi L)^3} = \frac{k_{\F}^3}{6 \pi^2} (1 + \cO(k_{\F}^{-1} L^{-1}))\;.
\label{eq:rho}
\end{equation}
Setting $ \hbar := k_{\F}^{-1} $, we can define $ Q_k^{(0)}(\mu) $ via \cref{eq:abbreviations1}, and the r.\,h.\,s.\ of \cref{eq:main} becomes
\begin{equation}
\label{eq:zzz}
	n_q(k_{\F}, L)
	\approx \sum_{k \in \cD^q \cap L^{-1} \ZZZ^3} \frac{1}{\pi} \frac{\hat{V}_k}{2 \hbar \kappa N |k|} \int_0^\infty \frac{(\mu^2 - \lambda_{q,k}^2)(\mu^2 + \lambda_{q,k}^2)^{-2}}{1 + Q_k^{(0)}(\mu)} \; \d \mu \;.
\end{equation}
In \cref{subsec:infvolapp} we compute the large-volume limit of \cref{eq:zzz} and obtain \cref{eq:nqbfinalresult}. For comparison, in \cref{subsec:DV60comp} we extrapolate \cite{DV60}'s result to short-ranged interaction potentials and obtain \cref{eq:nqDVSRfinal} for the momentum distribution outside the Fermi ball. In the high-density limit $ k_{\F} \to \infty $, \cref{eq:nqDVSRfinal} converges to half our result \cref{eq:nqbfinalresult}. In view of this remaining discrepancy we have moreover verified that our result agrees with the formula that can be obtained from \cite[Theorem~1.1]{Chr23PhD} through a formal Feynman--Hellmann argument. Unfortunately we have not been able to pin down the origin of the discrepancy in \cite{DV60}, presumably due to their only outlined adaption of the result of \cite{GB57}, where the latter also do not specify the choice of units or the Hamiltonian used.

\section{Construction of the Trial State}
\label{sec:trialstate}

The definition of the trial state $ \psi_N $ uses second quantization. That means, we extend the $ N $--particle space $ L_{\a}^2(\TTT^{3N}) $ by introducing the fermionic Fock space
\[
\cF := \bigoplus_{n = 0}^\infty L_{\a}^2(\TTT^{3n})\;.
\]
To each momentum mode $ q \in \ZZZ^3 $, we assign the plane wave
\[
f_q \in L^2(\TTT^3)\;, \qquad f_q(x) := (2 \pi)^{-\frac 32} e^{i q \cdot x}
\]
and the respective creation and annihilation operators on Fock space
\[
a_q^* := a^*(f_q)\;, \qquad a_q := a(f_q)
\]
which satisfy the canonical anticommutation relations (CAR)
\begin{equation}
	\{a_q, a_{q'}^*\} = \delta_{q, q'}, \qquad
	\{a_q, a_{q'}\} = \{a_q^*, a_{q'}^*\} = 0 \qquad \text{for all } q, q' \in \ZZZ^3\;.
\label{eq:CAR}
\end{equation}
The number operator on Fock space is defined as
\begin{equation}
	\cN := \sum_{q \in \ZZZ^3} a_q^* a_q
\label{eq:cN}
\end{equation}
and the vacuum vector is $ \Omega := (1, 0, 0, \ldots) \in \cF $, which satisfies $ a_q \Omega = 0$ for all $q \in \ZZZ^3 $.

\paragraph{The trial state} As a trial state $ \psi_N \in L_{\a}^2(\TTT^{3N}) \subset \cF $ for Theorem \ref{thm:main}, we use the state constructed by means of the random phase approximation (in its formulation as bosonization of particle--hole excitations) in \cite[(4.20)]{BNPSS20}, i.\,e.,
\begin{equation}
	\psi_N := \mathfrak{R} T \Omega\;,
\label{eq:psi}
\end{equation}
with a particle--hole transformation $ \mathfrak{R}: \cF \to \cF $ and an almost-bosonic Bogoliubov transformation $ T: \cF \to \cF $ that are both defined below. According to \cite{BNPSS20,BNPSS21,BPSS22}, the state \cref{eq:psi} is energetically close to the ground state, i.\,e., $\langle \psi_N, H_N \psi_N\rangle = E^{\textnormal{HF}}_N + E^{\textnormal{RPA}}_N + \mathcal{O}(N^{-\frac{1}{3}-\frac{1}{27}})$ in agreement with \cref{eq:rpaprecision}.

\paragraph{The particle--hole transformation} The particle--hole transformation is the unitary operator $ \mathfrak{R}: \cF \to \cF $ defined by its action on creation operators
\begin{equation}
	\mathfrak{R}^* a_q^* \mathfrak{R} := \chi(q \in \BFc)\, a^*_q  + \chi(q \in \BF) \, a_q\;, \qquad \forall q \in \Zbb^3
\end{equation}
and its action on the vacuum
\[\mathfrak{R}^*\Omega := \prod_{k_j \in \BF} a^*_{k_j} \Omega \,.\]
The latter product is (up to an irrelevant phase $e^{i\pi}$) uniquely defined and one easily verifies that it is a Slater determinant of plane waves as in \cref{eq:planewaveslater}. We have $\mathfrak{R}^* = \mathfrak{R}$.

\paragraph{Particle--hole pair operators} The key observation \cite{BNPSS20, BNPSS21, BPSS22} motivating the choice of $ T $ is that after the transformation $ \mathfrak{R} $, the Hamiltonian $ H $ becomes almost quadratic in some almost-bosonic operators $ c^*$ and $ c $. For their definition we use a patch decomposition of a shell around the Fermi surface. The requirements for that decomposition are described in the following; an example of such a patch decomposition was given in \cite{BNPSS20}. We divide half of the Fermi surface $ \partial B_{\F} := \{k \in \RRR^3: \lvert k\rvert = \kF\}$ into a number $ M/2 \in \NNN $ of patches $ \tilde{P_\alpha} $, each of surface area $ \sigma(\tilde{P_\alpha}) = 4 \pi k_{\F}^2/M $. The number of patches $ M $ is a parameter that will be chosen as a function of the particle number $N$, subject to the constraint
\begin{equation}
	N^{2\delta} \ll M \ll N^{\frac{2}{3}- 2 \delta}\;, \qquad \textnormal{where } 0 < \delta < \frac{1}{6}\;.
\label{eq:Mdelta}
\end{equation}
We assume that the patches do not degenerate into very long and narrow shapes as $N \to \infty$, or more precisely we assume that always
\begin{equation}
	\diam(\tilde{P}_\alpha)
	\le C N^{\frac{1}{3}} M^{-\frac{1}{2}}\;.
\label{eq:patchreg}
\end{equation}
Inside each $ \tilde{P}_\alpha $, we now choose a slightly smaller patch $ P_\alpha $ such that the distance between two adjacent patches is at least $ 2R $, that is, twice the diameter of the support of $ \hat{V} $. By radially extending $ P_\alpha $, we obtain the final patches $ B_\alpha $ with thickness $ 2R $ (see Figure \ref{fig:patches3d}):
\begin{equation}
	B_\alpha := \Big\{r q \in \RRR^3 : q \in P_\alpha, r \in \Big[ 1 - \frac{R}{k_{\F}}, 1 + \frac{R}{k_{\F}} \Big] \Big\}\;.
\label{eq:Balpha}
\end{equation}
The patches are separated by corridors wider than $ 2R $. To cover also the southern hemisphere, we define the patch $B_{\alpha + \frac M2}$ by applying the reflection $ k \mapsto -k $ to $ B_\alpha $\;.

\begin{figure}
\begin{minipage}{0.45\textwidth}\centering
 \includegraphics[width=4.5cm]{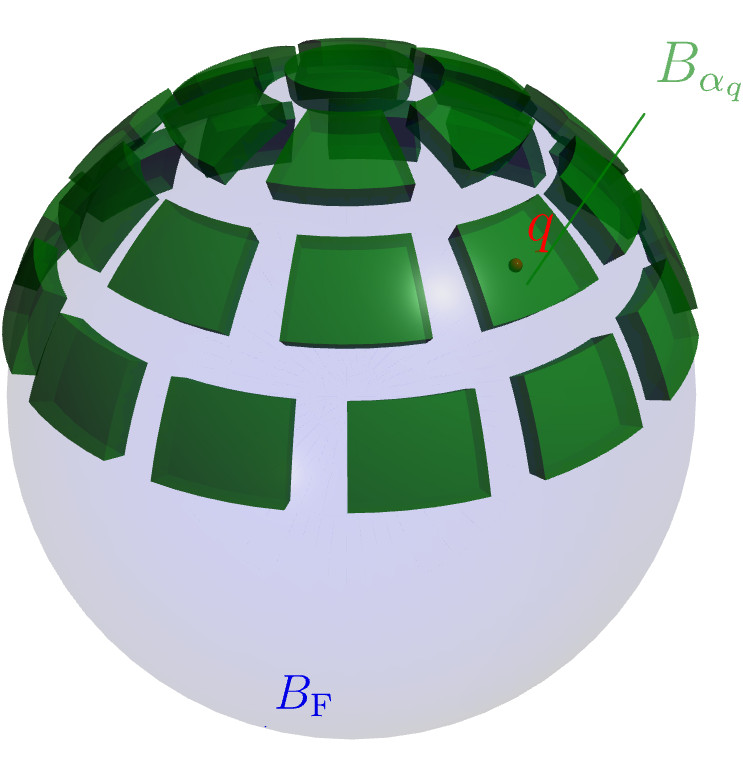}
 \caption{Patches on the Fermi ball in momentum space, with patch $ B_{\alpha_q} $ including $ q $.}
 \label{fig:patches3d}
 \end{minipage}
\hfill
\begin{minipage}{0.5\textwidth}\centering
	\scalebox{0.75}{\begin{tikzpicture}

\draw[thick, blue] (-3,2) arc(90:58:15);
\fill[opacity = .1, blue] (-3,-2.5) -- (-3,2) arc(90:58:15) -- ({-3+15*sin(32)},-2.5);

\filldraw[thick, fill opacity = .1, green!50!black] (-3,2.5) arc(90:80:15.5) -- ++({-sin(10)},{-cos(10)}) arc(80:90:14.5);
\filldraw[thick, fill opacity = .1, green!50!black] ({-3 + 15.5*sin(15)},{-13 + 15.5*cos(15)}) arc(75:60:15.5) -- ++({-sin(30)},{-cos(30)}) arc(60:75:14.5) -- cycle;

\filldraw[thick, blue!50!red, fill opacity = .1] (-1,2) circle (0.5);
\fill [red] (-1,2) circle (0.08) node[anchor = north west]{$ q_1 $};
\filldraw[thick, blue!50!red, dashed, fill opacity = .05] (0.5,1.6) circle (0.5);
\fill [red] (0.5,1.6) circle (0.08) node[anchor = south west]{$ q_2 $};
\filldraw[thick, blue!50!red, dashed, fill opacity = .05] (4.25,0) circle (0.5);
\fill [red] (4.25,0) circle (0.08) node[anchor = north west]{$ q_3 $};

\node[blue] at (4.3,-2) {$ B_{\F} $};
\node[gray] at (4.3,2.7) {$ B_{\F}^c $};
\draw[blue] ({-3 + 15*sin(3)},{-13 + 15*cos(3)}) -- ++(0.2,-0.8) node[anchor = north]{$ \partial B_{\F} $};
\draw[green!50!black] ({-3 + 15.5*sin(4)},{-13 + 15.5*cos(4)}) -- ++(0.5,0.3) node[anchor = west]{$ B_{\alpha_1} $};
\draw[green!50!black] ({-3 + 15.5*sin(20)},{-13 + 15.5*cos(20)}) -- ++(0.5,0.3) node[anchor = west]{$ B_{\alpha_2} $};

\draw[dashed] (-1,2) -- ++(0,-2.2);
\draw[dashed] (-1.5,2) -- ++(0,-2.2);
\draw[<-] (-1.5,0) -- ++(-0.3,0);
\draw[<-] (-1,0) -- ++(0.3,0);
\node at (-1.25,-0.4) {$ R $};

\end{tikzpicture}}
	\captionof{figure}{Close-up of a patch: in \cref{eq:edgeofthepatch}, $ q_1 $ is an included momentum, whereas $ q_2 $ and $ q_3 $ are excluded.}
	\label{fig:edgeofthepatch}
 \end{minipage}
\end{figure}

The center of patch $ B_\alpha $, a vector in $\RRR^3$, will be denoted $ \omega_\alpha \in P_\alpha $, with associated direction vector $ \hat{\omega}_\alpha := \omega_\alpha / |\omega_\alpha| $. For any $ k \in \ZZZ^3 \cap B_R(0) $, we define the index sets
\begin{equation}
\begin{aligned}
	\cI_{k}^+ & := \left\{ \alpha \in \{1, \ldots, M\} : k \cdot \hat{\omega}_\alpha \ge N^{-\delta}\right\}, \\
	\cI_{k}^- & := \left\{ \alpha \in \{1, \ldots, M\} : k \cdot \hat{\omega}_\alpha \le -N^{-\delta}\right\}, \\
	\cI_{k} & := \cI_{k}^+ \cup \cI_{k}^- \;.
\end{aligned}
\end{equation}
Visually speaking, $ \cI_k $ excludes a belt of patches near the equator of the Fermi ball (if the direction of $ k $ is taken as north).

Generally, we assume all momenta $ k, p, h $ to be in $ \ZZZ^3 $ and do not write this condition under summations. Moreover, we adopt the following convention: whenever a momentum is denoted by a lowercase $ p $ (``particle''), then we abbreviate the condition $ p \in B_{\F}^c \cap B_\alpha $ by $ p : \alpha $. Likewise, for a momentum denoted by $ h $ (``hole''), the condition $ h \in B_{\F} \cap B_\alpha $ is abbreviated as $ h : \alpha $. For $ k \in \ZZZ^3 \cap B_R(0) $ and $ \alpha \in \cI_{k}^+ $, we now define the \emph{particle--hole pair creation operator}
\begin{equation}
	b^*_\alpha(k)
	:= \frac{1}{n_{\alpha, k}} \sum_{p, h: \alpha} \delta_{p, h+k} a_p^* a_h^*
	= \frac{1}{n_{\alpha, k}} \sum_{\substack{p: p \in B_{\F}^c \cap B_\alpha \\ p-k \in B_{\F} \cap B_\alpha}} a_p^* a_{p-k}^*
\label{eq:bstar}
\end{equation}
with normalization constant $ n_{\alpha, k} $ defined by
\begin{equation}
	n_{\alpha, k}^2 := \sum_{p, h : \alpha} \delta_{p, h+k} = \sum_{\substack{p: p \in B_{\F}^c \cap B_\alpha \\ p-k \in B_{\F} \cap B_\alpha}} 1	\;.
\label{eq:n}
\end{equation}
Observe that if $\alpha \not \in \cI_k^+$, then $b^*_\alpha(k)$ will usually be an empty sum, in which case we understand it as the zero operator. (Strictly speaking near the equator the sums may still contain a small number of summands. The index set $\cI_k^+$ is defined such that for $\alpha \in \cI_k^+$ the sum \cref{eq:n} contains a number of summands that grows sufficiently fast as $N \to \infty$; see \cref{lem:fn}.) Therefore we introduce the ``northern'' half-space
\[ 
	H^{\nor} := \big\{ k \in \RRR^3 : k_3 > 0 \text{ or } (k_3 = 0 \text{ and } k_2 > 0)
	\text{ or } (k_3 = k_2 = 0 \text{ and } k_1 > 0) \big\}\;,
\]
as well as the half-ball
\begin{equation}
\label{eq:Gammanor}
	\Gamma^{\nor} := H^{\nor} \cap \ZZZ^3 \cap B_R(0) \;,
\end{equation}
and then, for $ k \in \Gamma^{\nor} $, define
\begin{equation}
\label{eq:cstar}
	c^*_\alpha(k) := \begin{cases}
		b^*_\alpha(k) \quad &\text{for } \alpha \in \cI_{k}^+\\
		b^*_\alpha(-k) \quad &\text{for } \alpha \in \cI_{k}^-	\;.
	\end{cases}
\end{equation}
In \cref{lem:approxCCR} we are going to show that these pair operators satisfy approximately the canonical commutation relations (CCR) of bosons
\[
[c_\alpha(k), c_\beta(\ell)] = 0\;, \qquad [c_\alpha(k), c^*_\beta(\ell)] \approx \delta_{\alpha, \beta} \delta_{k, \ell} \;.
\]

\paragraph{The almost-bosonic Bogoliubov transformation}
The almost-bosonic Bogoliubov transformation is the unitary operator on fermionic Fock space $ T: \cF \to \cF $ chosen such that it would diagonalize an effective quadratic Hamiltonian
\[
h_{\eff} = \sum_{k \in \Gamma^{\nor}} h_{\eff}(k)
\]
where
\[
\begin{aligned}
	h_{\eff}(k) = \! \sum_{\alpha, \beta \in \cI_{k}} \!\left( (D(k) + W(k))_{\alpha, \beta} c^*_\alpha(k) c_\beta(k) + \frac{1}{2} \widetilde{W}(k)_{\alpha, \beta} \big( c^*_\alpha(k) c^*_\beta(k) + c_\beta(k) c_\alpha(k) \big) \right),
\end{aligned}
\]
\emph{if} the $ c^* $ and $ c $ operators would exactly satisfy the CCR of bosons (see \cite[(1.47)]{BNPSS21}). Here, $ D(k) $, $W(k)$, and $\widetilde{W}(k)$ are symmetric matrices in $\RRR^{|\cI_{k}| \times |\cI_{k}|} $ given in block form
\[
	D(k) = \begin{pmatrix}
		d(k) & 0\\
		0 & d(k)\\
	\end{pmatrix}\;, \qquad
	W(k) = \begin{pmatrix}
		b(k) & 0\\
		0 & b(k)\\
	\end{pmatrix}\;, \qquad
	\widetilde{W}(k) = \begin{pmatrix}
		0 & b(k)\\
		b(k) & 0\\
	\end{pmatrix}\;,
\]
with the smaller matrices $ d(k)$ and $b(k)$ in $\RRR^{|\cI_{k}^+| \times |\cI_{k}^+|} $ given by
\begin{equation}
\label{eq:db}
	d(k) := \sum_{\alpha \in \cI_{k}^+} |\hat{k} \cdot \hat{\omega}_\alpha| \; |\alpha \rangle \langle \alpha |\;, \qquad
	b(k) := \sum_{\alpha, \beta \in \cI_{k}^+} \frac{\hat{V}_k}{2 \hbar \kappa N |k|} n_{\alpha, k} n_{\beta, k} \; |\alpha \rangle \langle \beta |\;,
\end{equation}
where $ | \alpha \rangle $ is the $ \alpha $--th canonical basis vector of $ \RRR^{|\cI_{k}^+|} $ (and $ \hat{k} := k/|k| $). We are going to define $T$ by the explicit formula which was given in \cite{BNPSS20,BNPSS21,BPSS22} as
\begin{equation}\label{eq:T}
	T := e^{-S}\, \qquad S := -\frac{1}{2} \sum_{k \in \Gamma^{\nor}} \sum_{\alpha, \beta \in \cI_{k}} K(k)_{\alpha, \beta} \big( c^*_\alpha(k) c^*_\beta(k) - \mathrm{h.c.} \big) \;.
\end{equation}
The operator $S$ is anti-selfadjoint (i.\,e., $ S^* = -S $) and the matrix $K(k) \in \RRR^{|\cI_{k}| \times |\cI_{k}|} $ is defined via
\begin{align*}
	K(k) &:= \log |S_1(k)^T|\;, \tagg{eq:K}\\
	S_1(k) &:= (D(k) + W(k) - \widetilde{W}(k))^{\frac{1}{2}} E(k)^{-\frac{1}{2}}\;, \tagg{eq:S1} \\
	E(k) &:= \left( (D(k) + W(k) - \widetilde{W}(k))^{\frac{1}{2}} (D(k) + W(k) + \widetilde{W}(k)) (D(k) + W(k) - \widetilde{W}(k))^{\frac{1}{2}} \right)^{\frac{1}{2}}\!.
\end{align*}
This concludes the construction of the trial state.

\subsection{Optimality of the Main Result}
The upper bound \cref{eq:main} in our main result is sharp if the momentum $ q $, satisfying \eqref{eq:cQepsilon}, is located in the interior of a patch. The precise statement is as follows, proven in \cref{sec:proofmain}.

\begin{proposition}[Optimality]
\label{prop:optimality}
Under the assumptions of Theorem \ref{thm:main}, whenever $ q $ is in the interior of a patch $ B_{\alpha_q} $ in the sense that
\begin{equation}
\label{eq:edgeofthepatch}
\begin{aligned}
	 \text{for } q \in B_{\F}^c && \text{we have } B_R(q) \cap B_{\F} &\subset B_{\alpha_q}\;,\\
		\text{for } q \in B_{\F} && \text{we have  } B_R(q) \cap B_{\F}^c &\subset B_{\alpha_q}\;,
\end{aligned}
\end{equation}
(this is represented in \cref{fig:edgeofthepatch}) the upper bound \cref{eq:main} becomes an equality:
\begin{equation}
\label{eq:optimality}
	n_q
	= N^{-\frac{2}{3}} \sum_{k \in \cC^q\cap \ZZZ^3}  \frac{\hat{V}_k}{2 \kappa |k|} \frac{1}{\pi} \int_0^\infty \frac{(\mu^2 - \lambda_{q,k}^2)(\mu^2 + \lambda_{q,k}^2)^{-2}}{1 + Q_k^{(0)}(\mu)} \; \d \mu + \cE\;,
\end{equation}
where $\mathcal{E}$ is bounded as in \cref{eq:mainerror}, and the set of allowed momentum transfers is
\begin{equation}
\label{eq:cCq}
	\cC^q := \begin{cases}
		B_R(0) \cap H^{\nor} \cap \left( (B_{\F} - q) \cup (\BF + q) \right) \quad &\text{for } q \in B_{\F}^c\\
		B_R(0) \cap H^{\nor} \cap \left( (B_{\F}^c - q) \cup (\BFc + q) \right) \quad &\text{for } q \in B_{\F} \;.
	\end{cases}
\end{equation}
\end{proposition}
The set $ \cC^q $ is obtained from $ \cD^q $ in \eqref{eq:abbreviations1} by reflecting all $ k $ in the lower to the upper half-space (i.\,e., $ k \mapsto -k $ whenever $ k \notin H^{\nor} $).

\section{Strategy of Proof of the Main Theorem}
\label{sec:strategyofproof}

The statement \cref{eq:mainenergy} in \cref{thm:main}, that $ \psi_N $ replicates the ground state energy, was proven in \cite{BNPSS20,BNPSS21,BPSS22}. So we need to establish \cref{eq:main} for the momentum distribution. Computing the expectation value of $n_q$ is non-trivial for two reasons: first, $n_q$ is a fermionic observable that does not have a bosonic representation in terms of the $c^*_\alpha(k)$- and $c_\alpha(k)$-operators; and second, because $n_q$ is typically of order $N^{-2/3}$, i.\,e., extremely small, it necessitates very precise bounds on the errors. To obtain these sharp error bounds, we employ a bootstrap which is novel to the bosonization context. This should be compared to the earlier works on bosonization, where error terms could be controlled by integrated quantities like the expectation value of $\Ncal$ (so the sum of $n_q$ over all momenta $q$), which was estimated much less precisely as being overall of order one.

\paragraph{The Bootstrap} We aim to show that $ n_q $ can be approximated by a bosonized $ n_q^{(\b)} $. Whenever $ q $ is not inside some patch $ B_{\alpha_q} $, $n_q$ vanishes, so we will set $ n_q^{(\b)} = 0 $ in that case. For all other $ q $, we will see that $ n_q $ amounts to a sum over contributions depending on the momentum exchange $ k $.
As explained in \cref{lem:g}, define
\begin{equation}
\label{eq:cCqtilde}
\begin{aligned}
	\tilde{\cC}^q := \begin{cases}
		B_R(0)
		\cap H^{\nor}
		\cap \left( \big( (B_{\F} \cap B_{\alpha_q}) - q \big) \cup \big( - (B_{\F} \cap B_{\alpha_q}) + q \big) \right)
		&\quad \text{for } q \in B_{\F}^c \\
		\quad \cap \{k : |k \cdot \hat{\omega}_{\alpha_q}| \ge N^{-\delta}\} \\
		B_R(0)
		\cap H^{\nor}
		\cap \left( \big( (B_{\F}^c \cap B_{\alpha_q}) - q \big) \cup \big( - (B_{\F}^c \cap B_{\alpha_q}) + q \big) \right)
		&\quad \text{for } q \in B_{\F} \\
		\quad \cap \{k : |k \cdot \hat{\omega}_{\alpha_q}| \ge N^{-\delta}\} \;.\\
	\end{cases}
\end{aligned}
\end{equation}
Note that $ \tilde{\cC}^q $ agrees with $ \cC^q $, defined in \eqref{eq:cCq}, up to the exclusion of $k$ such that $ |k \cdot \hat{\omega}_{\alpha_q}| < N^{-\delta} $ and the restriction to $ B_{\alpha_q} $.
In the exactly bosonic approximation the momentum distribution can be computed explicitly to be
\begin{equation}
\label{eq:nqb}
	n_q^{(\b)}
	:= \frac{1}{2} \sum_{k \in \tilde{\cC}^q \cap \ZZZ^3} \frac{1}{n_{\alpha_q, k}^2} \big( \cosh(2 K(k)) - 1 \big)_{\alpha_q, \alpha_q}\;.
\end{equation}
The rigorous justification that $ n_q \approx n_q^{(\b)} $ is given by the next theorem.

\begin{theorem}[Bosonized Momentum Distribution]
\label{thm:main2}
Assume that $\hat{V}$ is non-negative and compactly supported. Then
\begin{equation}
\label{eq:main2}
	|n_q - n_q^{(\b)}|
	\le C N^{-1 + 2 \delta}\;.
\end{equation}
\end{theorem}

The proof of \cref{thm:main2} is given in \cref{sec:proofmain2}, but let us explain the interplay of bosonization and the bootstrap. We expand $ n_q $ in the trial state $\psi_N = \mathfrak{R} e^{-S}\Omega$ as
\begin{equation}
\label{eq:BCH}
\begin{aligned}
	n_q
	= \langle \Omega, e^S a_q^* a_q e^{-S} \Omega \rangle = \sum_{n = 0}^\infty \frac{1}{n!} \langle \Omega, \ad^n_S(a_q^* a_q) \Omega \rangle \;,
\end{aligned}
\end{equation}
where $ \ad^n_A(B) := [A, \ldots[A, [A,B]] \ldots] $ denotes the $n$--fold commutator, with $S$ as defined in \cref{eq:T}.
%
In the exactly bosonic approximation we assume that the operators $ c^*$, $c $ satisfy canonical commutation relations.
Then for $ n \ge 1 $, the form of $ \ad^n_{q, (\b)} $ alternates between $\ad^n_{q, (\b)} \sim c^* c^* - c c$ for $n \text{ odd}$ and $\ad^n_{q, (\b)} \sim c^* c + \textnormal{const}$ for $n \text{ even}$. The only non-vanishing contribution to the vacuum expectation value is for even $n$ due to the const-terms, which sum up to a $ \cosh $--series and render $ n_q^{(\b)} $ in \cref{eq:nqb}. The difficulty in the proof of \cref{thm:main2} is proving that the deviation from exactly bosonic canonical commutator relations (the term $\Ecal$ in \cref{eq:cgcommutator}) has small effect. Here the bootstrap is crucial: by Duhamel's formula, we express $ |n_q - n_q^{(\b)}| $ in terms of expectation values in the parametrized states
\begin{equation}
\label{eq:xit}
	\xi_t := e^{-t S} \Omega\;, \qquad t \in [-1, 1]\;.
\end{equation}
Employing the lemmas of \cref{sec:bosonizationerror,sec:bosonizedterms} we bound these expectation values using $ \langle \xi_t, a_q^* a_q \xi_t \rangle $. The bootstrap is then based on \cref{lem:bootstrap}. Initially we know that $ 0 \le \langle \xi_t, a_q^* a_q \xi_t \rangle \le 1 $, which we write with $r :=0$ as
\begin{equation}
\label{eq:bootstrapbound1}
	\langle \xi_t, a_q^* a_q \xi_t \rangle
	= \cO(N^{-r}) \;.
\end{equation}
Using this bound within Duhamel's formula, Lemma \ref{lem:bootstrap} provides us with
\begin{equation}
\label{eq:bootstrapbound2}
	|n_q - n_q^{(\b)}|
	= |\langle \xi_1, a_q^* a_q \xi_1 \rangle - n_q^{(\b)}|
	= \cO(N^{-r'})
\end{equation}
for $ r' = \frac 23 - \frac 32 \delta + \frac r2 $. The same holds if the trial state $ \mathfrak{R} e^{-S} \Omega $ is replaced with $ \mathfrak{R} e^{-t S} \Omega $. Together with the observation \cref{eq:nqbscaling} that $ n_q^{(\b)} = \cO(N^{- \frac{2}{3}+ \delta}) $ (and the same if the trial state is replaced by its $t$--dependent version), we obtain that \eqref{eq:bootstrapbound1} is valid for $ r = \frac 23 - \frac 32 \delta $. Plugging this again into \eqref{eq:bootstrapbound2} yields $ r = \frac 23 - \delta $, which is the optimal exponent. Then $ r' $ is the claimed error exponent from \cref{thm:main2}.

\paragraph{Post-Processing} We evaluate the hyperbolic cosine of the matrix $K(k)$ in \cref{eq:nqb} by functional calculus, which brings us close to the final form \cref{eq:main}. The computation is given in \cref{sec:proofmain3}, the result is the next proposition.

\begin{proposition}
\label{prop:main3}
If $ q \in B_{\alpha_q} $ for some $ 1 \le \alpha_q \le M $, then
\begin{equation}
\label{eq:main3}
	n_q^{(\b)}
	= \sum_{k \in \tilde{\cC}^q\cap \ZZZ^3} \frac{1}{\pi} \frac{\hat{V}_k}{2 \hbar \kappa N |k|} \int_0^\infty \frac{(\mu^2 - \lambda_{\alpha_q,k}^2)(\mu^2 + \lambda_{\alpha_q,k}^2)^{-2}}{1 + Q_k(\mu)} \; \d \mu
\end{equation}
with
\begin{equation}
\label{eq:abbreviations2}
	\lambda_{\alpha,k} := |\hat{k} \cdot \hat{\omega}_\alpha|\;, \qquad
	Q_k(\mu) := \frac{\hat{V}_k}{\hbar \kappa N |k|} \sum_{\alpha \in \cI_k^+} n_{\alpha, k}^2 (\mu^2 + \lambda_{\alpha,k}^2)^{-1} \lambda_{\alpha,k}\;.
\end{equation}
\end{proposition}
Finally we approximate the sum within $ Q_k(\mu) $ by a surface integral over the half--sphere, which renders $ Q_k(\mu) \approx Q^{(0)}_k(\mu) $. We arrive at a formula resembling \cref{eq:main}, but with $ \tilde{\cC}^q $ instead of $ \cC^q $. Eq.~\cref{eq:main} then follows by bounding the contributions from $ k \in (\tilde{\cC}^q \setminus \cC^q) \cap \ZZZ^3 $. These computations are done in \cref{sec:proofmain}.

\smallskip

\section{Generalized Pair Operators}
\label{subsec:generalizedpair}

Recall the definition \cref{eq:cstar} of the bosonized pair creation operator $ c^*_\alpha(k) $. It will be convenient to use more general pair operators, similar to the weighted pair operators from \cite[Lemma~5.3]{BPSS22}. For a function $ g: \BFc \times \BF \to \CCC $, we define
\begin{equation}
	c^*(g) := \sum_{\substack{p \in \BFc\\h \in \BF}} g(p, h) a_p^* a_h^*\;, \qquad
	c(g) := \sum_{\substack{p \in \BFc\\h \in \BF}} \overline{g(p, h)} a_h a_p\;.
\label{eq:cgstar}
\end{equation}
We may then identify
\begin{equation}
	c^*_\alpha(k) = c^*(d_{\alpha, k})\;, \quad \text{where} \quad
	d_{\alpha, k}(p, h) := \begin{cases}
		\displaystyle \delta_{p, h+k} \frac{1}{n_{\alpha, k}} \chi(p, h : \alpha) \quad &\text{if } \alpha \in \cI_{k}^+\\
		\displaystyle \delta_{p, h-k} \frac{1}{n_{\alpha, k}} \chi(p, h : \alpha) \quad &\text{if } \alpha \in \cI_{k}^-
	\end{cases}
\label{eq:cstarabbreviation}
\end{equation}
and
\[
 \chi(p, h : \alpha) := \chi_{B_{\F}^c \cap B_\alpha}(p) \chi_{B_{\F} \cap B_\alpha}(h) \;.
\]
In the following we adopt the shorthand notation
\begin{equation}
\label{eq:pmk}
	\pm k := \begin{cases}
		+k \quad \text{if } \alpha \in \cI_{k}^+\\
		-k \quad \text{if } \alpha \in \cI_{k}^-
	\end{cases}
	, \qquad
	\mp k := \begin{cases}
		-k \quad \text{if } \alpha \in \cI_{k}^+\\
		+k \quad \text{if } \alpha \in \cI_{k}^-
	\end{cases} .
\end{equation}
So
\begin{equation}
	d_{\alpha, k}(p, h) = \delta_{p, h \pm k} \frac{1}{n_{\alpha, k}} \chi(p, h : \alpha)\;.
\label{eq:d}
\end{equation}

The generalized pair operators satisfy the following commutation relations.
\begin{lemma}[Generalized approximate CCR]
\label{lem:cgcommutator}
Consider $ g, \tilde{g}: \BFc \times \BF \to \CCC $. Then
\begin{equation}
\label{eq:cgcommutator}
	\left[c(g), c(\tilde{g})\right] = \left[c^*(g), c^*(\tilde{g})\right] = 0 \;, \qquad
	\left[c(g), c^*(\tilde{g})\right] = \langle g, \tilde{g} \rangle + \cE(g, \tilde{g}) \;,
\end{equation}
where $ \langle \cdot, \cdot \rangle $ is the inner product on $ \ell^2(\BFc \times \BF) $, and where
\begin{equation}
\label{eq:cEgg}
	\cE(g, \tilde{g})
	:= - \sum_{\substack{p \in \BFc\\h_1, h_2 \in \BF}} \overline{g(p, h_1)} \tilde{g}(p, h_2) a_{h_2}^* a_{h_1}
	- \sum_{\substack{p_1, p_2 \in \BFc\\h \in \BF}} \overline{g(p_1, h)} \tilde{g}(p_2, h) a_{p_2}^* a_{p_1}\;.
\end{equation}
\end{lemma}
\begin{proof}
Direct computation using the CAR \cref{eq:CAR}.
\end{proof}

The generalized pair operators arise from the first commutator $ \ad^1_S(a_q^* a_q) = [S, a_q^* a_q] $.
\begin{lemma}[Occupation number of a single mode in $c^*_\alpha(k)$]
\label{lem:g}
Let $ k \in \Gamma^{\nor} $, $ \alpha \in \cI_{k} $ and $ q \in B_{\alpha_q} $ for some $ 1 \le \alpha_q \le M $. Then, we have
\[ \left[c^*_\alpha(k), a_q^* a_q\right] = \begin{cases}
	0 					\quad &\text{if } \alpha \neq \alpha_q\\
	- c^*(g_{q, k}) 	\quad &\text{if } \alpha = \alpha_q\\
	\end{cases}
\]
with
\begin{equation}
\label{eq:g}
	g_{q, k}(p, h) := \delta_{p, h \pm k} \frac{1}{n_{\alpha_q, k}} \chi(p, h : \alpha_q) (\delta_{h, q} + \delta_{p, q})\;.
\end{equation}
In particular, $ \left[ c^*_\alpha(k), a_q^* a_q \right] = 0 $ whenever $ \alpha_q \notin \cI_{k} $.
\end{lemma}
\begin{proof}
Direct computation using the CAR \cref{eq:CAR}.
\end{proof}

\Cref{lem:g} motivates the definition of $ \tilde{\cC}^q $ in \cref{eq:cCqtilde}, chosen such that $ [c^*_\alpha(k), a_q^* a_q] $ does not vanish:
\begin{itemize}
\item For $ k \notin \Gamma^{\nor} = H^{\nor} \cap \ZZZ^3 \cap B_R(0) $, this commutator would not even be defined.
\item The condition $ |k \cdot \hat{\omega}_{\alpha_q}| \ge N^{-\delta} $ ensures $ \alpha_q \in \cI_{k} $ as otherwise $ [c^*_\alpha(k), a_q^* a_q] = 0 $.
\item For $ q \in \BF $, the condition $ k \in \big( (B_{\F} \cap B_{\alpha_q}) - q \big) \cup \big( - (B_{\F} \cap B_{\alpha_q}) + q \big) $ guarantees that the factor $ \chi(p, h : \alpha_q) $ in $ g_{q, k}(p, h) $ does not vanish. Analogously, for $ q \in \BF^c $, the condition $ k \in \big( (B_{\F}^c \cap B_{\alpha_q}) - q \big) \cup \big( - (B_{\F}^c \cap B_{\alpha_q}) + q \big) $ guarantees that $ \chi(p, h : \alpha_q) $ does not vanish.
\end{itemize}

The simplest case of \cref{lem:cgcommutator} are the approximate CCR from \cite[Lemma~4.1]{BNPSS20}.
\begin{lemma}[Approximate CCR]
\label{lem:approxCCR}
Let $ k, \ell \in \Gamma^{\nor} $ and $ \alpha \in \cI_{k}, \beta \in \cI_{\ell} $. Then we have
\begin{equation}
\label{eq:approxCCR}
	\left[ c_\alpha(k), c^*_\beta(\ell) \right] = \begin{cases}
		0 \quad &\text{if } \alpha \neq \beta\\
		\delta_{k, \ell} + \cE_\alpha(k, \ell) \quad &\text{if } \alpha = \beta
	\end{cases}
\end{equation}
with the deviation operator $ \cE_\alpha(k, \ell)^* = \cE_\alpha(\ell, k) $ explicitly given by
\begin{equation}
\label{eq:cEalphaql}
	\cE_\alpha(k, \ell) :=
	- \sum_{p, h_1, h_2 : \alpha} \frac{\delta_{h_1, p \mp k} \delta_{h_2, p \mp \ell}}{n_{\alpha, k} n_{\alpha, \ell}} a_{h_2}^* a_{h_1}
	- \sum_{p_1, p_2, h : \alpha} \frac{\delta_{h, p_1 \mp k} \delta_{h, p_2 \mp \ell}}{n_{\alpha, k} n_{\alpha, \ell}} a_{p_2}^* a_{p_1}\;.
\end{equation}
\end{lemma}
\begin{proof}
Follows from \cref{eq:cstarabbreviation} with $ g = d_{\alpha, k} $ and $ \tilde{g} = d_{\beta, \ell} $.
\end{proof}

To compute $\ad^{n+1}_S(a_q^* a_q) = [S, \ad^{n}_S(a_q^* a_q)]$ the following commutator is useful.
\begin{lemma}
\label{lem:dg}
Let $k,\ell \in \Gamma^{\nor}$, $\alpha \in \cI_k$, $ \alpha_q \in \cI_{\ell} $, $ q \in B_{\alpha_q} $ and $g_{q,\ell}$ as in \cref{eq:g}. Then
\begin{equation}
	\left[ c_\alpha(k), c^*(g_{q, \ell}) \right] = \begin{cases}
		0 \quad &\text{if } \alpha \neq \alpha_q\\
		\delta_{k, \ell} \rho_{q, k} + \cE^{(g)}_q(k, \ell) \quad &\text{if } \alpha = \alpha_q
	\end{cases}
\end{equation}
with
\begin{equation}	\label{eq:dg}
	\rho_{q, k} := \begin{cases}
		{n_{\alpha_q, k}^{-2}}\, \chi(q \mp k \in B_{\F} \cap B_{\alpha_q}) \quad &\text{if } q \in B_{\F}^c\\
		{n_{\alpha_q, k}^{-2}}\, \chi(q \pm k \in B_{\F}^c \cap B_{\alpha_q}) \quad &\text{if } q \in B_{\F}
	\end{cases}
\end{equation}
and with the operator $ \cE^{(g)}_q(k, k)^* = \cE^{(g)}_q(k, k) $ explicitly given as
\begin{equation}
\label{eq:cEg}
	\cE^{(g)}_q(k, \ell) := \begin{cases}\displaystyle
		& \hspace{-1em}- \frac{1}{n_{\alpha_q,k} n_{\alpha_q,\ell}}
		\left( \chi \Big( \substack{q \mp k \in B_{\F} \cap B_{\alpha_q} \\ q \mp \ell \in B_{\F} \cap B_{\alpha_q} } \Big) a_{q \mp \ell}^* a_{q \mp k}
		- \chi \Big( \substack{q \mp \ell \in B_{\F} \cap B_{\alpha_q} \\ q \mp \ell \pm k \in B_{\F}^c \cap B_{\alpha_q} } \Big) a_q^* a_{q \mp \ell \pm k} \right) \\ & \hspace{24em}\text{if } q \in B_{\F}^c\,,\\
		\displaystyle
		& \hspace{-1em} - \frac{1}{n_{\alpha_q,k} n_{\alpha_q,\ell}}
		\left( \chi \Big( \substack{q \mp \ell \in B_{\F}^c \cap B_{\alpha_q} \\ q \pm \ell \mp k \in B_{\F} \cap B_{\alpha_q} } \Big) a_q^* a_{q \pm \ell \mp k}
		- \chi \Big( \substack{q \pm \ell \in B_{\F}^c \cap B_{\alpha_q} \\ q \pm k \in B_{\F}^c \cap B_{\alpha_q} } \Big) a_{q \pm \ell}^* a_{q \pm k} \right) \\ & \hspace{24em}\text{if } q \in B_{\F} \,.
	\end{cases}
\end{equation}
\end{lemma}
By $\chi \Big( \substack{q \mp k \in B_{\F} \cap B_{\alpha_q} \\ q \mp \ell \in B_{\F} \cap B_{\alpha_q} } \Big)$ we denote the characteristic function of the set of all $q \in \Zbb^3$ satisfying \emph{both} $q \mp k \in B_{\F} \cap B_{\alpha_q}$ \emph{and} $q \mp \ell \in B_{\F} \cap B_{\alpha_q} $, with the sign of $\pm k$ and $\mp k$ as defined in \cref{eq:pmk}.

\section{Momentum Distribution from Bosonization}
\label{subsec:bosonizationapprox}

Recall that the exact momentum distribution, according to \cref{eq:BCH}, is given by
\begin{equation}
\label{eq:nqbarformula}
	n_q = \sum_{n = 0}^\infty \frac{1}{n!} \langle \Omega, \ad^n_S(a_q^* a_q) \Omega \rangle\;.
\end{equation}
The evaluation of multi-commutators $ \ad^n_S (a_q^* a_q) $ results in a rather involved expression. However, the dominant contribution is obtained pretending that bosonization was exact, i.\,e., if we drop $ \cE_\alpha(k, \ell) $ in the approximate CCR \cref{eq:approxCCR}.

\smallskip

For $ n = 0 $, we choose $ \ad^0_{q, (\b)} = a_q^* a_q = \ad^0_S (a_q^* a_q) $, so bosonization is exact.

For $ n \ge 1 $, the bosonized multi-commutator $ \ad^n_{q, (\b)} $ is expressed using six types of terms ($ \bA$, $\bB$, $\bC$, $\bD$, $\bE $, and $ \bF $). If $ q \in B_{\alpha_q} $ for some $ 1 \le \alpha_q \le M $, we define
\begin{equation}
\label{eq:adnb}
	\ad^n_{q, (\b)} := \begin{cases}
		\displaystyle 2^{n-1} \bA_n + \bB_n + \bB_n^* + \sum_{m = 1}^{n-1} \binom{n}{m} \bC_{n-m, m} \quad & \text{if } n \text{ is even}\\
		\displaystyle \bE_n + \bE_n^* + \sum_{m = 1}^s \binom{n}{m} \bD_{n-m, m} + \sum_{m = 1}^s \binom{n}{m} \bF_{m, n-m} \quad &\substack{\displaystyle \text{if } n \text{ is odd},\\ \displaystyle
		n = 2s + 1.}
	\end{cases}
\end{equation}
If $ q $ is not inside any patch we set $ \ad^n_{q, (\b)} := 0 $. With $\rho_{q,k}$ from \cref{eq:dg} the terms are
\begin{equation}
\label{eq:ABC}
\begin{aligned}
	\bA_n &:= \sum_{k \in \tilde{\cC}^q \cap \ZZZ^3}
		\big( K(k)^n \big)_{\alpha_q, \alpha_q} \rho_{q, k}\qquad \in \mathbb{C}\;,\\
	\bB_n &:= \sum_{k \in \tilde{\cC}^q \cap \ZZZ^3} \sum_{\alpha_1 \in \Ik}
		\big( K(k)^n \big)_{\alpha_q, \alpha_1} \;
		c^*(g_{q, k}) c_{\alpha_1}(k)\;,\\
	\bC_{m, m'} &:= \sum_{k \in \tilde{\cC}^q \cap \ZZZ^3} \sum_{\alpha_1, \alpha_2 \in \Ik}
		\big( K(k)^m \big)_{\alpha_q, \alpha_1} \big( K(k)^{m'} \big)_{\alpha_q, \alpha_2} \rho_{q, k} \;
		c^*_{\alpha_1}(k) c_{\alpha_2}(k)\;,
\end{aligned}
\end{equation}
and
\begin{equation}
\label{eq:DEF}
\begin{aligned}
	\bD_{m, m'} &:= \sum_{k \in \tilde{\cC}^q \cap \ZZZ^3} \sum_{\alpha_1, \alpha_2 \in \Ik}
		\big( K(k)^m \big)_{\alpha_q, \alpha_1} \big( K(k)^{m'} \big)_{\alpha_q, \alpha_2} \rho_{q, k} \;
		c_{\alpha_1}(k) c_{\alpha_2}(k)\;,\\
	\bE_n &:= \sum_{k \in \tilde{\cC}^q \cap \ZZZ^3} \sum_{\alpha_1 \in \Ik}
		\big( K(k)^n \big)_{\alpha_q, \alpha_1}
		c^*(g_{q, k}) c^*_{\alpha_1}(k)\;,\\
	\bF_{m, m'} &:= \sum_{k \in \tilde{\cC}^q \cap \ZZZ^3} \sum_{\alpha_1, \alpha_2 \in \Ik}
		\big( K(k)^m \big)_{\alpha_q, \alpha_1} \big( K(k)^{m'} \big)_{\alpha_q, \alpha_2} \rho_{q, k} \;
		c^*_{\alpha_1}(k) c^*_{\alpha_2}(k)\;.
\end{aligned}
\end{equation}
Since $\rho_{q,k}$ is real, and since $K(k)$ is a real and symmetric matrix, we have
\begin{equation}
\label{eq:ABCidentities}
\bA_n^* = \bA_n\;, \qquad
	\bC_{m, m'}^* = \bC_{m', m} \;,
\end{equation}
\begin{equation}
\label{eq:DEFidentities}
\bD_{m, m'} = \bD_{m', m}\;, \qquad
	\bF_{m, m'} = \bF_{m', m}\;, \qquad
	\bD_{m, m'}^* = \bF_{m', m}\;.
\end{equation}
Replacing $ \ad^n_S (a^*_q a_q) $ by $ \ad^n_{q, (\b)} $ in  \cref{eq:nqbarformula} yields the bosonization approximation $ n_q^{(\b)} $:  If $ q \in B_{\alpha_q} $ for some $ 1 \le \alpha_q \le M $, then\footnote{Note that $ k \in \tilde{\cC}^q \cap \ZZZ^3 $ enforces $ \alpha_q \in \cI_{k} $, so the denominator $ n_{\alpha_q, k}^2 $ does not vanish.}
\begin{align*}
	&\sum_{n = 0}^\infty \frac{1}{n!} \langle \Omega, \ad^n_{q, (\b)} \Omega \rangle
	= \sum_{m = 1}^\infty \frac{2^{2m-1}}{(2m)!} \langle \Omega, \bA_{2m} \Omega \rangle
	= \sum_{m = 1}^\infty \frac{2^{2m-1}}{(2m)!} \sum_{k \in \tilde{\cC}^q \cap \ZZZ^3} \big( K(k)^{2m} \big)_{\alpha_q, \alpha_q} \rho_{q, k}\\
	& = \frac{1}{2} \sum_{k \in \tilde{\cC}^q \cap \ZZZ^3} \frac{1}{n_{\alpha_q, k}^2} \big( \cosh(2 K(k)) - 1 \big)_{\alpha_q, \alpha_q} = n_q^{(\b)}\;.		\tagg{eq:nqbformula}
\end{align*}
Otherwise, $ \ad^n_{q, (\b)} = n_q^{(\b)} = 0 $. Both agrees with \cref{eq:nqb}.

\section{Controlling the Bosonization Error}
\label{sec:bosonizationerror}

In this section we compile the basic estimates required to control the bosonization.
\begin{lemma}[Bound on Powers of $ K $]
\label{lem:Knbound}
Suppose that $\hat{V}$ is non-negative. Then there is $ C > 0 $ such that for all $ k \in \Gamma^{\nor}, \alpha, \beta \in \cI_{k} $ and $ n \in \NNN $, we have
\begin{equation}
\label{eq:Knbound}
	\lvert (K(k)^n)_{\alpha, \beta}\rvert \le (C \hat{V}_k)^n M^{-1}\;.
\end{equation}
\end{lemma}
\begin{proof}
This follows using \cite[Lemma~7.1]{BPSS22}, which states that
$|K(k)_{\alpha, \beta}| \le C \hat{V}_k M^{-1}$.
\end{proof}
%
\begin{lemma}[Bounds on $ n_{\alpha,k} $]
\label{lem:fn}
For all $ k \in \Gamma^{\nor} $ and $ \alpha \in \cI_{k} $ we have
\begin{equation}
\label{eq:fn}
	n_{\alpha, k} \ge C \fn \qquad \text{for} \qquad \fn := N^{\frac{1}{3}-\frac{\delta}{2}} M^{-\frac{1}{2}}\;.
\end{equation}
\end{lemma}
\begin{proof}
This is just \cite[Eq.~(3.18)]{BNPSS20}. (In \cite{BNPSS20} it is assumed that $ M = N^{\frac{1}{3} + \varepsilon} $ with $\varepsilon > 0$. We only assume $ M \gg N^{2 \delta} $ but the proof remains true.)
\end{proof}
By \cref{eq:Mdelta} we conclude that $\fn \to \infty$ as $N\to \infty$ at least as fast as $ \fn \ge C N^{\frac{\delta}{2}} $.

\smallskip

Next, we compile estimates on the deviation operators $ \cE_\alpha(k, \ell) $ defined in \cref{eq:cEalphaql} and $ \cE^{(g)}_q(k, \ell) $ defined in \cref{eq:cEg}, which are partly based on \cite[Lemma~4.1]{BNPSS20}.

\begin{lemma}[Bounds on $ \cE_\alpha $ and $ \cE^{(g)} $]
\label{lem:cEbound}
Let $ k, \ell \in \Gamma^{\nor} $, $ \alpha, \alpha_q \in \cI_{k} \cap \cI_{\ell} $, and $ q \in B_{\alpha_q} $. Then for all $\psi \in \fock$ we have, for all choices\footnote{The meaning of $\sharp$ as ``adjoint'' or ``non adjoint'' may vary between every appearance of the symbol, even within the same formula; we mean that the statement holds for all possible combinations.} of $ \sharp \in \{ \cdot, * \} $,
\begin{equation}
\label{eq:cEbound}
	\Vert \cE_\alpha(k, \ell)^\sharp \psi \Vert
	\le \frac{2}{n_{\alpha, k} n_{\alpha, \ell}} \Vert \cN \psi \Vert\;, \qquad
	\Vert \cE^{(g)}_q(k, \ell)^\sharp \psi \Vert
	\le \frac{2}{n_{\alpha_q, k} n_{\alpha_q, \ell}} \Vert \psi \Vert\;,
\end{equation}
and
\begin{equation}
\label{eq:cEsumbound}
	\sum_{\beta \in \cI_{k} \cap \cI_{\ell}} \Vert \cE_\beta(k, \ell)^\sharp \psi \Vert^2
	\le \frac{C}{\fn^2} \Vert \cN^{\frac{1}{2}} \psi \Vert^2 \;.
\end{equation}
\end{lemma}
Note that $\cE^{(g)}_q(k, \ell)^\sharp$ satisfies a sharper bound than $\cE_\alpha(k, \ell)^\sharp$, not requiring the number operator $\Ncal$ on the r.\,h.\,s.\ because \cref{eq:cEg} does not contain any sum.
\begin{proof}
The first bound in \cref{eq:cEbound} was already given in \cite[Lemma~4.1]{BNPSS20} for $ \cE_\alpha(k, \ell) $. The statement then follows for $ \cE_\alpha(k, \ell)^* = \cE_\alpha(\ell, k) $.
For the bound on $\cE^{(g)}_q(k, \ell)^\sharp$ in \cref{eq:cEbound}, recall the definition  \cref{eq:cEg} of $ \cE^{(g)}_q(k, \ell) $
 and use the operator norm bound
 $ \Vert a_{q'}^\sharp \Vert_{\op} \le 1 $.
It remains to establish \cref{eq:cEsumbound}. Eq.~\cref{eq:cEalphaql} renders
\begin{align*}
	\sum_{\beta \in \Ik \cap \cI_{\ell}} \Vert \cE_\beta(k, \ell) \psi \Vert^2 \tagg{eq:cEbetaexpansion}
	\leq \sum_{\beta \in \Ik \cap \cI_{\ell}}  \frac{2}{n_{\beta, k}^2 n_{\beta, \ell}^2} \Bigg( & \Big\Vert
		\sum_{p : \beta}
		\chi \Big( \substack{p \mp k \in B_{\F} \cap B_{\beta} \\ p \mp \ell \in B_{\F} \cap B_{\beta} } \Big)
		a_{p \mp \ell}^* a_{p \mp k} \psi \Big\Vert^2 \\
		& +
		\Big\Vert \sum_{h : \beta}
		\chi \Big( \substack{h \pm k \in B_{\F}^c \cap B_{\beta} \\ h \pm \ell \in B_{\F}^c \cap B_{\beta} } \Big)
		a_{h \pm \ell}^* a_{h \pm k} \psi \Big\Vert^2 \Bigg)\;.
\end{align*}
Let us introduce the set $ \cS_\beta := \{ p : \beta \; \vert \; p \mp k \in B_{\F} \cap B_{\beta} \text{ and } p \mp \ell \in B_{\F} \cap B_{\beta} \} $. The first term in \cref{eq:cEbetaexpansion} then becomes
\[
\begin{aligned}
	&\sum_{\beta \in \Ik \cap \cI_{\ell}}  \frac{2}{n_{\beta, k}^2 n_{\beta, \ell}^2}
		\Big\Vert \sum_{p \in \cS_\beta} a_{p \mp \ell}^* a_{p \mp k} \psi \Big\Vert^2
	\le \sum_{\beta \in \Ik \cap \cI_{\ell}}  \frac{2}{n_{\beta, k}^2 n_{\beta, \ell}^2}
		\Big( \sum_{p \in \cS_\beta} \Vert a_{p \mp \ell}^* a_{p \mp k} \psi \Vert \Big)^2\\
	& \le \sum_{\beta \in \Ik \cap \cI_{\ell}}  \frac{2}{n_{\beta, k}^2 n_{\beta, \ell}^2}
		\Big( \sum_{p \in \cS_\beta} 1 \Big)
		\Big( \sum_{p \in \cS_\beta} \Vert a_{p \mp \ell}^* a_{p \mp k} \psi \Vert^2 \Big)
	\le \sum_{\beta \in \Ik \cap \cI_{\ell}}  \frac{2}{n_{\beta, k} n_{\beta, \ell}}
		\sum_{p \in \cS_\beta} \Vert a_{p \mp \ell}^* a_{p \mp k} \psi \Vert^2\\
	& \le \frac{C}{\fn^2} \sum_{\beta \in \Ik \cap \cI_{\ell}}
		\sum_{p \in \cS_\beta} \Vert a_{p \mp k} \psi \Vert^2
	\le \frac{C}{\fn^2} \Vert \cN^{\frac{1}{2}} \psi \Vert^2\;,
\end{aligned}
\]
where we used the Cauchy--Schwartz inequality, $ \sum_{p \in \cS_\beta} 1 \le \min \{ n_{\beta, k}^2, n_{\beta, \ell}^2 \} \le n_{\beta, k} n_{\beta, \ell} $ and $ \Vert a_h^\sharp \Vert_{\op} \le 1 $. The second term on the r.\,h.\,s.\ of \cref{eq:cEbetaexpansion} is bounded analogously and the proof for $ \cE_\beta(k, \ell)^* = \cE_\beta(\ell, k) $ works the same way.
\end{proof}

In the trial state $ \xi_t \in \cF $ from \cref{eq:xit}, the estimate \cref{eq:cEbound} for $ \cE^{(g)}_q(k, \ell) $ is far from optimal: $ \Vert a_{q'}^\sharp \Vert_{\op} \le 1 $ means that we bound the $q$--mode as if it was fully occupied. This yields the initial bootstrap bound $ \langle \xi_t, a_q^* a_q \xi_t \rangle = \cO(N^{-r}) $ with $r = 0$. It turns out below that we can improve the exponent up to $r = \frac{2}{3} -\delta$.

\begin{lemma}[Bootstrap Bounds on $ \cE^{(g)} $]
\label{lem:cEgbound}
Let $ k, \ell \in \Gamma^{\nor} $, $ \alpha_q \in \cI_{k} \cap \cI_{\ell} $, and $ q \in B_{\alpha_q} $. Assume that there is $ r \ge 0 $ such that for all $ t \in [-1, 1] $ and all $ q' \in \ZZZ^3 $ it is known that (with $ C $ independent of $ t, q' $)
\begin{equation}
\label{eq:nqcondition}
	\langle \xi_t, a_{q'}^* a_{q'} \xi_t \rangle
	\le C N^{-r}\;.
\end{equation}
Then (with $ \sharp \in \{ \cdot, * \} $) for all $ t \in [-1, 1] $ we have
\begin{equation}
\label{eq:cEgxitbound}
	\lVert \cE^{(g)}_q(k, \ell)^\sharp \xi_t \rVert
	\leq \frac{C}{n_{\alpha_q, k} n_{\alpha_q, \ell}} N^{-\frac{r}{2}} \;.
\end{equation}
\end{lemma}
\begin{proof}
From \cref{eq:cEg} and using $ \Vert a_{q'}^* \Vert_{\op} = 1 $ we have
\begin{equation}
\begin{aligned}
\label{eq:cEgdecomposition2}
	\Vert \cE^{(g)}_q(k, \ell) \xi_t \Vert
	& \le \frac{\Vert a_{q \mp \ell}^* a_{q \mp k} \xi_t \Vert + \Vert a_q^* a_{q \mp \ell \pm k} \xi_t \Vert}{n_{\alpha_q, k} n_{\alpha_q, \ell}} \le \frac{\Vert a_{q \mp k} \xi_t \Vert + \Vert a_{q \mp \ell \pm k} \xi_t \Vert}{n_{\alpha_q, k} n_{\alpha_q, \ell}} \;.
\end{aligned}
\end{equation}
By the bootstrap assumption, $ \Vert a_{q'} \xi_t \Vert = \langle \xi_t, a_{q'}^* a_{q'} \xi_t \rangle^{\frac 12} \le C N^{-\frac r2} $, which renders the desired bound for $ \Vert \cE^{(g)}_q(k, \ell) \xi_t \Vert $. Analogously one estimates $ \Vert \cE^{(g)}_q(k, \ell)^* \xi_t \Vert $.
\end{proof}

We also need to bound the $ c^*$- and $c$-operators.

\begin{lemma}[Bounds for $ c^*, c $]
\label{lem:ccbound}
Let $ k \in \Gamma^{\nor} $ and consider a family of bounded functions $ (g^{(\alpha)})_{\alpha \in \cI_{k}}$ with $g^{(\alpha)}: \BFc \times \BF \to \RRR $ such that
\[\supp(g^{(\alpha)}) \subseteq \{ (p, h : \alpha): p = h \pm k \}\;.\]
Then for all $ f \in \ell^2(\cI_{k}) $ we have
\begin{equation}
\label{eq:ccboundgeneral}
\begin{aligned}
	\Big\Vert \sum_{\alpha \in \cI_{k}} f_\alpha c^*(g^{(\alpha)}) \psi \Big\Vert
	& \le \Vert f \Vert_2  \max_{\alpha \in \cI_{k}} \left(n_{\alpha, k} \lVert  g^{(\alpha)} \rVert_\infty \right) \left\lVert (\cN + 1)^{\frac{1}{2}} \psi \right\rVert\;,\\
	\Big\Vert \sum_{\alpha \in \cI_{k}} f_\alpha c(g^{(\alpha)}) \psi \Big\Vert
	& \le \Vert f \Vert_2  \max_{\alpha \in \cI_{k}} \left( n_{\alpha, k} \lVert  g^{(\alpha)} \rVert_\infty \right) \lVert \cN^{\frac{1}{2}} \psi \rVert\;.
\end{aligned}
\end{equation}
In particular
\begin{equation}
\label{eq:ccbound}
	\Big\Vert \sum_{\alpha \in \cI_{k}} f_\alpha c^*_\alpha(k) \psi \Big\Vert
	\le \Vert f \Vert_2 \Vert (\cN + 1)^{\frac{1}{2}} \psi \Vert\;, \qquad
	\Big\Vert \sum_{\alpha \in \cI_{k}} f_\alpha c_\alpha(k) \psi \Big\Vert
	\le \Vert f \Vert_2 \Vert \cN^{\frac{1}{2}} \psi \Vert\;.
\end{equation}
\end{lemma}
\begin{proof}
For \cref{eq:ccboundgeneral} see \cite[Lemma~5.4]{BNPSS21} (note that in \cite{BNPSS21}, a factor of $ n_{\alpha, k}^{-1} $ is included in the definition of $ c^\sharp(g) $ instead of the weight function $ g $). The bounds \cref{eq:ccbound} follow setting $ g^{(\alpha)} = d_{\alpha, k} $, defined in \cref{eq:d}.
\end{proof}

Analogous bounds could be derived for the operators $ c^\sharp(g_{q, k}) $, which just differ from $ c^\sharp_{\alpha_q}(k) $ by an additional factor of $ (\delta_{q, p} + \delta_{q, h}) $, see \cref{eq:g}. However, these Kronecker deltas significantly reduce the number of summands, from $ \sim n_{\alpha_q, k} $ to $ \sim 1 $. Accordingly the following lemma provides sharper bounds. In particular, applied to $ \xi_t $ one achieves an even better bound depending on $ \langle \xi_t, a_{q'}^* a_{q'} \xi_t \rangle $, which is also used in the bootstrap.

\begin{lemma}[Bootstrap Bounds on $ c^*, c $]
\label{lem:ccgbound}
Let $ k \in \Gamma^{\nor} $, $ \alpha_q \in \cI_{k} $, and $ q \in B_{\alpha_q} $. Let $g_{q,k}$ be defined as in \cref{eq:g}. Then for all $ \psi \in \cF $ and any choice $ \sharp \in \{ \cdot, * \} $ we have
\begin{equation}
\label{eq:ccgbound1}
	\Vert c^\sharp(g_{q, k}) \psi \Vert
	\le \frac{1}{n_{\alpha_q, k}}\;.
\end{equation}

Further, suppose there is $ r \ge 0 $ and $C >0$ such that for all $ t \in [-1, 1] $ and all $ {q'} \in \ZZZ^3 $ we have
\begin{equation}
\label{eq:nqcondition2}
	\langle \xi_t, a_{q'}^* a_{q'} \xi_t \rangle
	\le C N^{-r}\;.
\end{equation}
Then there exists $C > 0$ such that for all $ t \in [-1, 1] $ we have
\begin{equation}
\label{eq:ccgbound2}
	\Vert c(g_{q, k}) \xi_t \Vert
	\le \frac{C}{n_{\alpha_q, k}} N^{-\frac{r}{2}}\;.
\end{equation}
\end{lemma}
We caution the reader that a bound like \cref{eq:ccgbound2} does not hold for $ c^*(g_{q, k}) $.
\begin{proof}
Statement \cref{eq:ccgbound1} follows from definition \cref{eq:g} and $ \Vert a_q^\sharp \Vert_{\op} \le 1 $: For $ q \in B_{\F}^c $, e.\,g.,
\begin{equation}
	\Vert c^*(g_{q, k}) \xi_t \Vert
	\le \frac{1}{n_{\alpha_q, k}} \Vert a_q^* a_{q \mp k}^* \xi_t \Vert
	\le \frac{1}{n_{\alpha_q, k}} \Vert \xi_t \Vert
	= \frac{1}{n_{\alpha_q, k}}\;.
\end{equation}
Concerning the stronger bound \cref{eq:ccgbound2}, in case $ q \in B_{\F}^c $, we use $a^*_{q\mp k} a_{q\mp k} \leq 1$ to get
\begin{equation}
	\Vert c(g_{q, k}) \xi_t \Vert
	 \le \frac{1}{n_{\alpha_q, k}} \Vert a_{q \mp k} a_q \xi_t \Vert \le \frac{1}{n_{\alpha,k}}\sqrt{\langle \xi_t, a^*_q a_q \xi_t\rangle} \leq \frac{1}{n_{\alpha,k}}\sqrt{C N^{-r}}\;.
	\end{equation}
The same arguments apply to $ q \in B_{\F} $.
\end{proof}

Finally, we also need to bound combinations of $ c^*$--, $c $--, and $ \cE_\alpha $--operators.
\begin{lemma}[Bounds on Combinations of $ c^*$, $c $, and $ \cE_\alpha $]
\label{lem:cccEbound}
Let $ k, \ell \in \Gamma^{\nor} $, $ \alpha \in \cI_{\ell}, \beta \in \cI_{k} \cap \cI_{\ell} $ and suppose that $\hat{V}$ is non-negative. Then there is $ C > 0 $ such that for all $ n, m \in \NNN $ and all choices of $ \sharp \in \{ \cdot, * \} $ we have
\begin{equation}
\label{eq:cccEbound}
\begin{aligned}
	\Big\Vert \sum_{\substack{\alpha \in \cI_{\ell} \\ \beta \in \cI_{k} \cap \cI_{\ell}}} (K(k)^n)_{\alpha_q, \beta} (K(\ell)^m)_{\alpha, \beta} c^\sharp_\alpha(\ell) \cE_\beta(k, \ell)^\sharp \psi \Big\Vert
	& \le \frac{(C \hat{V}_k)^n (C \hat{V}_\ell)^m}{\fn M} \Vert (\cN + 1) \psi \Vert \;,\\
	\Big\Vert \sum_{\substack{\alpha \in \cI_{\ell} \\ \beta \in \cI_{k} \cap \cI_{\ell}}} (K(k)^n)_{\alpha_q, \beta} (K(\ell)^m)_{\alpha, \beta} \cE_\beta(k, \ell)^\sharp c^\sharp_\alpha(\ell) \psi \Big\Vert
	& \le \frac{(C \hat{V}_k)^n (C \hat{V}_\ell)^m}{\fn M} \Vert (\cN + 1) \psi \Vert \;.
\end{aligned}
\end{equation}
\end{lemma}
\begin{proof}
To establish the first bound, we use the fact that $ \cE_\beta(k, \ell)^\sharp $ commutes with $ \cN $:
\begin{align*}
	&\Big\Vert \sum_{\substack{\alpha \in \cI_{\ell} \\ \beta \in \cI_{k} \cap \cI_{\ell}}} (K(k)^n)_{\alpha_q, \beta} (K(\ell)^m)_{\alpha, \beta} c^\sharp_\alpha(\ell) \cE_\beta(k, \ell)^\sharp \psi \Big\Vert\\
	& \le \sum_{\beta \in \cI_{k} \cap \cI_{\ell}} |(K(k)^n)_{\alpha_q, \beta}| \Big\Vert \sum_{\alpha \in \cI_{\ell}} (K(\ell)^m)_{\alpha, \beta} c^\sharp_\alpha(\ell) \cE_\beta(k, \ell)^\sharp \psi \Big\Vert\\
	& \overset{\cref{eq:ccbound}}{\le} \sum_{\beta \in \cI_{k} \cap \cI_{\ell}} |(K(k)^n)_{\alpha_q, \beta}| \Big( \sum_{\alpha \in \cI_{\ell}} |(K(\ell)^m)_{\alpha, \beta}|^2 \Big)^{\frac 12} \Vert (\cN + 1)^{\frac 12} \cE_\beta(k, \ell)^\sharp \psi \Big\Vert \\
	& \overset{\cref{eq:Knbound}}{\le} (C \hat{V}_k)^n (C \hat{V}_\ell)^m M^{-\frac{2}{3}} \sum_{\beta \in \cI_{k} \cap \cI_{\ell}} \Vert \cE_\beta(k, \ell)^\sharp (\cN + 1)^{\frac 12} \psi \Vert\\
	& \le (C \hat{V}_k)^n (C \hat{V}_\ell)^m M^{-\frac{2}{3}} \Big( \sum_{\beta \in \cI_{k} \cap \cI_{\ell}} 1 \Big)^{\frac 12} \Big( \sum_{\beta \in \cI_{k} \cap \cI_{\ell}} \Vert \cE_\beta(k, \ell)^\sharp (\cN + 1)^{\frac 12} \psi \Vert^2 \Big)^{\frac{1}{2}}\\
	& \overset{\cref{eq:cEsumbound}}{\le} (C \hat{V}_k)^n (C \hat{V}_\ell)^m \fn^{-1} M^{-1} \Vert (\cN + 1) \psi \Vert \;. \tagg{eq:cccEbound1stline}
\end{align*}

For the second line of \cref{eq:cccEbound}, we start with
\begin{equation}
\begin{aligned}
\label{eq:triangleestimateKKcEc}
	&\Big\Vert \sum_{\substack{\alpha \in \cI_{\ell} \\ \beta \in \cI_{k} \cap \cI_{\ell}}} (K(k)^n)_{\alpha_q, \beta} (K(\ell)^m)_{\alpha, \beta} \cE_\beta(k, \ell)^\sharp c^\sharp_\alpha(\ell) \psi \Big\Vert\\
	& \le \Big\Vert \sum_{\substack{\alpha \in \cI_{\ell} \\ \beta \in \cI_{k} \cap \cI_{\ell}}} (K(k)^n)_{\alpha_q, \beta} (K(\ell)^m)_{\alpha, \beta} c^\sharp_\alpha(\ell) \cE_\beta(k, \ell)^\sharp \psi \Big\Vert\\
	& \quad + \Big\Vert \sum_{\substack{\alpha \in \cI_{\ell} \\ \beta \in \cI_{k} \cap \cI_{\ell}}} (K(k)^n)_{\alpha_q, \beta} (K(\ell)^m)_{\alpha, \beta} [\cE_\beta(k, \ell)^\sharp, c^\sharp_\alpha(\ell)] \psi \Big\Vert \;.
\end{aligned}
\end{equation}
The first term is bounded by \cref{eq:cccEbound1stline}. The commutator can be explicitly evaluated:
\[
\begin{aligned}\relax
	&[\cE_\beta(k, k'), c_\alpha(k'')]\\
	& = - \frac{\delta_{\alpha, \beta}}{n_{\alpha, k} n_{\alpha, k'} n_{\alpha, k''}} \left(
		\sum_{h : \alpha} f^{(\alpha)}_{k, k', k''}(h) a_{h \pm k''} a_{h \pm k' \mp k}
		- \sum_{p : \alpha} g^{(\alpha)}_{k, k', k''}(p) a_{p \mp k''} a_{p \mp k' \pm k}
	\right)  \;,
\end{aligned}
\]
with 
\begin{equation}
	f^{(\alpha)}_{k, k', k''}(h) := \chi \left( \substack{
		h \pm k'' \in B_\alpha \cap B_{\F}^c\\
		h \pm k' \in B_\alpha \cap B_{\F}^c\\
		h \pm k' \mp k \in B_\alpha \cap B_{\F}
		} \right)\;, \qquad
	g^{(\alpha)}_{k, k', k''}(p) := \chi \left( \substack{
		p \mp k'' \in B_\alpha \cap B_{\F}\\
		p \mp k' \in B_\alpha \cap B_{\F}\\
		p \mp k' \pm k \in B_\alpha \cap B_{\F}^c
		} \right) \;.
\end{equation}
The three commutators for different choices of $ \sharp \in \{ \cdot, * \} $ can easily be deduced via $ \cE_\beta(k, \ell)^* = \cE_\beta(\ell, k) $ and $ [\cE_\beta(k, \ell)^*, c_\alpha(\ell)^*] = -([\cE_\beta(k, \ell), c_\alpha(\ell)])^* $. If in both instances $ \sharp = * $ is chosen, we can bound
\[
\begin{aligned}
	&\Big\Vert \sum_{\substack{\alpha \in \cI_{\ell} \\ \beta \in \cI_{k} \cap \cI_{\ell}}} (K(k)^n)_{\alpha_q, \beta} (K(\ell)^m)_{\alpha, \beta} [\cE_\beta(k, \ell)^*, c^*_\alpha(\ell)] \psi \Big\Vert\\
	& \le \Big\Vert \!\sum_{\alpha \in \cI_{k} \cap \cI_{\ell}} \!\!  \frac{(K(k)^n)_{\alpha_q, \alpha} (K(\ell)^m)_{\alpha, \alpha}}{n_{\alpha, k} n_{\alpha, \ell}^2}
		\Big( \sum_{h : \alpha} f^{(\alpha)}_{k, \ell, \ell}(h) a^*_{h \pm \ell \mp k} a^*_{h \pm \ell} 
		- \sum_{p : \alpha} g^{(\alpha)}_{k, \ell, \ell}(p) a^*_{p \mp \ell \pm k} a^*_{p \mp \ell} \Big) \psi
	\Big\Vert \,.
\end{aligned}
\]
Recall the following elementary bounds for $ f \in \ell^2(\ZZZ^3 \times \ZZZ^3) $ from \cite[Lemma~3.1]{BPS14}:
\begin{equation}
\label{eq:astarastarbound}
\begin{aligned}
	\Big\Vert \sum_{k_1, k_2 \in \ZZZ^3} f(k_1, k_2) a_{k_1}^* a_{k_2}^* \psi \Big\Vert
	& \le 2 \Vert f \Vert_2 \Vert (\cN + 1)^{\frac 12} \psi \Vert \;,\\
	\Big\Vert \sum_{k_1, k_2 \in \ZZZ^3} f(k_1, k_2) a_{k_1} a_{k_2} \psi \Big\Vert
	& \le \Vert f \Vert_2 \Vert \cN^{\frac{1}{2}} \psi \Vert \;.
\end{aligned}
\end{equation}
Introducing the functions
\begin{equation}
\begin{aligned}
	f_{k, k', k''}(h) := &\sum_{\alpha \in \cI_{k} \cap \cI_{\ell}} \frac{\chi(h : \alpha)}{n_{\alpha, k} n_{\alpha, k'} n_{\alpha, k''}} (K(k)^n)_{\alpha_q, \alpha} (K(k'')^m)_{\alpha, \alpha} f^{(\alpha)}_{k, k', k''}(h) \;,\\
	g_{k, k', k''}(p) := &\sum_{\alpha \in \cI_{k} \cap \cI_{\ell}} \frac{\chi(p : \alpha)}{n_{\alpha, k} n_{\alpha, k'} n_{\alpha, k''}} (K(k)^n)_{\alpha_q, \alpha} (K(k'')^m)_{\alpha, \alpha} g^{(\alpha)}_{k, k', k''}(p)\\
\end{aligned}
\end{equation}
we estimate
\begin{equation}
\begin{aligned}
	&\Big\Vert \sum_{\substack{\alpha \in \cI_{\ell} \\ \beta \in \cI_{k} \cap \cI_{\ell}}} (K(k)^n)_{\alpha_q, \beta} (K(\ell)^m)_{\alpha, \beta} [\cE_\beta(k, \ell)^*, c^*_\alpha(\ell)] \psi \Big\Vert\\
	& \le \Big\Vert \sum_{h \in \BF} f_{k, \ell, \ell}(h) a^*_{h \pm \ell \mp k} a^*_{h \pm \ell} \psi \Big\Vert
	+ \Big\Vert \sum_{p \in \BFc} g_{k, \ell, \ell}(p) a^*_{p \mp \ell \pm k} a^*_{p \mp \ell} \psi \Big\Vert \\
	& \le 2 \big( \Vert f_{k, \ell, \ell} \Vert_2 + \Vert g_{k, \ell, \ell} \Vert_2 \big)
		\Vert (\cN + 1)^{\frac{1}{2}} \psi \Vert \;,
\end{aligned}
\end{equation}
where in the last line we used \cref{eq:astarastarbound} with $ f(h \pm \ell \mp k, h \pm \ell) = f_{k, \ell, \ell}(h) $, so $ \Vert f \Vert_2 =  \Vert f_{k, \ell, \ell} \Vert_2 $, and $ f(p \mp \ell \pm k, p \mp \ell) = g_{k, \ell, \ell}(h) $, so $ \Vert f \Vert_2 =  \Vert g_{k, \ell, \ell} \Vert_2 $. To estimate $ \Vert f_{k, \ell, \ell} \Vert_2 $, note that
\begin{equation}
	|f^{(\alpha)}_{k, \ell, \ell}(h)| \le 1 \quad \Rightarrow \quad
	|f_{k, \ell, \ell}(h)| \overset{\cref{eq:Knbound}}{\le} (C \hat{V}_k)^n (C \hat{V}_\ell)^m \fn^{-3} M^{-2} \;.
\end{equation}
Further, the support of $ f_{k, \ell, \ell} $ only contains holes with a distance $ \le R $ to the Fermi surface $ \partial B_{\F} $, whose surface area is of order $ N^{\frac{2}{3}} $. Thus
\begin{equation}
\begin{aligned}
	\Vert f_{k, \ell, \ell} \Vert_2
	& \le (C \hat{V}_k)^n (C \hat{V}_\ell)^m \fn^{-3} M^{-2} \lvert \supp(f_{k, \ell, \ell}) \rvert^{\frac{1}{2}}\\
	& \le (C \hat{V}_k)^n (C \hat{V}_\ell)^m \fn^{-3} M^{-2} N^{\frac{1}{3}}
	= (C \hat{V}_k)^n (C \hat{V}_\ell)^m \fn^{-1} M^{-1} N^{- \frac{1}{3} + \delta} \;.
\end{aligned}
\end{equation}
The same bound applies to $ \Vert g_{k, \ell, \ell} \Vert_2 $. Therefore, the second term on the r.\,h.\,s.\ of \cref{eq:triangleestimateKKcEc} for $ \sharp = * $ is bounded by
\begin{equation}
\begin{aligned}
	&\Big\Vert \sum_{\substack{\alpha \in \cI_{\ell} \\ \beta \in \cI_{k} \cap \cI_{\ell}}} (K(k)^n)_{\alpha_q, \beta} (K(\ell)^m)_{\alpha, \beta} [\cE_\beta(k, \ell)^*, c^*_\alpha(\ell)] \psi \Big\Vert\\
	& \le  (C \hat{V}_k)^n (C \hat{V}_\ell)^m \fn^{-1} M^{-1} N^{- \frac{1}{3} + \delta} \Vert (\cN + 1)^{\frac 12} \psi \Vert \;.
\end{aligned}
\end{equation}
Since $ \delta < \frac{1}{6}< \frac{1}{3} $ and $ (\cN + 1)^{\frac{1}{2}} \le (\cN + 1) $, this is smaller than the required bound in the second line of \cref{eq:cccEbound}. So we established the second line of \cref{eq:cccEbound} for $ \sharp = * $ in both places. The bound for the three other choices of $ \sharp \in \{ \cdot, * \} $ is obtained analogously.
\end{proof}

The following lemma was already proved in \cite[Lemma~7.2]{BNPSS21}, using Gr\"onwall's lemma.
\begin{lemma}[Stability of Number Operators]
\label{lem:gronwall}
Assume that $\hat{V}$ is non-negative and compactly supported. Then for all $m \in \Nbb_0$ there exists $ C_m > 0 $ such that for all $ t \in [-1, 1] $
\begin{equation}
\label{eq:gronwall}
	e^{tS} (\cN + 1)^m e^{-tS} \le e^{C_m |t|} (\cN + 1)^m\;.
\end{equation}
\end{lemma}

\section{Controlling the Bosonized Terms}
\label{sec:bosonizedterms}
We now use the bounds from the previous section to derive estimates on the expectation values of $ \ad^n_{q, (\b)} $, defined in \cref{eq:adnb}, and on the commutation error
\begin{equation}
\label{eq:cE2n}
	\cE_{n, q} := [S, \ad^{n-1}_{q, (\b)}] - \ad^{n}_{q, (\b)}\;, \qquad n \in \NNN\;.
\end{equation}
\begin{lemma}[Bound on the Commutation Error]
\label{lem:cE2estimate}
Suppose there is $ r \in [0, \frac{2}{3}] $ such that for all $ t \in [-1, 1] $ and $ q' \in \ZZZ^3 $ we have (with $ C $ independent of $ t, q' $)
\begin{equation}
\label{eq:nqcondition3}
	\langle \xi_t, a_{q'}^* a_{q'} \xi_t \rangle
	\le C N^{-r}\;.
\end{equation}
Further, suppose that $ \hat{V}$ is non-negative and compactly supported. Then there exists a constant $ \fC_1 > 0 $ such that\footnote{We use $ \fC_1 $ instead of $ C $ to explicitly track the $ n $-dependence of the constants. This will be important for ensuring that sums over $ n $ converge.} for all $ n, N \in \NNN $ and $ t \in [-1, 1] $ we have
\begin{equation}
\label{eq:cE2estimate}
	|\langle \xi_t, \cE_{n, q} \xi_t \rangle|
	\le \fC_1^n e^{C |t|} \; N^{- \frac{2}{3}+ \frac{3}{2} \delta - \frac r2} \;.
\end{equation}
\end{lemma}

\begin{proof}
If $ q $ is not inside any patch, then the statement is trivial since $ \ad^{n-1}_{q, (\b)} = 0  = \ad^n_{q, (\b)} $, thus also $\cE_{n, q} = 0 $. So we may assume that $ q \in B_{\alpha_q} $ for some $ 1 \le \alpha_q \le M $. Recall that $ \ad^n_{q, (\b)} $ is given by $ \bA$, $\bB $, and $ \bC $ (for even $ n $) or $ \bD$, $\bE $, and $ \bF $ (for odd $ n $) as defined in \cref{eq:ABC,eq:DEF}. We define
\begin{equation}
\label{eq:cEABCDEF}
\begin{aligned}
	\cE^{(A)}_n & := [S, \bA_n] = 0\\
	\cE^{(B)}_n & := [S, \bB_n] - \bD_{n, 1} - \bE_{n+1}\\
	\cE^{(C)}_{n-m, m} & := [S, \bC_{n-m, m}] - \bD_{n-m+1, m} - \bF_{n-m, m+1}\\
	\cE^{(D)}_{n-m, m} & := [S, \bD_{n-m, m}] - \bA_{n+1} - \bC_{n-m+1, m} - \bC_{n-m, m+1}\\
	\cE^{(E)}_n & := [S, \bE_n] - \bA_{n+1} - \bC_{n, 1} - \bB_{n+1}\\
	\cE^{(F)}_{n-m, m} & := [S, \bF_{n-m, m}] - \bA_{n+1} - \bC_{n-m+1, m} - \bC_{n-m, m+1}\;.
\end{aligned}
\end{equation}
(In the proof of Lemma \ref{lem:shoelace}, we will see that the subtracted terms, such as $ \bD_{n, 1} + \bE_{n+1} $ for $ \cE^{(B)} $, are the bosonization approximation of the commutators. So $ \cE^{(A)}_n$, \ldots,  $\cE^{(F)}_{n-m, m} $ are the deviations from the bosonization approximation.)
Then we may express the commutation error for $ n \ge 1 $ as
\begin{equation}
\label{eq:cE2nplus1}
	\cE_{n+1, q} = \begin{cases}
		\displaystyle \cE^{(B)}_n + (\cE^{(B)}_n)^* + \sum_{m=1}^{n-1} \binom{n}{m} \cE^{(C)}_{n-m, m} \quad &\text{for } n: \text{even}\\
		\displaystyle \cE^{(E)}_n + (\cE^{(E)}_n)^* + \sum_{m=1}^{s} \binom{n}{m} \cE^{(D)}_{n-m, m} + \sum_{m=1}^{s} \binom{n}{m} \cE^{(F)}_{m, n-m} \quad &\text{for } n: \text{odd}\;,\\
	\end{cases}
\end{equation}
where $ n = 2s+1 $ for odd $ n $. For $ n = 0 $, we have $ \cE_1 = 0 $, since $ \ad^1_{q, (\b)} = [S, a_q^* a_q] = [S, \ad^0_{q, (\b)}] $ (see also Appendix \ref{app:motivationbosocc}).

\smallskip

\noindent \underline{Bounding $ \cE^{(B)} $:} Recall that $ K(k) $ is symmetric and that $ k \in \tilde{\cC}^q $ implies $ \alpha_q \in \cI_{k} $. For this proof, let us adopt the convention that $ \sum_k $ denotes a sum over $ k \in \tilde{\cC}^q \cap \ZZZ^3$ and $ \sum_{k'} $ a sum over all $ k' \in \north $ such that $ \alpha_q \in \cI_{k'} $. A straightforward computation starting from the definitions \cref{eq:ABC,eq:DEF} and using \cref{lem:approxCCR,lem:dg} renders
\begin{align*}
	|\langle \xi_t, \cE^{(B)}_n \xi_t \rangle|
	& = \Bigg\vert \frac{1}{2} \sum_{k, k'} \sum_{\substack{\alpha \in \cI_{k'} \\ \alpha_1 \in \cI_{k}}} K(k')_{\alpha, \alpha_q} (K(k)^n)_{\alpha_q, \alpha_1}
		\\ & \hspace{6em} \times \langle \xi_t, \big( c_{\alpha}(k') \cE^{(g)}_q(k', k) c_{\alpha_1}(k) + \cE^{(g)}_q(k', k) c_{\alpha}(k') c_{\alpha_1}(k) \big) \xi_t \rangle\\
	& \quad + \frac{1}{2} \sum_{k, k'} \sum_{\substack{\alpha \in \cI_{k'} \\ \beta \in \cI_{k} \cap \cI_{k'}}} K(k')_{\alpha, \beta} (K(k)^n)_{\alpha_q, \beta}
		\\ & \hspace{6em} \times \langle \xi_t, \big( c^*(g_{q, k}) c^*_{\alpha}(k') \cE_\beta(k, k') + c^*(g_{q, k}) \cE_\beta(k, k') c^*_{\alpha}(k') \big) \xi_t \rangle \Bigg\vert\;.
\end{align*}
By the Cauchy--Schwartz inequality
\begin{equation}
	|\langle \xi_t, \cE^{(B)}_n \xi_t \rangle| \leq (\I^{(B)}) + (\II^{(B)}) + (\III^{(B)}) + (\IV^{(B)})\;,
\end{equation}
where
{\allowdisplaybreaks
\begin{align*}
	(\I^{(B)}) & := \frac{1}{2} \sum_{k, k'}
		\Big\Vert \sum_{\alpha \in \cI_{k'}} K(k')_{\alpha_q, \alpha} c^*_{\alpha}(k') \xi_t \Big\Vert
		\Big\Vert \sum_{\alpha_1 \in \cI_{k}} (K(k)^n)_{\alpha_q, \alpha_1} \cE^{(g)}_q(k', k) c_{\alpha_1}(k) \xi_t \Big\Vert \;,\\
	(\II^{(B)})& := \frac{1}{2} \sum_{k, k'}
		\Big\Vert \sum_{\alpha \in \cI_{k'}} K(k')_{\alpha_q, \alpha} c^*_{\alpha}(k') \cE^{(g)}_q(k', k)^* \xi_t \Big\Vert
		\Big\Vert \sum_{\alpha_1 \in \cI_{k}} (K(k)^n)_{\alpha_q, \alpha_1} c_{\alpha_1}(k) \xi_t \Big\Vert \;,\\
	(\III^{(B)})& := \frac{1}{2} \sum_{k, k'}
		\Vert c(g_{q, k}) \xi_t \Vert
		\Big\Vert \sum_{\substack{\alpha \in \cI_{k'} \\ \beta \in \cI_{k} \cap \cI_{k'}}} K(k')_{\alpha, \beta} (K(k)^n)_{\alpha_q, \beta} c^*_{\alpha}(k') \cE_\beta(k, k') \xi_t \Big\Vert \;,\\
	(\IV^{(B)}) & := \frac{1}{2} \sum_{k, k'}
		\Vert c(g_{q, k}) \xi_t \Vert
		\Big\Vert \sum_{\substack{\alpha \in \cI_{k'} \\ \beta \in \cI_{k} \cap \cI_{k'}}} K(k')_{\alpha, \beta} (K(k)^n)_{\alpha_q, \beta} \cE_\beta(k, k') c^*_{\alpha}(k') \xi_t \Big\Vert \;.
\end{align*}
}
We start with bounding $ (\II^{(B)}) $, which is slightly easier than $ (\I^{(B)}) $. By \cref{lem:ccbound}
\begin{equation}
\label{eq:IIBestimate1}
	(\II^{(B)})
	\le \frac{1}{2} \sum_{k, k'} \Vert K(k')_{\alpha_q, \cdot} \Vert_2 \;
	\big\Vert (\cN + 1)^{\frac{1}{2}} \cE^{(g)}_q(k', k)^* \xi_t \big\Vert \;
	\Vert (K(k)^n)_{\alpha_q, \cdot} \Vert_2 \;
	\big\Vert \cN^{\frac{1}{2}} \xi_t \big\Vert\;.
\end{equation}
\Cref{lem:Knbound} yields
\begin{equation}
\label{eq:Kncontrol}
	\Vert (K(k)^n)_{\alpha_q, \cdot} \Vert_2
	= \bigg( \sum_\alpha |(K(k)^n)_{\alpha_q, \alpha}|^2 \bigg)^{\frac{1}{2}}
	\le \left( M (C \hat{V}_k)^{2n} M^{-2} \right)^{\frac{1}{2}}
	\le (C \hat{V}_k)^n M^{-\frac{1}{2}}\;.
\end{equation}
The second factor containing $ \cN $ is controlled using \cref{lem:gronwall}
\begin{equation}
\label{eq:C01estimate}
\begin{aligned}
	\big\Vert \cN^{\frac{1}{2}} \xi_t \big\Vert^2
	\le &\langle \Omega, e^{tS} (\cN + 1) e^{-tS} \Omega \rangle 
	\le \langle \Omega, e^{C |t|} (\cN + 1) \Omega \rangle
	= e^{C |t|}\;.
\end{aligned}
\end{equation}
In the first factor containing $ \cN $ in \cref{eq:IIBestimate1}, we apply \cref{lem:cEgbound} together with $ n_{\alpha_q, k} \ge C \fn $ to get, for an arbitrarily small $ \varepsilon > 0 $ and some $ c_\varepsilon > 0 $ depending on $ \varepsilon $,
\begin{equation}
\label{eq:cNdelta112bound}
	\big\Vert (\cN + 1)^{\frac{1}{2}} \cE^{(g)}_q(k', k)^* \xi_t \big\Vert
	= \big\Vert \cN^{\frac{1}{2}} \cE^{(g)}_q(k', k)^* \xi_t \big\Vert
	+ \big\Vert \cE^{(g)}_q(k', k)^* \xi_t \big\Vert
	\le C \frac{e^{c_\varepsilon |t|}}{\fn^2} N^{-\frac r2 + \varepsilon}\;,
\end{equation}
where we used that $ \cE^{(g)}_q(k', k) $ preserves the particle number. So finally, choosing a fixed $\varepsilon > 0$ small enough, there exists $ c_\varepsilon $ (independent of $ r, t, q' $ and $ n $) such that
\begin{equation}
\label{eq:IIBestimate}
	(\II^{(B)}) \le e^{c_\varepsilon |t|} \sum_{k, k'} (C \hat{V}_k)^n (C \hat{V}_{k'}) \fn^{-2} M^{-1} N^{-\frac r2 + \varepsilon} \;.
\end{equation}
(The $\varepsilon$-dependent estimates turn out to be subleading to the other error terms, so that the dependence on $\varepsilon$ does not show up in the statement of the lemma.)

In $ (\I^{(B)}) $, using \cref{lem:ccbound,lem:Knbound}, the first norm can be bounded by
\begin{equation}
 \label{eq:Ib_first}
 \Big\Vert \sum_{\alpha \in \cI_{k'}} K(k')_{\alpha_q, \alpha} c^*_{\alpha}(k') \xi_t \Big\Vert \leq \norm{K(k')_{\alpha_q,\cdot}}_2 \norm{(\Ncal + 1)^{\frac{1}{2}} \xi_t} \leq C\hat{V}_{k'} M^{-\frac{1}{2}} e^{C\lvert t\rvert \norm{\hat{V}}_1}\;.
\end{equation}
For the second norm in $ (\I^{(B)}) $, we have
\begin{align*}
	&\Big\Vert \sum_{\alpha_1 \in \cI_{k}} (K(k)^n)_{\alpha_q, \alpha_1} \cE^{(g)}_q(k', k) c_{\alpha_1}(k) \xi_t \Big\Vert \tagg{eq:cEgcommutatorsum}\\
	& \le \Big\Vert \sum_{\alpha_1 \in \cI_{k}} (K(k)^n)_{\alpha_q, \alpha_1} c_{\alpha_1}(k) \cE^{(g)}_q(k', k) \xi_t \Big\Vert
	+ \Big\Vert \sum_{\alpha_1 \in \cI_{k}} (K(k)^n)_{\alpha_q, \alpha_1} [ \cE^{(g)}_q(k', k), c_{\alpha_1}(k) ] \xi_t \Big\Vert\;.
\end{align*}
The first norm on the r.\,h.\,s.\ is bounded just as the second norm in $ (\II^{(B)}) $ by
\begin{equation}
\label{eq:Ib_rsec}
\Big\Vert \sum_{\alpha_1 \in \cI_{k}} (K(k)^n)_{\alpha_q, \alpha_1} c_{\alpha_1}(k) \cE^{(g)}_q(k', k) \xi_t \Big\Vert \leq e^{c_\varepsilon |t|} (C \hat{V}_k)^n \fn^{-2} M^{-\frac{1}{2}} N^{-\frac{r}{2} + \varepsilon}\;.
\end{equation}
For the second norm on the r.\,h.\,s., in order to control the commutator term, we use the explicit form of $ \cE^{(g)}_q(k', k) $ as in \cref{eq:cEg}. We restrict to the case $ q \in B_{\F}^c $ ($ q \in B_{\F} $ can be treated analogously) and use that $ [\cE^{(g)}_q(k', k), c_\alpha^\sharp(k)] = 0 $ whenever $ \alpha \neq \alpha_q $:
\begin{align*}
	&\Big\Vert \sum_{\alpha_1 \in \cI_{k}} (K(k)^n)_{\alpha_q, \alpha_1} [ \cE^{(g)}_q(k', k), c_{\alpha_1}(k) ] \xi_t \Big\Vert \quad \le (C \hat{V}_k)^n M^{-1} \Vert [ \cE^{(g)}_q(k', k), c_{\alpha_q}(k) ] \xi_t \Vert\\
	& \le \frac{(C \hat{V}_k)^n M^{-1}}{n_{\alpha_q, k}^2 n_{\alpha_q, k'}} \sum_{p : \alpha_q}
		\left( \Vert [a_{q \mp k}^* a_{q \mp k'}, a_{p \mp k} a_p] \xi_t \Vert
		+ \Vert [a_q^* a_{q \mp k \pm k'}, a_{p \mp k} a_p] \xi_t \Vert \right)\\
	& = \frac{(C \hat{V}_k)^n M^{-1}}{n_{\alpha_q, k}^2 n_{\alpha_q, k'}}
	\left( \Vert a_{q \mp k'} a_q \xi_t \Vert
		+ \Vert a_{q \mp k \pm k'} a_{q \pm k} \xi_t \Vert \right) \quad \le (C \hat{V}_k)^n \fn^{-3} M^{-1}\;.	\tagg{eq:n716}
\end{align*}
By definition \cref{eq:fn} of $ \fn $, \cref{eq:n716} scales like $ \fn^{-3} M^{-1} = \fn^{-2} M^{-\frac{1}{2}} N^{- \frac 13 + \frac{\delta}{2}} $, whereas \cref{eq:Ib_rsec} scales like $ \fn^{-2} M^{-\frac{1}{2}} N^{-\frac r2 + \varepsilon} $. Since $ r \le \frac{2}{3}$, remembering the factor \cref{eq:Ib_first} and fixing some $ \varepsilon < \frac{\delta}{2} $, we have the common bound
\begin{equation}
\label{eq:IBestimate}
	(\I^{(B)}) \le e^{C |t|} \sum_{k, k'} (C \hat{V}_k)^n (C \hat{V}_{k'}) \fn^{-2} M^{-1} N^{-\frac r2 + \frac{\delta}{2}}\;.
\end{equation}

Concerning $ (\III^{(B)}) $, the second norm is bounded by Lemma \ref{lem:cccEbound} as
\[
	\Big\Vert \sum_{\substack{\alpha \in \cI_{k'} \\ \beta \in \cI_{k} \cap \cI_{k'}}} K(k')_{\alpha, \beta} (K(k)^n)_{\alpha_q, \beta} c^*_{\alpha}(k') \cE_\beta(k, k') \xi_t \Big\Vert
	\le (C \hat{V}_k)^n (C \hat{V}_{k'}) \fn^{-1} M^{-1}
	\Vert (\cN + 1) \xi_t \Vert \;.
\]
Estimating the first norm in $ (\III^{(B)}) $ with \cref{eq:ccgbound2}, we obtain
\begin{equation}
\begin{aligned}
	(\III^{(B)})
	\le & \sum_{k, k'} (C \hat{V}_k)^n (C \hat{V}_{k'}) \fn^{-2} M^{-1} N^{- \frac r2}
	\Vert (\cN + 1) \xi_t \Vert\;.
\end{aligned}
\end{equation}
By \cref{lem:gronwall} we have
	$\big\Vert (\cN + 1) \xi_t \big\Vert^2
	\le e^{C |t|} \langle \Omega, (\cN + 1)^2 \Omega \rangle
	\le e^{C |t|}$.
Consequently
\begin{equation}
\label{eq:IIIBestimate}
	(\III^{(B)}) \le e^{C |t|} \sum_{k, k'} (C \hat{V}_k)^n (C \hat{V}_{k'}) \fn^{-2} M^{-1} N^{- \frac{r}{2}}\;.
\end{equation}
The term $ (\IV^{(B)}) $ is bounded analogously by
\begin{equation}
\label{eq:IVBestimate}
	(\IV^{(B)}) \le e^{C |t|} \sum_{k, k'} (C \hat{V}_k)^n (C \hat{V}_{k'}) \fn^{-2} M^{-1} N^{- \frac{r}{2}}\;.
\end{equation}
We add up all four contributions to $ \cE^{(B)}_n $ and get
\begin{equation}
\label{eq:cEBestimate}
	| \langle \xi_t, \cE^{(B)}_n \xi_t \rangle |
	\le e^{C |t|} \sum_{k, k'} (C \hat{V}_k)^n (C \hat{V}_{k'}) \fn^{-2} M^{-1} N^{- \frac{r}{2} + \frac{\delta}{2}} \;.
\end{equation}
 The same bound applies to $ | \langle \xi_t, (\cE^{(B)}_n)^* \xi_t \rangle | =  | \langle \xi_t, \cE^{(B)}_n \xi_t \rangle | $.

\smallskip

\noindent \underline{Bounding $ \cE^{(C)} $:}
As before, from definitions \cref{eq:ABC,eq:DEF}, using \cref{lem:approxCCR,lem:dg}, we obtain
\begin{align*}
	| \langle \xi_t, \cE^{(C)}_{n-m, m} \xi_t \rangle |
	& = \Bigg\vert \frac{1}{2} \sum_{k, k'} \sum_{\substack{\alpha \in \cI_{k'} \\ \alpha_2 \in \cI_{k} \\ \beta \in \cI_{k} \cap \cI_{k'}}} K(k')_{\alpha, \beta} (K(k)^{n-m})_{\alpha_q, \beta} (K(k)^m)_{\alpha_q, \alpha_2} \rho_{q, k} \\
		&\qquad\qquad \times \langle \xi_t, \big( c_{\alpha}(k') \cE_\beta(k', k) c_{\alpha_2}(k) + \cE_\beta(k', k) c_{\alpha}(k') c_{\alpha_2}(k) \big) \xi_t \rangle \\
	& \quad + \frac{1}{2} \sum_{k, k'} \sum_{\substack{\alpha \in \cI_{k'} \\ \alpha_1 \in \cI_{k} \\ \beta \in \cI_{k} \cap \cI_{k'}}} K(k')_{\alpha, \beta} (K(k)^{n-m})_{\alpha_q, \alpha_1} (K(k)^m)_{\alpha_q, \beta} \rho_{q, k} \\
		&\qquad\qquad \times \langle \xi_t, \big( c^*_{\alpha_1}(k) c^*_{\alpha}(k') \cE_\beta(k, k') + c^*_{\alpha_1}(k) \cE_\beta(k, k') c^*_{\alpha}(k') \big) \xi_t \rangle \Bigg\vert
	\;.
\end{align*}
By the Cauchy--Schwarz inequality
\begin{align*}
	&| \langle \xi_t, \cE^{(C)}_{n-m, m} \xi_t \rangle |  \leq (\I^{(C)}) + (\II^{(C)}) + (\III^{(C)}) + (\IV^{(C)})
\end{align*}
where
{\allowdisplaybreaks
\begin{align*}
	(\I^{(C)}) & := \sum_{k, k'} \frac{\rho_{q, k}}{2}
		\Big\lVert \!\!\!\sum_{\substack{\alpha \in \cI_{k'} \\ \beta \in \cI_{k} \cap \cI_{k'}}} \!\!\! K(k')_{\alpha, \beta} (K(k)^{n-m})_{\alpha_q, \beta} \cE_\beta(k', k)^* c^*_{\alpha}(k') \xi_t \Big\rVert
		\Big\lVert \!\!\sum_{\alpha_2 \in \cI_{k}} (K(k)^m)_{\alpha_q, \alpha_2} c_{\alpha_2}(k) \xi_t \Big\rVert\\
	(\II^{(C)}) & := \sum_{k, k'} \frac{\rho_{q, k}}{2}
		\Big\lVert \!\!\!\! \sum_{\substack{\alpha \in \cI_{k'} \\ \beta \in \cI_{k} \cap \cI_{k'}}} \!\!\! K(k')_{\alpha, \beta} (K(k)^{n-m})_{\alpha_q, \beta} c^*_{\alpha}(k') \cE_\beta(k', k)^* \xi_t \Big\rVert
		\Big\lVert \!\!\sum_{\alpha_2 \in \cI_{k}} \! (K(k)^m)_{\alpha_q, \alpha_2} c_{\alpha_2}(k) \xi_t \Big\rVert\\
	(\III^{(C)}) & := \sum_{k, k'} \frac{\rho_{q, k}}{2}
		\Big\lVert \! \sum_{\alpha_1 \in \cI_{k}} \! (K(k)^{n-m})_{\alpha_q, \alpha_1} c_{\alpha_1}(k) \xi_t \Big\rVert
		\Big\lVert \!\!\!\sum_{\substack{\alpha \in \cI_{k'} \\ \beta \in \cI_{k} \cap \cI_{k'}}} \!\!\! K(k')_{\alpha, \beta} (K(k)^m)_{\alpha_q, \beta} c^*_\alpha(k') \cE_\beta(k, k') \xi_t \Big\rVert\\
	(\IV^{(C)}) & := \sum_{k, k'} \frac{\rho_{q, k}}{2}
		\Big\lVert \!\sum_{\alpha_1 \in \cI_{k}} (K(k)^{n-m})_{\alpha_q, \alpha_1}\! c_{\alpha_1}(k) \xi_t \Big\rVert
		\Big\lVert \!\!\!\!\sum_{\substack{\alpha \in \cI_{k'} \\ \beta \in \cI_{k} \cap \cI_{k'}}} \!\!\! K(k')_{\alpha, \beta} (K(k)^m)_{\alpha_q, \beta} \cE_\beta(k, k') c^*_\alpha(k') \xi_t \Big\rVert \,.
\end{align*}
}
In contribution $ (\I^{(C)}) $, the first norm is bounded by \cref{lem:cccEbound} and the second norm  as the second norm of $ (\II^{(B)}) $. We end up with
\begin{equation}
\label{eq:ICestimate}
	(\I^{(C)})
	\le e^{C |t|} \sum_{k, k'} \rho_{q, k} (C \hat{V}_k)^n (C \hat{V}_{k'}) \fn^{-1} M^{-\frac{3}{2}}\;.
\end{equation}
The term $ (\II^{(C)}) $ is bounded analogously and the last two terms are identical to the first two, under the replacement of $ m $ by $n-m $. Thus we have the three bounds
\begin{equation}
\label{eq:IICestimate}
	(\II^{(C)}),\ (\III^{(C)}),\ (\IV^{(C)})
	\le e^{C |t|} \sum_{k, k'} \rho_{q, k} (C \hat{V}_k)^n (C \hat{V}_{k'}) \fn^{-1} M^{-\frac{3}{2}}\;.
\end{equation}
From the definition \cref{eq:dg} of $ \rho_{q, k} $ and from \cref{eq:fn}, and recalling that $ k \in \tilde{\cC}^q $ implies $ \alpha_q \in \cI_{k} $, it becomes clear that
\begin{equation}
\label{eq:dgestimate}
	\rho_{q, k}
	\le {n_{\alpha_q, k}^{-2}}
	\le C \fn^{-2}\;.
\end{equation}
Adding up all contributions renders (independent of the bootstrap assumption)
\begin{equation}
\label{eq:cECestimate}
	| \langle \xi_t, \cE^{(C)}_{n-m, m} \xi_t \rangle |
	\le e^{C |t|} \sum_{k, k'} (C \hat{V}_k)^n (C \hat{V}_{k'}) \fn^{-3} M^{-\frac{3}{2}}\;.
\end{equation}

\noindent \underline{Bounding $ \cE^{(D)} $:} Again, from \cref{eq:ABC,eq:DEF}, using \cref{lem:approxCCR,lem:dg}, we obtain
\begin{align*}
	|\langle \xi_t, \cE^{(D)}_{n-m, m} \xi_t \rangle| & = \Bigg\vert  \sum_{k, k'} \frac{\rho_{q, k}}{2} \sum_{\substack{\alpha \in \cI_{k'} \\ \alpha_1 \in \cI_{k} \\ \beta \in \cI_{k} \cap \cI_{k'}}} K(k')_{\alpha, \beta} (K(k)^{n-m})_{\alpha_q, \alpha_1} (K(k)^m)_{\alpha_q, \beta} \\
		& \qquad\quad \times \langle \xi_t, \big( c_{\alpha_1}(k) \cE_\beta(k, k') c^*_{\alpha}(k') + c_{\alpha_1}(k) c^*_{\alpha}(k') \cE_\beta(k, k') \big) \xi_t \rangle\\
	& \quad +  \sum_{k, k'} \frac{\rho_{q, k}}{2} \sum_{\substack{\alpha \in \cI_{k'} \\ \alpha_2 \in \cI_{k} \\ \beta \in \cI_{k} \cap \cI_{k'}}} K(k')_{\alpha, \beta} (K(k)^{n-m})_{\alpha_q, \beta} (K(k)^m)_{\alpha_q, \alpha_2}\\
		& \qquad\quad \times \langle \xi_t, \big( c^*_\alpha(k') \cE_\beta(k, k') c_{\alpha_2}(k) + \cE_\beta(k, k') c^*_\alpha(k') c_{\alpha_2}(k) \big) \xi_t \rangle\\
	&\quad + \sum_k \rho_{q, k} \sum_{\alpha, \beta \in \cI_{k}} K(k)_{\alpha, \beta} (K(k)^{n-m})_{\alpha_q, \beta} (K(k)^m)_{\alpha_q, \alpha} \langle \xi_t, \cE_\beta(k, k) \xi_t \rangle \Bigg\vert \,.
\end{align*}
By the Cauchy--Schwarz inequality
\begin{align*}
&|\langle \xi_t, \cE^{(D)}_{n-m, m} \xi_t \rangle| \leq (\I^{(D)}) + (\II^{(D)}) + (\III^{(D)}) + (\IV^{(D)}) + (\V^{(D)})
\end{align*}
where
{\allowdisplaybreaks
\begin{align*}
	(\I^{(D)}) & := \sum_{k, k'} \frac{\rho_{q, k}}{2}
		\Big\Vert \sum_{\alpha_1 \in \cI_{k}} (K(k)^{n-m})_{\alpha_q, \alpha_1} c^*_{\alpha_1}(k) \xi_t \Big\Vert
		\Big\Vert \!\!\sum_{\substack{\alpha \in \cI_{k'} \\ \beta \in \cI_{k} \cap \cI_{k'}}} \!\!K(k')_{\alpha, \beta} (K(k)^m)_{\alpha_q, \beta} \cE_\beta(k, k') c^*_\alpha(k') \xi_t \Big\Vert\\
	(\II^{(D)}) & := \sum_{k, k'} \frac{\rho_{q, k}}{2}
		\Big\Vert \!\!\sum_{\alpha_1 \in \cI_{k}} (K(k)^{n-m})_{\alpha_q, \alpha_1} c^*_{\alpha_1}(k) \xi_t \Big\Vert
		\Big\Vert \!\!\!\sum_{\substack{\alpha \in \cI_{k'} \\ \beta \in \cI_{k} \cap \cI_{k'}}} \!\!\!K(k')_{\alpha, \beta} (K(k)^m)_{\alpha_q, \beta} c^*_\alpha(k') \cE_\beta(k, k') \xi_t \Big\Vert\\
	(\III^{(D)}) & := \sum_{k, k'} \frac{\rho_{q, k}}{2}
		\Big\Vert \!\!\sum_{\substack{\alpha \in \cI_{k'} \\ \beta \in \cI_{k} \cap \cI_{k'}}} \!\!\!\! K(k')_{\alpha, \beta} (K(k)^{n-m})_{\alpha_q, \beta} \cE_\beta(k, k')^* c_\alpha(k') \xi_t \Big\Vert
		\Big\Vert \!\!\sum_{\alpha_2 \in \cI_{k}} \!\!(K(k)^m)_{\alpha_q, \alpha_2} c_{\alpha_2}(k) \xi_t \Big\Vert\\
	(\IV^{(D)}) & := \sum_{k, k'} \frac{\rho_{q, k}}{2}
		\Big\Vert \!\!\sum_{\substack{\alpha \in \cI_{k'} \\ \beta \in \cI_{k} \cap \cI_{k'}}} \!\!\!\! K(k')_{\alpha, \beta} (K(k)^{n-m})_{\alpha_q, \beta} c_\alpha(k') \cE_\beta(k, k')^* \xi_t \Big\Vert
		\Big\Vert \!\!\sum_{\alpha_2 \in \cI_{k}}\!\! (K(k)^m)_{\alpha_q, \alpha_2} c_{\alpha_2}(k) \xi_t \Big\Vert \\
	(\V^{(D)}) & := \sum_k \rho_{q, k}
		\Big\Vert \sum_{\beta \in \cI_{k}} (K(k)^{n-m})_{\alpha_q, \beta} (K(k)^{m+1})_{\alpha_q, \beta} \cE_\beta(k, k) \xi_t \Big\Vert \Vert \xi_t \Vert \;.
\end{align*}
}
The first four contributions are exactly bounded as the four contributions of $ \cE^{(C)}_{n-m, m} $. For $ (\V^{(D)}) $, we use $ \rho_{q, k} \le C \fn^{-2} $, $ \Vert \xi_t \Vert = 1 $ and apply the Cauchy--Schwartz inequality:
\begin{align*}
	(\V^{(D)})
	& \le C \sum_k \fn^{-2} \Big\Vert \sum_{\beta \in \cI_{k}} (K(k)^{n-m})_{\alpha_q, \beta} (K(k)^{m+1})_{\alpha_q, \beta} \cE_\beta(k, k) \xi_t \Big\Vert\\
	&\le C \sum_k \fn^{-2} \left( \sum_{\beta \in \cI_{k}} \left\vert (K(k)^{n-m})_{\alpha_q, \beta} (K(k)^{m+1})_{\alpha_q, \beta} \right\vert^2 \right)^{\frac{1}{2}}
		\left( \sum_{\beta \in \cI_{k}} \Vert \cE_\beta(k, k) \xi_t \Vert^2 \right)^{\frac{1}{2}}\\
	&\le \sum_k \fn^{-2} \left( \sum_{\beta \in \cI_{k}} (C \hat{V}_k)^{2n+2} M^{-4} \right)^{\frac{1}{2}}
		\left( \fn^{-2} \Vert \cN^{\frac{1}{2}} \xi_t \Vert^2 \right)^{\frac{1}{2}}\\
	&\le e^{C |t|} \sum_k (C \hat{V}_k)^{n+1} \fn^{-3} M^{- \frac 32}
	\hspace{4em} \le e^{C |t|} \sum_{k, k'} (C \hat{V}_k)^n (C \hat{V}_{k'}) \fn^{-3} M^{- \frac{3}{2}}
\end{align*}
where we used \cref{lem:Knbound,lem:cEbound} in the third line. Thus
\begin{equation}
\label{eq:cEDestimate}
	| \langle \xi_t, \cE^{(D)}_{n-m, m} \xi_t \rangle |
	\le e^{C |t|} \sum_{k, k'} (C \hat{V}_k)^n (C \hat{V}_{k'}) \fn^{-3} M^{- \frac{3}{2}}\;.
\end{equation}

\noindent \underline{Bounding $ \cE^{(E)} $:} Finally, again from \cref{eq:ABC,eq:DEF}, using \cref{lem:approxCCR,lem:dg}, we obtain
\begin{align*}
	& |\langle \xi_t, \cE^{(E)}_n \xi_t \rangle| = \\
	& \Bigg\vert \frac{1}{2} \sum_{k, k'} \sum_{\substack{\alpha \in \cI_{k'} \\ \alpha_1 \in \cI_{k}}} \! K(k')_{\alpha, \alpha_q} (K(k)^n)_{\alpha_q, \alpha_1}  \langle \xi_t, \big( c_\alpha(k') \cE^{(g)}_q(k', k) c^*_{\alpha_1}(k) + \cE^{(g)}_q(k', k) c_\alpha(k') c^*_{\alpha_1}(k) \big) \xi_t \rangle\\
	& + \frac{1}{2} \sum_{k, k'} \! \!\!\sum_{\substack{\alpha \in \cI_{k'} \\ \beta \in \cI_{k} \cap \cI_{k'}}} \!\!\! K(k')_{\alpha, \beta} (K(k)^n)_{\alpha_q, \beta}  \langle \xi_t, \big( c^*(g_{q, k}) c_\alpha(k') \cE_\beta(k', k) + c^*(g_{q, k}) \cE_\beta(k', k) c_\alpha(k') \big) \xi_t \rangle\\
	& + \sum_k \rho_{q, k} \sum_{\beta \in \cI_{k}}  K(k)_{\alpha_q, \beta} (K(k)^n)_{\alpha_q, \beta} \langle \xi_t, \cE_\beta(k, k) \xi_t \rangle \Bigg\vert \;.
\end{align*}
By the Cauchy--Schwarz inequality
\begin{align*}
&|\langle \xi_t, \cE^{(E)}_n \xi_t \rangle| \leq (\I^{(E)}) + (\II^{(E)}) + (\III^{(E)}) + (\IV^{(E)}) + (\V^{(E)})
\end{align*}
where
{\allowdisplaybreaks
\begin{align*}
	(\I^{(E)}) & := \frac{1}{2} \sum_{k, k'}
		\Big\Vert \sum_{\alpha \in \cI_{k'}} K(k')_{\alpha, \alpha_q} c^*_\alpha(k') \xi_t \Big\Vert
		\Big\Vert \sum_{\alpha_1 \in \cI_{k}} (K(k)^n)_{\alpha_q, \alpha_1} \cE^{(g)}_q(k', k) c^*_{\alpha_1}(k) \xi_t \Big\Vert \\
	(\II^{(E)}) & := \frac{1}{2} \sum_{k, k'}
		\Big\Vert \sum_{\alpha \in \cI_{k'}} K(k')_{\alpha, \alpha_q} c^*_\alpha(k') \cE^{(g)}_q(k', k)^* \xi_t \Big\Vert
		\Big\Vert \sum_{\alpha_1 \in \cI_{k}} (K(k)^n)_{\alpha_q, \alpha_1} c^*_{\alpha_1}(k) \xi_t \Big\Vert \\
	(\III^{(E)}) & := \frac{1}{2} \sum_{k, k'}
		\Vert c(g_{q, k}) \xi_t \Vert
		\Big\Vert \sum_{\substack{\alpha \in \cI_{k'} \\ \beta \in \cI_{k} \cap \cI_{k'}}} K(k')_{\alpha, \beta} (K(k)^n)_{\alpha_q, \beta} c_{\alpha}(k') \cE_\beta(k', k) \xi_t \Big\Vert \\
	(\IV^{(E)}) & := \frac{1}{2} \sum_{k, k'}
		\Vert c(g_{q, k}) \xi_t \Vert
		\Big\Vert \sum_{\substack{\alpha \in \cI_{k'} \\ \beta \in \cI_{k} \cap \cI_{k'}}} K(k')_{\alpha, \beta} (K(k)^n)_{\alpha_q, \beta} \cE_\beta(k', k) c_{\alpha}(k') \xi_t \Big\Vert \\
	(\V^{(E)}) & := \sum_k \rho_{q, k}
		\Big\Vert \sum_{\beta \in \cI_{k}} K(k)_{\alpha_q, \beta} (K(k)^n)_{\alpha_q, \beta} \cE_\beta(k, k) \xi_t \Big\Vert \Vert \xi_t \Vert \;.
\end{align*}
}
The first four contributions are bounded like the four contributions of $ \cE^{(B)}_n $. The fifth contribution is bounded by the same steps as $ (\V^{(D)}) $:
\[
\begin{aligned}
	(\V^{(E)})
	& \le C \sum_k \fn^{-2} \bigg( \sum_{\beta \in \cI_{k}} \left\vert K(k)_{\alpha_q, \beta} (K(k)^n)_{\alpha_q, \beta} \right\vert^2 \bigg)^{\frac{1}{2}}
		\bigg( \sum_{\beta \in \cI_{k}} \Vert \cE_\beta(k, k) \xi_t \Vert^2 \bigg)^{\frac{1}{2}}\\
	&\le \sum_k (C \hat{V}_k)^{n+1} \fn^{-2} M^{- \frac 32}
		\left( \fn^{-2} \Vert \cN^{\frac{1}{2}} \xi_t \Vert^2 \right)^{\frac{1}{2}}
	\le e^{C |t|} \sum_{k, k'} (C \hat{V}_k)^n (C \hat{V}_{k'}) \fn^{-3} M^{- \frac{3}{2}}\;.
\end{aligned}
\]
Recalling $ \fn = N^{\frac 13 - \frac{\delta}{2}} M^{- \frac 12} $, the final bound for $ | \langle \xi_t, \cE^{(E)}_n \xi_t \rangle | $ (and its adjoint) is thus
\begin{equation}
\label{eq:cEEestimate}
	| \langle \xi_t, \cE^{(E)}_n \xi_t \rangle |
	\le e^{C |t|} \sum_{k, k'} (C \hat{V}_k)^n (C \hat{V}_{k'}) \fn^{-2} M^{-1} N^{- \frac r2 + \frac{\delta}{2}} \;.
\end{equation}

\noindent \underline{Bounding $ \cE^{(F)} $:}
Since $ S^* = -S $, $ \bF_{m, m'}^* = \bD_{m', m} $, $ \bA_n^* = \bA_n$, and $\bC_{m, m'}^* = \bC_{m', m} $, we obtain from \eqref{eq:cEABCDEF} that  $ (\cE^{(F)}_{n-m, m})^* = \cE^{(D)}_{m, n-m} $, so
\begin{equation}
\label{eq:cEFestimate}
	| \langle \xi_t, \cE^{(F)}_{n-m, m} \xi_t \rangle |
	= | \langle \xi_t, \cE^{(D)}_{m, n-m} \xi_t \rangle |
	\le e^{C |t|} \sum_{k, k'} (C \hat{V}_k)^n (C \hat{V}_{k'}) \fn^{-3} M^{- \frac{3}{2}}\;.
\end{equation}

\noindent \underline{Summing up the bounds:}
Consider again \cref{eq:cE2nplus1}. If $ n $ is even, then there are $ 1 + 1 + \sum_{m = 1}^{n-1} \binom{n}{m} = \sum_{m = 0}^n \binom{n}{m} = 2^n $ terms to bound. Since $ r \le \frac{2}{3}$, and $ \sum_k $ runs over finitely many elements, we have, for some fixed $\fC_1 > 0$, the bound
\begin{equation}
\label{eq:cEn2boundfinal}
	| \langle \xi_t, \cE_{n+1, q} \xi_t \rangle |
	 \le 2^n e^{C |t|} \sum_{k, k'} (C \hat{V}_k)^n \hat{V}_{k'} \fn^{-2} M^{-1} N^{- \frac r2 + \frac{\delta}{2}}
	 \le \fC_1^{n+1} e^{C |t|}  N^{- \frac{2}{3} + \frac{3}{2} \delta - \frac{r}{2}}\;.
\end{equation}
This is \cref{eq:cE2estimate}, which we wanted to prove.
If $ n = 2s + 1 $ is odd, then the number of occurring terms is $ 1 + 1 + 2 \sum_{m = 1}^s \binom{n}{m} = \sum_{m = 0}^n \binom{n}{m} = 2^n $, and \cref{eq:cEn2boundfinal} still holds.
\end{proof}

The second bound required to prove Lemma \ref{lem:bootstrap} is the following.

\begin{lemma}[Bound on the Bosonized Commutator]
\label{lem:adnestimate}
Suppose that $ \hat{V}$ is non-negative and compactly supported.
Suppose there is $ r \ge 0 $ and $C > 0$ such that for all $ t \in [-1, 1] $ and all $ q' \in \ZZZ^3 $ we have
\begin{equation}
	\langle \xi_t, a_{q'}^* a_{q'} \xi_t \rangle
	\le C N^{-r}\;.
\end{equation}
Then there exists $ \fC_2 > 0 $ such that\footnote{We use $ \fC_2 $ instead of $ C $ to explicitly track the $ n $-dependence of the constants. This will be important for ensuring that sums over $ n $ converge.} for all $ n, N \in \NNN \setminus \{0\}$, $q \in \ZZZ^3 $ and $ t \in [-1, 1] $
\begin{equation}
\label{eq:adnestimate}
	| \langle \xi_t, \ad^n_{q, (\b)} \xi_t \rangle |
	\le \fC_2^n e^{C |t|} \;.
\end{equation}
\end{lemma}
	It becomes clear from the proof that we could also obtain the stronger bound
	\begin{equation}
	| \langle \xi_t, \ad^n_{q, (\b)} \xi_t \rangle |
	\le \fC_2^n e^{C |t|} \max\{ N^{-\frac{1}{3} - \frac{r}{2} + \frac{\delta}{2}}, N^{-\frac{2}{3}+ \delta} \}\;.
	\end{equation}
	However, for applying \cref{lem:adnestimate} in \cref{eq:errorboundbootstrap}, any bound not growing with $N$ is sufficient.

\begin{proof}
Recall \cref{eq:adnb}. If $ q $ is not inside any patch, then $ \ad^n_{q, (\b)} = 0 $, so the statement is trivially satisfied. We may therefore assume that $ q \in B_{\alpha_q} $ for some $ 1 \le \alpha_q \le M $. We use the convention $ \sum_k = \sum_{k \in \tilde{\cC}^q \cap \ZZZ^3} $. By \cref{lem:Knbound} and by \cref{eq:dgestimate} we obtain
\begin{equation}
\label{eq:Abound}
	| \langle \xi_t, \bA_n \xi_t \rangle |
	= \sum_k (K(k)^n)_{\alpha_q, \alpha_q} \rho_{q, k}
	\le \sum_k (C \hat{V}_k)^n M^{-1} \fn^{-2}
	= \sum_k (C \hat{V}_k)^n N^{-\frac{2}{3}+ \delta} \;.
\end{equation}

For $ \bB_n $ we use \cref{lem:ccbound,lem:ccgbound} to get
\begin{equation}
\begin{aligned}
	| \langle \xi_t, \bB_n \xi_t \rangle |
	& \le \sum_k \Vert c(g_{q, k}) \xi_t \Vert
		\Big\Vert \sum_{\alpha_1 \in \cI_{k}} (K(k)^n)_{\alpha_q, \alpha_1} c_{\alpha_1}(k) \xi_t \Big\Vert \\
	& \le \sum_k \frac{C}{n_{\alpha_q, k}} N^{-\frac r2}
		\Vert (K(k)^n)_{\alpha_q, \cdot} \Vert_2
		\Vert \cN^{\frac{1}{2}} \xi_t \Vert\;.
\end{aligned}
\end{equation}
Now, by Lemma \ref{lem:fn}, we have $ n_{\alpha_q, k} \ge C \fn $ (recall that $k \in \tilde{\Ccal}^q$ implies $\alpha_k \in \Ik$), and \cref{eq:C01estimate} allows us to bound $ \Vert \cN^{\frac{1}{2}} \xi_t \Vert \le e^{C |t|} $, so together with \cref{eq:Kncontrol} we obtain
\begin{equation}
\label{eq:Bbound}
	| \langle \xi_t, \bB_n \xi_t \rangle |
	\le e^{C |t|} \sum_k (C \hat{V}_k)^n \fn^{-1} N^{-\frac r2} M^{-\frac{1}{2}}
	= e^{C |t|} \sum_k (C \hat{V}_k)^n N^{-\frac{1}{3} - \frac{r}{2} +\frac{\delta}{2}}\;.
\end{equation}
The same holds for the adjoint.
%
The bound on $ \bC_{n-m, m} $ is obtained by the same steps, together with \cref{eq:dgestimate}:
\begin{align*}
	|\langle \xi_t, \bC_{n-m, m} \xi_t \rangle|
	& \le \sum_k \rho_{q, k} \Big\Vert \sum_{\alpha_1 \in \cI_{k}} (K(k)^{n-m})_{\alpha_q, \alpha_1} c_{\alpha_1}(k) \xi_t \Big\Vert
		\Big\Vert \sum_{\alpha_2 \in \cI_{k}} (K(k)^m)_{\alpha_q, \alpha_2} c_{\alpha_2}(k) \xi_t \Big\Vert \\
	& \le \sum_k \rho_{q, k}
		\Vert (K(k)^{n-m})_{\alpha_q, \cdot} \Vert_2
		\Vert \cN^{\frac{1}{2}} \xi_t \Vert
		\Vert (K(k)^m)_{\alpha_q, \cdot} \Vert_2
		\Vert \cN^{\frac{1}{2}} \xi_t \Vert\\
	& \le e^{C |t|} \sum_k (C \hat{V}_k)^n M^{-1} \fn^{-2}
	\qquad = e^{C |t|} \sum_k (C \hat{V}_k)^n N^{-\frac{2}{3} + \delta}\;. \tagg{eq:Cbound}
\end{align*}
For $ \bD_{n-m, m} $, since $ \bD_{n-m, m} $ is obtained from $ \bC_{n-m, m} $ by replacing $ c_{\alpha_1}(k)$ by $ c^*_{\alpha_1}(k) $, so we only need to replace $ \Vert \cN^{\frac{1}{2}} \xi_t \Vert^2 \le e^{C |t|} $ by $ \Vert (\cN + 1)^{\frac{1}{2}} \xi_t \Vert^2 \le e^{C |t|} $, yielding
\begin{equation}
\label{eq:Dbound}
	| \langle \xi_t, \bD_{n-m, m} \xi_t \rangle |
	\le e^{C |t|} \sum_k (C \hat{V}_k)^n N^{-\frac{2}{3}+ \delta} \;.
\end{equation}
Since $ \bF_{n-m, m}^* = \bD_{m, n-m} $ the same bound also applies to $ | \langle \xi_t, \bF_{m, n-m} \xi_t \rangle | $.

Finally, $ \bE_n $ is obtained from $ \bB_n $ by replacing $ c_{\alpha_1}(k)$ by $c^*_{\alpha_1}(k) $. So again
\begin{equation}
\label{eq:Ebound}
	| \langle \xi_t, \bE_n \xi_t \rangle |
	\le e^{C |t|} \sum_k (C \hat{V}_k)^n N^{-\frac{1}{3} - \frac{r}{2} + \frac{\delta}{2}}\;.
\end{equation}
The total number of terms involved for even $ n $ is now $ 2^{n-1} + 1 + 1 + \sum_{m = 1}^{n-1} \binom{n}{m} = 2^{n-1} + 2^n < 2^{n+1} $, while for odd $ n = 2s+1 $, it is $ 1 + 1 + \sum_{m = 1}^s \binom{n}{m} + \sum_{m = 1}^s \binom{n}{m} = 2^n $. So in any case, with $ \sum_k $ running over finitely many elements, we have
\[
	| \langle \xi_t, \ad^n_{q, (\b)} \xi_t \rangle |
	 \le e^{C |t|} 2 \sum_k (C \hat{V}_k)^n \max\{ N^{-\frac 13 - \frac r2 + \frac{\delta}{2}}, N^{-\frac{2}{3}+ \delta} \} \;.
\]
As the exponents on $ N $ are negative, this implies \eqref{eq:adnestimate}.
\end{proof}

\section{Proof of \cref{thm:main2}}
\label{sec:proofmain2}

In this section we prove that the bosonized excitation density $ n_q^{(\b)} $ is a good approximation for the true excitation density $ n_q $. The proof employs the bootstrap.

\begin{lemma}[Bootstrap Step]
\label{lem:bootstrap}
Suppose that $ \hat{V}$ is non-negative and compactly supported.
Suppose there is $ r \in [0, \frac{2}{3}) $ and $C > 0$ such that for all $ t \in [-1, 1] $ and $ q' \in \ZZZ^3 $ we have
\begin{equation}
\label{eq:nqcondition4}
	\langle \xi_t, a_{q'}^* a_{q'} \xi_t \rangle
	\le C N^{-r}\;.
\end{equation}
Then, for all $ t \in [-1, 1] $ and all $q \in \Zbb^3$ we have
\begin{equation}
\label{eq:bootstrap}
	\bigg\lvert \langle \xi_t, a_q^* a_q \xi_t \rangle -
	\sum_{n = 0}^\infty \frac{t^n}{n!} \langle \Omega, \ad^n_{q, (\b)} \Omega \rangle \bigg\rvert
	\le C e^{C |t|} N^{- \frac{2}{3}+ \frac{3}{2} \delta - \frac r2}\;.
\end{equation}
\end{lemma}
\begin{proof}
If $ q $ is not inside any patch the bound is trivial since then $ \ad^n_{q, (\b)} = 0 $ and $  \langle \xi_t, a_q^* a_q \xi_t \rangle = 0 $. So we may assume that $ q \in B_{\alpha_q} $ for some $ 1 \le \alpha_q \le M $. Without loss of generality $ t \ge 0 $. We use the abbreviations $ \langle A \rangle_\psi := \langle \psi, A \psi \rangle $, $ \ad^n_q := \ad^n_S(a_q^* a_q) $, and denote the scaled $ n $--dimensional simplex by $ t \Delta^{(n)} := \{ (t_1, \ldots, t_n) : 0 \le t_n \le t_{n-1} \le \ldots \le t_2 \le t_1 \le t \} $.
Recall that $ \cE_{n, q} = [S, \ad^{n-1}_{q, (\b)}] - \ad^n_{q, (\b)} $. We expand
\begin{equation}
\label{eq:Duhamelexpansion}
\begin{aligned}
	e^{tS} a_q^* a_q e^{-tS}
	= &\sum_{n = 0}^{n_*} \frac{t^n}{n!} \ad^n_{q, (\b)} + \sum_{n = 1}^{n_*} \int_{t \Delta^{(n)}} \d t_1 \ldots \d t_n \; e^{t_n S} \cE_{n, q} e^{-t_n S}\\
	&+ \int_{t \Delta^{(n_*+1)}} \d t_1 \ldots \d t_{n_*+1} \; e^{t_{n_*+1} S} [S, \ad^{n_*}_{q, (\b)}] e^{-t_{n_*+1} S}\;.
\end{aligned}
\end{equation} 
This expansion can be checked inductively: The case $ n_* = 0 $ is just the Duhamel formula $ e^{tS} B e^{-tS} = B + \int_0^t \d t_1 \; e^{t_1 S} [S, B] e^{-t_1 S} $ with $ B = \ad^0_{q, (\b)} = a_q^* a_q $. For the induction step from $ n_*$ to $n_*+1 $ we write $ [S, \ad^{n_*}_{q, (\b)}] = \ad^{n_*+1}_{q, (\b)} + \cE_{n_*+1, q} $. Then, we apply the Duhamel formula to $ \ad^{n_*+1}_{q, (\b)} $ and use $ \int_{t \Delta^{(n_*+1)}} \d t_1 \ldots \d t_n \; \ad^{n_*+1}_{q, (\b)} = \frac{ t^{n_*+1} }{(n_*+1)!} \ad^{n_*+1}_{q, (\b)} $. Now, using $ \xi_t = e^{-t S} \Omega $, this expansion renders
\begin{equation}
\begin{aligned}
	&\langle \xi_t, a_q^* a_q \xi_t \rangle
	- \sum_{n = 0}^\infty \frac{t^n}{n!} \langle \Omega, \ad^n_{q, (\b)} \Omega \rangle \\
	& = -\sum_{n = n_*+1}^\infty \frac{t^n}{n!} \langle \ad^n_{q, (\b)} \rangle_\Omega
	+ \sum_{n = 1}^{n_*} \int_{t \Delta^{(n)}} \d t_1 \ldots \d t_n \; \langle \cE_{n, q} \rangle_{\xi_{t_n}}\\
	& \quad + \int_{t \Delta^{(n_*+1)}} \d t_1 \ldots \d t_{n_*+1} \; \left( \langle \ad^{n_*+1}_{q, (\b)} \rangle_{\xi_{t_{n_*+1}}}
	+ \langle \cE_{n_*+1, q} \rangle_{\xi_{t_{n_*+1}}} \right)\;.
\end{aligned}
\end{equation}
Applying Lemmas \ref{lem:cE2estimate} and \ref{lem:adnestimate}, as well as $ \int_{t \Delta^{(n)}} \d t_1 \ldots \d t_n \; 1 = \frac{t^n}{n!} $, we  obtain
\begin{align*}
	&\Big\vert \langle \xi_t, a_q^* a_q \xi_t \rangle
	- \sum_{n = 0}^\infty \frac{t^n}{n!} \langle \Omega, \ad^n_{q, (\b)} \Omega \rangle \Big\vert \\
	& \le \sum_{n = n_*+1}^\infty \frac{(t \fC_2)^n}{n!}
	+ e^{C t} \sum_{n = 0}^{n_*} \frac{(t \fC_1)^n}{n!} N^{- \frac{2}{3}+ \frac{3}{2} \delta - \frac r2}  + e^{C t}\frac{  (t \fC_2)^{n_*+1} + (t \fC_1)^{n_*+1} N^{- \frac{2}{3}+ \frac{3}{2} \delta - \frac r2} }{(n_*+1)!}\\
	& \le e^{C t} N^{- \frac{2}{3}+ \frac{3}{2} \delta - \frac r2} + \left( \sum_{n = n_*+1}^\infty \frac{(t \fC_2)^n}{n!}
	+ e^{C t} \frac{ (t \fC_2)^{n_*+1}}{(n_*+1)!} \right)\;. \tagg{eq:errorboundbootstrap}
\end{align*}
The last line vanishes as $ n_* \to \infty $, completing the proof.
\end{proof}

\begin{proof}[Proof of Theorem \ref{thm:main2}]
For $ t \in [-1,1] $, we introduce
\begin{equation}
	n_{q;t} := \langle \Omega, e^{t S} a_q^* a_q e^{-t S} \Omega \rangle
	= \langle \xi_t, a_q^* a_q \xi_t \rangle \;, \qquad
	n_{q;t}^{(\b)} := \sum_{n = 0}^\infty \frac{t^n}{n!} \langle \Omega, \ad^n_{q, (\b)} \Omega \rangle \;,
\end{equation}
so $ n_q = n_{q;1} $. First, note that the statement in Theorem \ref{thm:main2} is trivial if $ q $ is not inside any patch, since then $ n_q = n_q^{(\b)} = 0 $. So we can assume that $ q \in B_{\alpha_q} $ for some $ 1 \le \alpha_q \le M $. \cref{lem:bootstrap} (the bootstrap step) provides us with
\begin{equation}
\label{eq:diff}
	| n_{q;t} - n_{q;t}^{(\b)} |
	\le C N^{- \frac{2}{3}+ \frac{3}{2} \delta - \frac{r}{2}} \;,
\end{equation}
whenever the bootstrap assumption $ \langle \xi_t, a_{q'}^* a_{q'} \xi_t \rangle \leq C N^{-r} $ holds for all $ q' \in \ZZZ^3 $.

\smallskip

By $ \Vert a_{q'}^\sharp \Vert_{\op} = 1 $, it is obvious that $ \langle \xi_t, a_{q'}^* a_{q'} \xi_t \rangle \le 1 $, so the bootstrap assumption is initially fulfilled for $ r = 0 $.

\smallskip

Using \cref{eq:nqbformula} and applying the estimates \cref{eq:Knbound} and \cref{eq:dgestimate} we obtain
\begin{equation}
\label{eq:nqbscaling}
\begin{aligned}
	 n_{q;t}^{(\b)}
	& = \sum_{m = 1}^\infty \frac{2^{2m-1}}{(2m)!} \sum_{k \in \tilde{\cC}^q \cap \ZZZ^3} \big( t^{2m} K(k)^{2m} \big)_{\alpha_q, \alpha_q} \rho_{q, k}\\
	& \le \sum_{m = 1}^\infty \sum_{k \in \tilde{\cC}^q \cap \ZZZ^3} \frac{2^{2m-1}}{(2m)!} (C \hat{V}_k)^{2m} M^{-1} \fn^{-2}
	= \cO(N^{-\frac{2}{3} + \delta})\;.
\end{aligned}
\end{equation}
Thus
\begin{equation}
n_{q;t} = \lvert n_{q;t} - n_{q;t}^\textnormal{(b)} \rvert + n_{q;t}^\textnormal{(b)} \leq C N^{- \frac{2}{3}+ \frac{3}{2} \delta - \frac r2} + C N^{-\frac{2}{3} + \delta} \;.
 \label{eq:it}
\end{equation}
We now do two bootstrap steps: First, plugging $ r=0 $ into \eqref{eq:it} for any $ q' = q $, we see that \eqref{eq:nqcondition4} holds with $ r = \frac 23 - \frac 32 \delta $. Second, using this improved $ r $ in \eqref{eq:it} yields $n_{q;t} \leq C N^{-\frac{2}{3}+\delta}$, with optimal coefficient $ r = \frac{2}{3} - \delta $. Using this again in \cref{eq:diff} yields
\[
	\left\vert n_q - n_q^{(\b)} \right\vert 
	\leq C N^{-1 + 2 \delta} \;.	\qedhere
\]
\end{proof}

\section{Proof of \cref{prop:main3}}
\label{sec:proofmain3}

\begin{proof}[Proof of \cref{prop:main3}]
Starting from \cref{eq:nqb}, we have to compute the diagonal matrix elements of $\cosh(2 K(k)) - 1$. Recall the definition of $S_1(k)$ from \cref{eq:S1}. We use the identities \cite[(7.4)]{BPSS22}, where in all the proof we suppress the $k$--dependence:
\begin{equation}
	\cosh(K) = \frac{|S_1^T| + |S_1^T|^{-1}}{2} \;, \qquad
	\sinh(K) = \frac{|S_1^T| - |S_1^T|^{-1}}{2} \;,
\end{equation}
This yields
\begin{equation}
\label{eq:cosh2K-1}
	\cosh(2 K) - 1 = 2 \sinh(K)^2 = \frac{1}{2}(|S_1^T|^2 - 2 + |S_1^T|^{-2}) \;.
\end{equation}
The $ |\cI| \times |\cI| $--matrices $ |S_1^T|^2 $ and $ |S_1^T|^{-2} $ can be diagonalized following the steps in \cite[(7.7)]{BPSS22} and thereafter: We introduce
\begin{equation}
\label{eq:U}
	U := \frac{1}{\sqrt{2}} \begin{pmatrix}
		\1 & \1\\
		\1 & -\1\\
	\end{pmatrix} \;,
\end{equation}
with $ \1 $ being the $ |\cI^+| \times |\cI^+| $ identity matrix, and  obtain
\begin{align*}
	U^T |S_1^T|^2 U
	&= \begin{pmatrix}
		d^{\frac{1}{2}} (d^{\frac{1}{2}} (d+2b) d^{\frac{1}{2}})^{-\frac{1}{2}} d^{\frac{1}{2}} & 0\\
		0 & (d+2b)^{\frac{1}{2}} ((d+2b)^{\frac{1}{2}} d (d+2b)^{\frac{1}{2}})^{-\frac{1}{2}} (d+2b)^{\frac{1}{2}}\\
	\end{pmatrix}\\
	U^T |S_1^T|^{-2} U
	&= \begin{pmatrix}
		d^{-\frac{1}{2}} (d^{\frac{1}{2}} (d+2b) d^{\frac{1}{2}})^{\frac{1}{2}} d^{-\frac{1}{2}} & 0\\
		0 & (d+2b)^{-\frac{1}{2}} ((d+2b)^{\frac{1}{2}} d (d+2b)^{\frac{1}{2}})^{\frac{1}{2}} (d+2b)^{-\frac{1}{2}}\\
	\end{pmatrix}
\end{align*}
with the matrices $d$ and $b$ defined in \cref{eq:db}.
The diagonal matrix element of \cref{eq:cosh2K-1} is
\begin{equation}
\label{eq:ABCD}
\begin{aligned}
	(\cosh(2 K) - 1)_{\alpha_q, \alpha_q}
	& = \frac{1}{4} \Big( -4
	+ \langle \alpha_q| d^{\frac{1}{2}} (d^{\frac{1}{2}} (d+2b) d^{\frac{1}{2}})^{-\frac{1}{2}} d^{\frac{1}{2}} | \alpha_q \rangle\\
	&\qquad + \langle \alpha_q| (d+2b)^{\frac{1}{2}} ((d+2b)^{\frac{1}{2}} d (d+2b)^{\frac{1}{2}})^{-\frac{1}{2}} (d+2b)^{\frac{1}{2}} | \alpha_q \rangle\\
	&\qquad + \langle \alpha_q| d^{-\frac{1}{2}} (d^{\frac{1}{2}} (d+2b) d^{\frac{1}{2}})^{\frac{1}{2}} d^{-\frac{1}{2}} | \alpha_q \rangle\\
	&\qquad + \langle \alpha_q| (d+2b)^{-\frac{1}{2}} ((d+2b)^{\frac{1}{2}} d (d+2b)^{\frac{1}{2}})^{\frac{1}{2}} (d+2b)^{-\frac{1}{2}} | \alpha_q \rangle \Big)\\
	& =: \frac{1}{4} \big( -4 + \textnormal{(A)} + \textnormal{(B)} + \textnormal{(C)} + \textnormal{(D)} \big) \;,
\end{aligned}
\end{equation}
with $ |\alpha_q \rangle \in \CCC^{|\cI^+|} $ denoting the canonical basis vector corresponding to the patch $ B_{\alpha_q} $ if $ \alpha_q \in \cI^+ $ or to the patch opposite to $ B_{\alpha_q} $ if $ \alpha_q \in \cI^- $. Using \cref{eq:abbreviations2} we can write
\begin{equation}
\label{eq:dbabbreviation}
\begin{aligned}
	d &= \sum_{\alpha \in \cI^+} \lambda_\alpha |\alpha \rangle \langle \alpha |\;,
	\qquad \textnormal{therefore} \quad  d^s |\alpha \rangle = \lambda_\alpha^s |\alpha \rangle \quad \forall s \in \RRR \;;\\
	b &= g |n \rangle \langle n | \qquad \text{with} \qquad
	|n \rangle = \sum_{\alpha \in \cI^+} n_{\alpha} |\alpha \rangle \quad \textnormal{and} \quad g := \frac{\hat{V}_k}{2 \hbar \kappa N |k|} \;.\\
\end{aligned}
\end{equation}
To evaluate $ \textnormal{(A)}$--$ \textnormal{(D)}$, we use the identities (valid for any symmetric matrix $ A $)
\begin{equation}
\label{eq:integralidentities}
	A^{\frac{1}{2}} = \frac{2}{\pi} \int_0^\infty \left( 1 - \frac{\mu^2}{A + \mu^2} \right) \; \d \mu \;, \qquad
	A^{-\frac{1}{2}} = \frac{2}{\pi} \int_0^\infty \frac{1}{A + \mu^2} \; \d \mu \;,
\end{equation}
and the Sherman--Morrison formula, for all vectors $v,w \in \Cbb^{\lvert \Ical^+\rvert}$,
\begin{equation}
	(A + |v \rangle \langle w|)^{-1}
	= A^{-1} - \frac{A^{-1} |v \rangle \langle w| A^{-1}}{1 + \langle w | A^{-1} | v \rangle} \;.
\end{equation}

\smallskip

\noindent \underline{Evaluation of $\textnormal{(A)}$:} Applying \cref{eq:dbabbreviation} we get
\begin{equation}
	\textnormal{(A)}
	= \lambda_{\alpha_q} \langle \alpha_q| (d^{\frac{1}{2}} (d+2b) d^{\frac{1}{2}})^{-\frac{1}{2}} | \alpha_q \rangle
	= \lambda_{\alpha_q} \langle \alpha_q| (d^2 + 2 \tilde{b})^{-\frac{1}{2}} | \alpha_q \rangle
\end{equation}
with the rank-one operator
\begin{equation}
	\tilde{b}
	:= d^{\frac{1}{2}} b d^{\frac{1}{2}}
	= g d^{\frac{1}{2}} | n \rangle \langle n | d^{\frac{1}{2}}
	= g |\tilde{n} \rangle \langle \tilde{n}|
	\quad \text{where} \quad
	| \tilde{n} \rangle
	:= d^{\frac{1}{2}} | n \rangle
	= \sum_{\alpha \in \cI^+} \lambda_\alpha^{\frac{1}{2}} n_{\alpha} | \alpha \rangle\,.
\end{equation}
Applying the integral identities \cref{eq:integralidentities} and the Sherman--Morrison formula, we obtain
\begin{equation}
\begin{aligned}
	(d^2 + 2 \tilde{b})^{-\frac{1}{2}}
	& = \frac{2}{\pi} \int_0^\infty \; \frac{\d \mu}{\mu^2 + d^2 + 2 \tilde{b}}\\
	& = d^{-1} - \frac{2}{\pi} \int_0^\infty \frac{2 g (\mu^2 + d^2)^{-1} |\tilde{n} \rangle \langle \tilde{n}| (\mu^2 + d^2)^{-1}}{1 + 2 g \langle \tilde{n}| (\mu^2 + d^2)^{-1} | \tilde{n} \rangle} \; \d \mu \;.
\end{aligned}
\end{equation}
Taking the expectation value with $ | \alpha_q \rangle $ and multiplying by $ \lambda_{\alpha_q} $  renders
\begin{equation}
\label{eq:Afinalresult}
\begin{aligned}
	\textnormal{(A)}
	& = \lambda_{\alpha_q} \langle \alpha_q| d^{-1} |\alpha_q \rangle
	- \frac{2}{\pi} \int_0^\infty \frac{2 g \lambda_{\alpha_q} \big| \langle \alpha_q | (\mu^2 + d^2)^{-1} | \tilde{n} \rangle \big|^2}{1 + 2 g \langle \tilde{n} | (\mu^2 + d^2)^{-1} | \tilde{n} \rangle} \; \d \mu\\
	& = 1
	- \frac{2}{\pi} \int_0^\infty \frac{2 g n_{\alpha_q}^2 \lambda_{\alpha_q}^2 (\mu^2 + \lambda_{\alpha_q}^2)^{-2}}{1 + 2 g \sum_{\alpha \in \cI^+} n_{\alpha}^2 (\mu^2 + \lambda_\alpha^2)^{-1} \lambda_\alpha} \; \d \mu\;.
\end{aligned}
\end{equation}

\noindent \underline{Evaluation of $ \textnormal{(B)} $:} By means of the integral identities \cref{eq:integralidentities} we get
\begin{equation}
\begin{aligned}
	&\left((d+2b)^{\frac{1}{2}} d (d+2b)^{\frac{1}{2}}\right)^{-\frac{1}{2}}
	= \frac{2}{\pi} \int_0^\infty \frac{\d \mu}{(d+2b)^{\frac{1}{2}} d (d+2b)^{\frac{1}{2}} + \mu^2}\\
	& = \frac{2}{\pi} \int_0^\infty \left( (d+2b)^{\frac{1}{2}} \left( d + \mu^2 (d+2b)^{-1} \right) (d+2b)^{\frac{1}{2}} \right)^{-1} \d \mu \;.
\end{aligned}
\end{equation}
Thus
\begin{equation}
\label{eq:Bintegral}
\begin{aligned}
	\textnormal{(B)}
	& = \langle \alpha_q | (d+2b)^{\frac{1}{2}} \left((d+2b)^{\frac{1}{2}} d (d+2b)^{\frac{1}{2}}\right)^{-\frac{1}{2}} (d+2b)^{\frac{1}{2}} | \alpha_q \rangle\\
	& = \frac{2}{\pi} \int_0^\infty \langle \alpha_q | \left(d + \mu^2 (d+2b)^{-1}\right)^{-1} | \alpha_q \rangle \; \d \mu \;.
\end{aligned}
\end{equation}
To compute the integrand, we use the Sherman--Morrison formula twice. First
\[
\begin{aligned}
	(d+2b)^{-1}
	& = \left(d + 2 g |n \rangle \langle n |\right)^{-1}
	= d^{-1} - \frac{2 g d^{-1} |n \rangle \langle n | d^{-1}}{1 + 2 g \langle n | d^{-1} | n \rangle}
\end{aligned}
\]
and in the second step
\[
\begin{aligned}
	\left( d + \mu^2 (d+2b)^{-1} \right)^{-1}
	& = \left( d + \mu^2 d^{-1} - \frac{2 g \mu^2}{1 + 2 g \langle n | d^{-1} | n \rangle} d^{-1} | n \rangle \langle n | d^{-1} \right)^{-1}\\
	& = (d + \mu^2 d^{-1})^{-1} + \frac{
		2 g \mu^2 (d^2 + \mu^2)^{-1} | n \rangle \langle n | (d^2 + \mu^2)^{-1}
	}{
		1 + 2 g \langle n | d^{-1} | n \rangle - 2 g \mu^2 \langle n | d^{-1} (d^2 + \mu^2)^{-1} | n \rangle
	}.
\end{aligned}
\]
Using this in \cref{eq:Bintegral} and using the integral identities \cref{eq:integralidentities} for the first summand, we get
\begin{align*}
	\textnormal{(B)}
	& = \frac{2}{\pi} \int_0^\infty \frac{\d \mu}{\lambda_{\alpha_q}^2 + \mu^2} \lambda_{\alpha_q} + \frac{2}{\pi} \int_0^\infty \frac{
		2 g \mu^2 \big| \langle \alpha_q | (d^2 + \mu^2)^{-1} | n \rangle \big|^2
	}{
		1 + 2 g \langle n | d^{-1} | n \rangle - 2 g \mu^2 \langle n | d^{-1} (d^2 + \mu^2)^{-1} | n \rangle
	} \; \d \mu\\
	& = 1 + \frac{2}{\pi} \int_0^\infty \frac{2 g n_{\alpha_q}^2 \mu^2 (\mu^2 + \lambda_{\alpha_q}^2)^{-2}}{1 + 2 g \sum_{\alpha \in \cI^+} n_{\alpha}^2 (\mu^2 + \lambda_\alpha^2)^{-1} \lambda_\alpha} \; \d \mu\;. \tagg{eq:Bfinalresult}
\end{align*}
This agrees with $ \textnormal{(A)} $ up to a replacement of $ -2 g n_{\alpha_q}^2 \lambda_{\alpha_q}^2 $ by $ 2 g n_{\alpha_q}^2 \mu^2 $ in the numerator.

\smallskip

\noindent \underline{Evaluation of $ \textnormal{(C)} $:} As in $ \textnormal{(A)} $ we use \cref{eq:dbabbreviation} to write
\begin{equation}
\label{eq:Cformula}
	\textnormal{(C)}
	= \lambda_{\alpha_q}^{-1} \langle \alpha_q| (d^2 + 2 \tilde{b})^{\frac{1}{2}} | \alpha_q \rangle \;.
\end{equation}
Applying the integral identities \cref{eq:integralidentities} and the Sherman--Morrison formula, we obtain
\begin{equation}
\begin{aligned}
	(d^2 + 2 \tilde{b})^{\frac{1}{2}}
	& = \frac{2}{\pi} \int_0^\infty \left( 1 - \frac{\mu^2}{\mu^2 + d^2 + 2 g |\tilde{n} \rangle \langle \tilde{n} |} \right) \; \d \mu\\
	& = d + \frac{2}{\pi} \int_0^\infty \mu^2 \frac{2 g (\mu^2 + d^2)^{-1} |\tilde{n} \rangle \langle \tilde{n} | (\mu^2 + d^2)^{-1} }{1 + 2 g \langle \tilde{n} | (\mu^2 + d^2)^{-1} | \tilde{n} \rangle } \; \d \mu\;.
\end{aligned}
\end{equation}
Plugging this into \cref{eq:Cformula} yields
\begin{equation}
\begin{aligned}
\label{eq:Cfinalresult}
	\textnormal{(C)}
	& = \lambda_{\alpha_q}^{-1} \langle \alpha_q| d | \alpha_q \rangle
	+ \frac{2}{\pi} \int_0^\infty \mu^2 \lambda_{\alpha_q}^{-1} \frac{2 g \big| \langle \alpha_q |(\mu^2 + d^2)^{-1} |\tilde{n} \rangle \big|^2 }{1 + 2 g \langle \tilde{n} | (\mu^2 + d^2)^{-1} | \tilde{n} \rangle } \; \d \mu\\
	& = 1 + \frac{2}{\pi} \int_0^\infty \frac{2 g n_{\alpha_q}^2 \mu^2 (\mu^2 + \lambda_{\alpha_q}^2)^{-2} }{1 + 2 g \sum_{\alpha \in \cI^+} n_{\alpha}^2 (\mu^2 + \lambda_\alpha^2)^{-1} \lambda_\alpha} \; \d \mu
	= \textnormal{(B)}\;.
\end{aligned}
\end{equation}

\noindent \underline{Evaluation of $\textnormal{(D)}$:} As for $\textnormal{(B)}$, we will make the factors of $ (d+2b)^{\frac{1}{2}} $ cancel. Let this time $ A := (d+2b)^{\frac{1}{2}} $ and $ B := d $. Then
\begin{equation}
\begin{aligned}
	((d+2b)^{\frac{1}{2}} d (d+2b)^{\frac{1}{2}})^{\frac{1}{2}}
	= (A B A)^{\frac{1}{2}}
	& = \frac{2}{\pi} \int_0^\infty \left( 1 - \frac{\mu^2}{A B A + \mu^2} \right) \; \d \mu\\
	& = \frac{2}{\pi} \int_0^\infty A B A (A B A + \mu^2)^{-1} \; \d \mu \;.
\end{aligned}
\end{equation}
Now, using $ Y X^{-1} Y = (Y^{-1} X Y^{-1})^{-1} $ for any matrices $ X $ and $ Y $, we get
\begin{equation}
\begin{aligned}
	& A B A (A B A + \mu^2)^{-1}
	= A B A \left( (A B A)^2 + \mu^2(A B A) \right)^{-1} A B A\\
	& = AB \left( B A^2 B + \mu^2 B \right)^{-1} B A
	= A \left(A^2 + \mu^2 B^{-1} \right)^{-1} A\;.
\end{aligned}
\end{equation}
Thus
\begin{align*}
	&(d+2b)^{-\frac{1}{2}} ((d+2b)^{\frac{1}{2}} d (d+2b)^{\frac{1}{2}})^{\frac{1}{2}} (d+2b)^{-\frac{1}{2}}
	= \frac{2}{\pi} \int_0^\infty (A^2 + \mu^2 B^{-1})^{-1} \; \d \mu\\
	& = d \frac{2}{\pi} \int_0^\infty \frac{\d \mu}{\mu^2 + d^2 + 2db}
	= d \frac{2}{\pi} \int_0^\infty \left( \frac{1}{\mu^2 + d^2} - \frac{2 g (\mu^2 + d^2)^{-1} d |n \rangle \langle n | \left(\mu^2 + d^2\right)^{-1}}{1 + 2 g \langle n | (\mu^2 + d^2)^{-1} d | n \rangle} \right) \d \mu\\
	& = 1 - \frac{2}{\pi} \int_0^\infty \frac{2 g (\mu^2 + d^2)^{-1} d^2 |n \rangle \langle n | \left(\mu^2 + d^2\right)^{-1}}{1 + 2 g \langle n | (\mu^2 + d^2)^{-1} d | n \rangle} \; \d \mu\;.
\end{align*}
In the expectation value of $ | \alpha_q \rangle $ we obtain
\begin{equation}
\label{eq:Dfinalresult}
	\textnormal{(D)}
	= 1
	- \frac{2}{\pi} \int_0^\infty \frac{2 g n_{\alpha_q}^2 \lambda_{\alpha_q}^2 (\mu^2 + \lambda_{\alpha_q}^2)^{-2}}{1 + 2 g \sum_{\alpha \in \cI^+} n_{\alpha}^2 (\mu^2 + \lambda_\alpha^2)^{-1} \lambda_\alpha} \; \d \mu
	= \textnormal{(A)}\;.
\end{equation}

Plugging \cref{eq:Afinalresult}, \cref{eq:Bfinalresult}, \cref{eq:Cfinalresult} and \cref{eq:Dfinalresult} into \cref{eq:ABCD} results in
\begin{equation}
(\cosh(2 K) - 1)_{\alpha_q, \alpha_q}
	= \frac{1}{\pi} \int_0^\infty \frac{2 g n_{\alpha_q}^2 (\mu^2 - \lambda_{\alpha_q}^2) (\mu^2 + \lambda_{\alpha_q}^2)^{-2}}{1 + 2 g \sum_{\alpha \in \cI^+} n_{\alpha}^2 (\mu^2 + \lambda_\alpha^2)^{-1} \lambda_\alpha} \; \d \mu\;.
\end{equation}
Proposition \ref{prop:main3} now follows by inserting this expression into \cref{eq:nqb}.
\end{proof}

\section{Conclusion of the Proof of \cref{thm:main,thm:jump}}
\label{sec:proofmain}
With \cref{thm:main2} and \cref{prop:main3} at hand, it remains to show that the replacements of $ Q_k(\mu)$ by $Q_k^{(0)}(\mu) $, of $ \lambda_{\alpha_q,k}$ by $\lambda_{q,k} $, and of $ \tilde{\cC}^q$ by $\cC^q $ only produce small errors. We start with the following estimate.

\begin{lemma}
\label{lem:QminusQnot}
Let $Q_k(\mu)$ as defined in \cref{eq:abbreviations2} and $Q_k^{(0)}(\mu)$ as defined in \cref{eq:abbreviations1}. Then
\[
Q_k(\mu) = Q_k^{(0)}(\mu) + \mathcal{O}(M^{\frac{1}{2}}N^{-\frac{1}{3}+\delta} + M^{-\frac{1}{2}} N^{ \delta} + N^{-\delta})\;,
\]
with the error bound being uniform in $\mu \in [0,\infty)$.
\end{lemma}
\begin{proof}
Recall that, with a different convention than \cite{BNPSS20}, in the present paper we chose $g_k = \frac{\hat{V}_k}{2 \hbar \kappa N |k|} \sim N^{-\frac{2}{3}}$.
According to \cite[Proposition~3.1 and (3.17)]{BNPSS20} we have
\[
Q_k(\mu) = 2 g_k \lvert k\rvert \kF^2 \sum_{\alpha \in \Ical_k^+} \frac{\lambda_{\alpha,k}^2}{\lambda_{\alpha,k}^2 + \mu^2} \sigma(p_\alpha) \left( 1 + \mathcal{O}\left(M^{\frac{1}{2}} N^{-\frac{1}{3}+\delta}\right)\right).
\]
Here $\sigma(p_\alpha) = 4\pi M^{-1} + \mathcal{O}\big( N^{-\frac{1}{3}} M^{-\frac{1}{2}} \big)$ is the surface measure of the patch $p_\alpha := \kF^{-1} P_\alpha$ on the unit sphere.
(In \cite{BNPSS20}, $ \lambda_{\alpha,k} $ is called $ u_\alpha(k)^2 $. Moreover it was assumed that $ \delta \le \frac{1}{6}- \frac{\varepsilon}{2} $ with $ \varepsilon \in \RRR $ being the parameter in the choice $ M = N^{\frac{1}{3} + \varepsilon} $.
As pointed out already in \cite[Lemma 5.1]{BNPSS21}, this constraint is superfluous.)
As in \cite[(5.15)--(5.17)]{BNPSS20}, we understand $ Q_k(\mu) $ as the Riemann sum of an integral over half of the unit sphere. This integral is
\[
\int_{\mathbb{S}_{\textnormal{half}}^2} \frac{\cos^2 \theta}{\cos^2 \theta + \mu^2} \di\theta \di\varphi = 2\pi \left(1 - \mu \operatorname{arctan}\left(\frac{1}{\mu}\right)\right)\;.
\]
The claimed error bound is then the same as following \cite[(5.17)]{BNPSS20}.
\end{proof}

We can now establish the following proposition, which is almost our main result, except that we still have $ \tilde{\cC}^q $ \eqref{eq:cCqtilde} in the place of  $ \cC^q $ \eqref{eq:cCq}.
\begin{proposition}
\label{prop:maintilde}
If $\hat{V}$ is non-negative and compactly supported, and if $ q \in B_{\alpha_q} $ for some $ 1 \le \alpha_q \le M $, then
\begin{equation}
\label{eq:maintilde}
	n_q
	= \sum_{k \in \tilde{\cC}^q\cap \Zbb^3} \frac{1}{\pi} \int_0^\infty \frac{g_k (\mu^2 - \lambda_{q,k}^2)(\mu^2 + \lambda_{q,k}^2)^{-2}}{1 + Q_k^{(0)}(\mu)} \; \d \mu + \cE \;,
\end{equation}
where
\begin{equation}
\label{eq:mainerrortilde}
\lvert \cE \rvert
	\le C \epsilon^{-1} N^{-\frac{2}{3} - \frac{1}{12}}\;.
\end{equation}
\end{proposition}
\begin{proof}
 With $n_q^{(\b)}$ as defined in \cref{eq:main3}, the error can be decomposed into three terms as
\begin{align*}
	\lvert \cE \rvert & \le | n_q - n_q^{(\b)} |
	+ \bigg\lvert n_q^{(\b)}
	- \sum_{k \in \tilde{\cC}^q\cap \ZZZ^3} \frac{1}{\pi} \int_0^\infty \frac{g_k (\mu^2 - \lambda_{\alpha_q,k}^2)(\mu^2 + \lambda_{\alpha_q,k}^2)^{-2}}{1 + Q_k^{(0)}(\mu)} \; \d \mu \bigg\rvert\\
	& \quad + \bigg\lvert \sum_{k \in \tilde{\cC}^q\cap \ZZZ^3} \frac{1}{\pi} \int_0^\infty \left( \frac{g_k (\mu^2 - \lambda_{\alpha_q,k}^2)(\mu^2 + \lambda_{\alpha_q,k}^2)^{-2}}{1 + Q_k^{(0)}(\mu)}
	- \frac{g_k (\mu^2 - \lambda_{q,k}^2)(\mu^2 + \lambda_{q,k}^2)^{-2}}{1 + Q_k^{(0)}(\mu)} \right) \d \mu \bigg\rvert\\
	& =: \cE_{\I} + \cE_{\II} + \cE_{\III}\;. \tagg{eq:threesuberrors}
\end{align*}
Theorem \ref{thm:main2} already shows that
\begin{equation}
\label{eq:cEIfinal}
	\cE_{\I} \le C N^{-1 + 2 \delta} \;.
\end{equation}
To express the other two error terms more compactly we define
\begin{align*}
	f_k(\mu) & := \frac{g_k (\mu^2 - \lambda_{\alpha_q,k}^2)(\mu^2 + \lambda_{\alpha_q,k}^2)^{-2}}{1 + Q_k(\mu)}\;, &
		f_k^{(0)}(\mu) & := \frac{g_k (\mu^2 - \lambda_{\alpha_q,k}^2)(\mu^2 + \lambda_{\alpha_q,k}^2)^{-2}}{1 + Q_k^{(0)}(\mu)}\;,\\
	\tilde{f}_k^{(0)}(\mu) & := \frac{g_k (\mu^2 - \lambda_{q,k}^2)(\mu^2 + \lambda_{q,k}^2)^{-2}}{1 + Q_k^{(0)}(\mu)}\;. \tagg{eq:fkbullet}
\end{align*}
According to \cref{prop:main3} we can write
\begin{equation}
	\cE_{\II}
	= \bigg\lvert \sum_{k \in \tilde{\cC}^q\cap \ZZZ^3} \frac{1}{\pi} \int_0^\infty (f_k(\mu) - f_k^{(0)}(\mu)) \; \d \mu \bigg\rvert
\end{equation}
while the error $ \cE_{\III} $ becomes
\begin{equation}
\label{eq:cEIIIformula1}
	\cE_{\III}
	= \bigg\lvert \sum_{k \in \tilde{\cC}^q\cap \ZZZ^3} \frac{1}{\pi} \int_0^\infty (f_k^{(0)}(\mu) - \tilde{f}_k^{(0)}(\mu)) \; \d \mu \bigg\rvert \;.
\end{equation}
To estimate $\mathcal{E}_{\II}$, we use that according to \cref{lem:QminusQnot} we have $Q_k(\mu) - Q_k^{(0)}(\mu) = \mathcal{O}(M^{\frac{1}{2}}N^{-\frac{1}{3}+\delta} + M^{-\frac{1}{2}} N^{ \delta} + N^{-\delta})$ uniformly in $\mu \in [0,\infty)$.
Therefore
\begin{align*}
	\cE_{\II}
	& = \bigg\lvert \sum_{k \in \tilde{\cC}^q\cap \ZZZ^3} \frac{1}{\pi} \int_0^\infty g_k \frac{\mu^2 - \lambda_{\alpha_q,k}^2}{(\mu^2 + \lambda_{\alpha_q,k}^2)^2} \frac{Q_k(\mu) - Q_k^{(0)}(\mu)}{(1+Q_k(\mu))(1+Q_k^{(0)}(\mu))}\, \d \mu \bigg\rvert \\
	& \leq \sum_{k \in \tilde{\cC}^q\cap \ZZZ^3} \frac{g_k}{\pi} \int_0^\infty \bigg\lvert \frac{\mu^2 - \lambda_{\alpha_q,k}^2}{(\mu^2 + \lambda_{\alpha_q,k}^2)^2} \bigg\rvert \, \d \mu \;  \mathcal{O}(M^{\frac{1}{2}}N^{-\frac{1}{3}+\delta} + M^{-\frac{1}{2}}N^{ \delta} + N^{-\delta})\;.		\tagg{eq:fint}
\end{align*}
Here we used $ Q_k(\mu) \geq 0$ and $Q_k^{(0)}(\mu) \ge 0 $. Since $(\mu^2 - \lambda_{\alpha_q,k}^2)(\mu^2 + \lambda_{\alpha_q,k}^2)^{-2}$ has the antiderivative $-\mu(\mu^2 + \lambda_{\alpha_q,k}^2)^{-1}$, the integral is found to be $\lambda_{\alpha_q,k}^{-1}$. Further, the opening angle of a patch is of order $ M^{- \frac{1}{2}} \ll \lambda_{\alpha_q,k} $, so
\begin{equation}
\label{eq:lambdaerror}
	\lambda_{\alpha_q,k}
	= |\hat{k} \cdot \hat{\omega}_{\alpha_q}|
	= |\hat{k} \cdot \hat{q}| (1 + \cO (M^{- \frac{1}{2}}))
	= \lambda_{q,k} (1 + \cO (M^{- \frac{1}{2}}))\;.
\end{equation}
The condition \eqref{eq:cQepsilon} from \cref{thm:main} implies $\lambda_{q, k}^{-1} \le \epsilon^{-1} $ and $\lambda_{\alpha_q, k}^{-1} \le \epsilon^{-1} $. We have $g_k \sim N^{-\frac{2}{3}}$ and the sum over $ k $ is finite, so we obtain
\begin{equation}
\label{eq:cEIIafinal}
	\cE_{\II}
	= \epsilon^{-1} N^{-\frac{2}{3}} \mathcal{O}(M^{\frac{1}{2}}N^{-\frac{1}{3}+\delta} + M^{-\frac{1}{2}} N^{ \delta} + N^{-\delta})\;.
\end{equation}
Now we turn to $\Ecal_\III$. We define the function $h_\mu: \RRR \to \RRR$ by
\[h_\mu(x) := \frac{\mu^2 - x^2}{(\mu^2 + x^2)^2}\;.\]
Recalling that $Q_k^{(0)}(\mu) \geq 0$ for all $\mu \in [0,\infty)$, and using the Fubini theorem, we obtain
\begin{align*}
 \cE_{\III}
	& = \bigg\lvert \sum_{k \in \tilde{\cC}^q\cap \ZZZ^3} \frac{1}{\pi} \int_0^\infty \frac{g_k}{1 + Q_k^{(0)}(\mu)} \left( \frac{\mu^2 - \lambda_{\alpha_q,k}^2}{(\mu^2 + \lambda_{\alpha_q,k}^2)^{2}} -
	 \frac{ \mu^2 - \lambda_{q,k}^2}{(\mu^2 + \lambda_{q,k}^2)^{2}}\right) \d \mu \bigg\rvert \\
	& = \bigg\lvert \sum_{k \in \tilde{\cC}^q\cap \ZZZ^3} \frac{1}{\pi} \int_0^\infty \frac{g_k}{1 + Q_k^{(0)}(\mu)} \int_{\lambda_{q,k}}^{\lambda_{\alpha_q,k}} \frac{\partial h_\mu}{\partial x}(x) \,\di x \, \d \mu \bigg\rvert \\
	& \leq  \sum_{k \in \tilde{\cC}^q\cap \ZZZ^3} \frac{g_k}{\pi} \int_0^\infty  \int_{\lambda_{q,k}}^{\lambda_{\alpha_q,k}} \left\vert\frac{\partial h_\mu}{\partial x}(x)\right\vert \,\di x \, \d \mu  \quad  = \sum_{k \in \tilde{\cC}^q\cap \ZZZ^3} \frac{g_k}{\pi}   \int_{\lambda_{q,k}}^{\lambda_{\alpha_q,k}} \int_0^\infty \left\vert\frac{\partial h_\mu}{\partial x}(x)\right\vert  \, \d \mu \,\di x \;.
\end{align*}
Since $\frac{\partial h_\mu}{\partial x}(x) = 2x\frac{x^2 - 3 \mu^2}{(\mu^2 + x^2)^3}$ we conclude that
\[
\cE_{\III} \leq \sum_{k \in \tilde{\cC}^q\cap \ZZZ^3} \frac{g_k}{\pi}   \int_{\lambda_{q,k}}^{\lambda_{\alpha_q,k}} \left\vert 2x\right\rvert  \bigg( - \int_0^{x/\sqrt{3}}   \frac{3\mu^2 - x^2}{(\mu^2 + x^2)^3}  \, \d \mu + \int_{x/\sqrt{3}}^\infty   \frac{3\mu^2 - x^2}{(\mu^2 + x^2)^3}  \, \d \mu \bigg) \,\di x \;. \]
An antiderivative of $(3\mu^2 - x^2)(\mu^2 + x^2)^{-3}$ w.\,r.\,t.~$\mu$ is given by $-\mu(\mu^2 + x^2)^{-2}$. Thus
\begin{align*}
\cE_{\III} & \leq \sum_{k \in \tilde{\cC}^q\cap \ZZZ^3} \frac{g_k}{\pi}   \int_{\lambda_{q,k}}^{\lambda_{\alpha_q,k}}   \frac{\left\vert x\right\rvert9}{4\sqrt{3}} \frac{\di x}{x^3}
= \frac{9}{4\pi \sqrt{3}} \sum_{k \in \tilde{\cC}^q\cap \ZZZ^3} g_k  \frac{\left\lvert\lambda_{q,k} - \lambda_{\alpha_q,k} \right\rvert}{\lambda_{\alpha_q,k} \lambda_{q,k}} \;.
\end{align*}
From \cref{eq:lambdaerror} we get, with a constant $C > 0$ independent of $k$, $M$, and $\epsilon$, the bound
\[
\frac{\left\lvert\lambda_{q,k} - \lambda_{\alpha_q,k} \right\rvert}{\lambda_{\alpha_q,k} \lambda_{\alpha_q,k}} \leq \lambda_{q,k}^{-1}\mathcal{O}(M^{-\frac{1}{2}}) \leq C \epsilon^{-1} M^{-\frac{1}{2}}\;.
\]
Since $g_k \sim N^{-\frac{2}{3}}$ and the sum over $k$ has finitely many summands, we conclude that
\begin{equation}
\label{eq:EIII}
 \Ecal_{\III} \leq \epsilon^{-1} \mathcal{O}\big( N^{-\frac{2}{3}} M^{-\frac{1}{2}} \big) \;.
\end{equation}
So collecting the estimates \cref{eq:cEIfinal}, \cref{eq:cEIIafinal}, and \cref{eq:EIII}, we conclude that
\[
\lvert \mathcal{E} \rvert \leq  C N^{-1 + 2 \delta} + \epsilon^{-1} N^{-\frac{2}{3}} \mathcal{O}(M^{\frac{1}{2}}N^{-\frac{1}{3}+\delta} + M^{-\frac{1}{2} }N^{ \delta} + N^{-\delta}) \;.
\]
Choosing the parameters $ \delta = \frac{1}{12} $ and $M = N^{\frac 13} $ yields the claimed bound.
\end{proof}

\begin{proof}[Proof of Theorem \ref{thm:main}]
By \cite[Thm.~1.1]{BPSS22} the ground state energy satisfies
\begin{equation}
	\inf \sigma(H_N) = E^{\HF}_N + E^{\RPA}_N + \cO(N^{-\frac 13 -\alpha})\;,
\end{equation}
where $ E^{\HF}_N $ and $ E^{\RPA}_N $ are explicit constants (the Hartree--Fock ground state energy and the Random Phase Approximation of the correlation energy) and  $ \alpha > 0 $. On the other hand, in \cite[Sect.~5.5]{BNPSS20}, the energy expectation in $ \psi_N $ was computed as
\[
\langle \psi_N, H_N \psi_N \rangle \leq E^{\HF}_N + {E}^{\RPA}_N + \mathfrak{E}\;,
\]
where the error term $\mathfrak{E}$ was bounded by
\[
\lvert \mathfrak{E}\rvert \leq C\left( N^{-1} + M^{-1} + N^{-1+\delta}M \right) + C \hbar \left( M^{\frac{1}{4}} N^{-\frac{1}{6} + \frac{\delta}{2}} + N^{-\frac{\delta}{2}} + M^{-\frac{1}{4}} N^{\frac{\delta}{2}} \right).
\]
The second term, of order $M^{-1}$, resulted from the bound on the linearization of the kinetic energy operator by $M^{-1} \Ncal$ in \cite[Sect.~5.3]{BNPSS20} and required a choice of $M \gg N^{\frac{1}{3}}$. This is not optimal; if instead of linearizing the kinetic energy operator directly we linearize only its commutator with a pair operator as in \cite[Lemma~8.2]{BNPSS21}, its contribution to $\mathfrak{E}$ can be improved from $M^{-1}$ to $\hbar (M^{-1/2} + M N^{-\frac{2}{3} + \delta}) $ as in \cite[Lemma~8.1]{BNPSS21}. With this improvement, the dominant terms in the error bound are
\[
\lvert \mathfrak{E} \rvert \leq C \hbar \left( N^{-\frac{2}{3}+\delta } M + N^{-\frac{\delta}{2}} + M^{-\frac{1}{4}} N^{\frac{\delta}{2}} \right)\;.
\]
With the choice $ \delta = \frac{1}{12}$ and $M = N^{\frac{1}{3}} $ from the proof of \cref{prop:maintilde} we obtain
\begin{equation}
	\langle \psi_N, H_N \psi_N \rangle \leq E^{\HF}_N + {E}^{\RPA}_N + \cO(N^{-\frac{1}{3} - \frac{1}{24}})\;.
\end{equation}
This establishes \cref{eq:mainenergy}.

\smallskip

It remains to establish \cref{eq:main}. Note that the sum in $ k $ is symmetric under reflection $ k \mapsto -k $, so we may replace $ \cD^q $ \eqref{eq:abbreviations1} by $ \cC^q $ \eqref{eq:cCq}. Then, \eqref{eq:main} follows from \cref{prop:maintilde} if we can show that extending the sum from $ k \in \tilde{\cC}^q $ to $ k \in \cC^q $ never decreases the result by more than $ C \epsilon^{-1} N^{-\frac{2}{3}- \frac{1}{12}} $. In fact, since \eqref{eq:cQepsilon} implies $ \lambda_{q,k} \ge \epsilon$, the same arguments as in the proof of \cref{eq:mainerrortilde} apply for any $ k \in \cC^q \setminus \tilde{\cC}^q $, so recalling \cref{eq:nqbformula} we get
\begin{equation}
\begin{split}
	& \left\vert
	\frac{1}{\pi} \int_0^\infty \frac{g_k (\mu^2 - \lambda_{q,k}^2)(\mu^2 + \lambda_{q,k}^2)^{-2}}{1 + Q_k^{(0)}(\mu)} \; \d \mu
	- n_q^{(\b)} \right\vert \\
	& =\left\vert
	\frac{1}{\pi} \int_0^\infty \frac{g_k (\mu^2 - \lambda_{q,k}^2)(\mu^2 + \lambda_{q,k}^2)^{-2}}{1 + Q_k^{(0)}(\mu)} \; \d \mu
	- \frac{1}{2 n_{\alpha_q, k}^2} \big( \cosh(2 K(k)) - 1 \big)_{\alpha_q, \alpha_q} \right\vert \\
	& \overset{\cref{eq:threesuberrors}}{\le} \cE_{\II} + \cE_{\III}
	\le C \epsilon^{-1} N^{-\frac{2}{3}- \frac{1}{12}} \;.
\end{split}
\end{equation}
Since $ \cosh(2 K) - 1 $ is a positive matrix, we have $ \big( \cosh(2 K(k)) - 1 \big)_{\alpha_q, \alpha_q} \ge 0 $. Since the number of momenta $ k \in (\cC^q \setminus \tilde{\cC}^q) \cap \ZZZ^3 $ is bounded by $ |B_R(0) \cap \ZZZ^3| \sim 1 $, we conclude
\[
	\sum_{k \in (\cC^q \setminus \tilde{\cC}^q) \cap \ZZZ^3} \frac{1}{\pi} \int_0^\infty \frac{g_k (\mu^2 - \lambda_{q,k}^2)(\mu^2 + \lambda_{q,k}^2)^{-2}}{1 + Q_k^{(0)}(\mu)} \; \d \mu 
	\ge - C \epsilon^{-1} N^{-\frac{2}{3}- \frac{1}{12}} \;.	\qedhere
\]
\end{proof}

\begin{proof}[Proof of Theorem \ref{thm:jump}]
By definition of $ \psi_N $, we have $ n_q = 0 $ whenever $ q $ is not inside some patch $ B_{\alpha_q} $. So we may focus on the case $ q \in B_{\alpha_q} $ for some $ 1 \le \alpha_q \le M $. \Cref{thm:main2} and \cref{prop:main3} combine to
\begin{equation}
\label{eq:nqbarupperbound}
	n_q
	\le \sum_{k \in \tilde{\cC}^q\cap \ZZZ^3} \frac{g_k}{\pi}  \int_0^\infty \frac{(\mu^2 - \lambda_{\alpha_q,k}^2)(\mu^2 + \lambda_{\alpha_q,k}^2)^{-2}}{1 + Q_k(\mu)} \; \d \mu
	+ C N^{- 1 + 2 \delta}\;.
\end{equation}
Since $Q_k(\mu) \geq 0$, for an upper bound the denominator of the integrand can be dropped; the integral is then the same as in \cref{eq:fint}, found there to be $\lambda_{\alpha_q,k}^{-1}$. For $ k \in \tilde{\cC}^q \cap \ZZZ^3 $, the definition \cref{eq:cCqtilde} of $ \tilde{\cC}^q $ entails
\begin{equation}
	\lambda_{\alpha_q,k}^{-1}
	= |k| |k \cdot \hat{\omega}_{\alpha_q}|^{-1}
	\le R N^\delta \;.
\end{equation}
The sum over such $ k $ is finite and $ g_k \sim N^{-\frac{2}{3}} $, so the leading order in \cref{eq:nqbarupperbound} is
\begin{equation}
	\sum_{k \in \tilde{\cC}^q\cap \ZZZ^3} \frac{g_k}{\pi}  \int_0^\infty \frac{(\mu^2 - \lambda_{\alpha_q,k}^2)(\mu^2 + \lambda_{\alpha_q,k}^2)^{-2}}{1 + Q_k(\mu)} \; \d \mu
	\le C N^{-\frac{2}{3}+ \delta}\;.
\end{equation}
Since $ \delta < \frac{1}{6}$, the error in \cref{eq:nqbarupperbound} is subleading. The optimal common upper bound is
	$n_q
	\le C N^{-\frac{2}{3}+ \delta}$,
where $ \delta = \frac{1}{12} $. This bound is uniform in $ q \in \ZZZ^3 $.
\end{proof}

\begin{proof}[Proof of Proposition \ref{prop:optimality}]
Comparing \cref{eq:maintilde}, which holds due to Proposition \ref{prop:maintilde}, and \cref{eq:optimality}, which is what we want to show, it suffices to bound the contributions from $ k \in \cC^q \setminus \tilde{\cC}^q $ by
\begin{equation}
\label{eq:cCqsetminuscCqtilde}
	\bigg\lvert \sum_{k \in (\cC^q \setminus \tilde{\cC}^q)\cap \ZZZ^3} \frac{1}{\pi} \int_0^\infty \frac{g_k (\mu^2 - \lambda_{q,k}^2)(\mu^2 + \lambda_{q,k}^2)^{-2}}{1 + Q_k^{(0)}(\mu)} \; \d \mu \bigg\rvert
	\le C \epsilon^{-1} N^{-\frac{2}{3} - \frac{1}{12}}\;.
\end{equation}
Comparing the definition \cref{eq:cCq} of $ \cC^q $ and the definition \cref{eq:cCqtilde} of $ \tilde{\cC}^q $, we observe that $ k \in \cC^q \setminus \tilde{\cC}^q $ can occur only
if the momentum $ q \pm k $ is outside the patch $ B_{\alpha_q} $,
or if $ |k \cdot \hat{\omega}_{\alpha_q}| < N^{-\delta} $.
The first case is ruled out by \cref{eq:edgeofthepatch}. The second case can be ruled out via \eqref{eq:cQepsilon} which implies $ \lambda_{q,k} \ge \epsilon $, and $ |k| \ge 1 $, as
\[
	 |k \cdot \hat{\omega}_{\alpha_q}|
	 = |k| \lambda_{\alpha_q,k}
	 \overset{\cref{eq:lambdaerror}}{=} |k| \lambda_{q,k} (1 + \cO(M^{-\frac{1}{2}}))
	 \ge \epsilon (1 + \cO(M^{-\frac 12}))\;.
\]
In that case, the sum on the l.\,h.\,s.\  of \cref{eq:cCqsetminuscCqtilde} is empty.
\end{proof}

\appendix
\section{Bosonization Approximation}
\label{app:motivationbosocc}

In this section we show how $ \ad^n_{q, (\b)} $, defined in \cref{eq:adnb}, arises from a bosonization approximation. We replace the almost-bosonic operators $ c^*(g)$, $c(g) $ defined in \cref{eq:cgstar} by exactly bosonic operators $ \tilde{c}^*(g)$, $\tilde{c}(g) $. \Cref{lem:shoelace} then shows that the multi-commutator $ \ad^n_S (a_q^* a_q) $ becomes $ \ad^n_{q, (\b)} $ with $ c^*$, $c $ replaced by $ \tilde{c}^*$, $\tilde{c} $.

\smallskip

The exact bosonic operators $ \tilde{c}^*$, $\tilde{c} $ can be defined as  elements of an abstract $ ^* $--algebra $ \cA $. More precisely, we define $ \cA $ to be the $ * $--algebra generated by
\begin{equation}
	\{ \tilde{a}^*_q, \tilde{c}_{p, h}^* \; : \; q \in \ZZZ^3, p \in \BFc, h \in \BF \}\;,
\end{equation}
where we impose the (anti-)commutator relations
\begin{equation}
\label{eq:commutationrelations}
\begin{aligned}
	&\{\tilde{a}_q, \tilde{a}^*_{q'}\} = \delta_{q, q'} \;, \quad
	[\tilde{c}_{p, h}, \tilde{c}_{p', h'}^*] = \delta_{p, p'} \delta_{h, h'} \;,\\
	&\{\tilde{a}_q, \tilde{a}_{q'}\}
	= \{\tilde{a}^*_q, \tilde{a}^*_{q'}\}
	= [\tilde{c}_{p, h}, \tilde{c}_{p', h'}]
	= [\tilde{c}^*_{p, h}, \tilde{c}^*_{p', h'}]
	= 0 \;.
\end{aligned}
\end{equation}
Furthermore, we impose that $ \tilde{c}^*_{p, h} $ behaves like a pair creation operator, that is
\begin{equation}
\label{eq:ctildeaacommutator}
	[\tilde{c}_{p, h}^*, \tilde{a}_q^* \tilde{a}_q] = -\tilde{c}^*_{p, h} (\delta_{h, q} + \delta_{p, q}) \;.
\end{equation}
In analogy to \cref{eq:cgstar}, for $ g: \BFc \times \BF \to \CCC $ we define $\tilde{c}^*(g) := \sum_{\substack{p \in \BFc \\ h \in \BF}} g(p, h) \tilde{c}^*_{p, h}$ and $\tilde{c}(g) := \sum_{\substack{p \in \BFc \\ h \in \BF}} \overline{g(p, h)} \tilde{c}_{p, h}$; moreover in analogy to \cref{eq:cstarabbreviation}, for $\alpha \in \cI_{k}$, we define
$\tilde{c}^*_\alpha(k) := \tilde{c}^*(d_{\alpha, k})$ with $d_{\alpha, k}(p, h)
	= \delta_{p, h \pm k} \frac{1}{n_{\alpha, k}} \chi(p, h : \alpha)$.
The statement of \Cref{lem:g} then holds true also for the modified operators, with $ g_{q, k} $ as defined in \cref{eq:g}:
\begin{equation}
\label{eq:exactcaarelation}
	[\tilde{c}^*_\alpha(k), \tilde{a}_q^* \tilde{a}_q] = - \delta_{\alpha, \alpha_q} \tilde{c}^*(g_{q, k})\;.
\end{equation}
The approximate CCR from \cref{lem:cgcommutator,lem:approxCCR,lem:dg} become exact, i.\,e., with $ \rho_{q, k} $ as defined in \cref{eq:dg} we have
\begin{equation}
\label{eq:exactCCR}
\begin{aligned}\relax
	[\tilde{c}(g), \tilde{c}^*(\tilde{g})]
	= \langle g, \tilde{g} \rangle\,, \quad	[\tilde{c}_\alpha(k), \tilde{c}^*_\beta(\ell)]
	= \delta_{\alpha, \beta} \delta_{k, \ell}\,, \quad	[\tilde{c}_\alpha(k), \tilde{c}^*(g_{q, \ell})]
	= \delta_{\alpha, \alpha_q} \delta_{k, \ell} \rho_{q, k}\,.
\end{aligned}
\end{equation}
Accordingly in \cref{eq:T} we replace $c$ and $c*$ by $\tilde{c}$ and $\tilde{c}^*$ to obtain $\tilde{S}$.
Thus we obtain an exactly bosonic equivalent $ \ad^n_{\tilde{S}} (\tilde{a}_q^* \tilde{a}_q) $ of $ \ad^n_S (a_q^* a_q) $. To compare it to the bosonized multi-commutator $ \ad^n_{q, (\b)} $ defined in \cref{eq:adnb}, we introduce an exact bosonic equivalent $ \widetilde{\ad}^n_{q, (\b)} $, given for $ n = 0 $ by $ \widetilde{\ad}^0_{q, (\b)} := \tilde{a}_q^* \tilde{a}_q $ and for $ n \ge 1 $ by
\begin{equation}
\label{eq:adnqbtilde}
	\widetilde{\ad}^n_{q, (\b)} := \begin{cases}\displaystyle
		2^{n-1} \tilde{\bA}_n + \tilde{\bB}_n + \tilde{\bB}_n^* + \sum_{m = 1}^{n-1} \binom{n}{m} \tilde{\bC}_{n-m, m} \quad &\text{if } n \text{ is even}\\
		\displaystyle \tilde{\bE}_n + \tilde{\bE}_n^* + \sum_{m = 1}^s \binom{n}{m} \tilde{\bD}_{n-m, m} + \sum_{m = 1}^s \binom{n}{m} \tilde{\bF}_{m, n-m} \quad &\substack{\displaystyle \text{if } n \text{ is odd},\\ \displaystyle
		n = 2s + 1,}
	\end{cases}
\end{equation}
where $ \tilde{\bA}$, $\tilde{\bB}$, $\tilde{\bC}$, $\tilde{\bD}$, $\tilde{\bE}$, and $ \tilde{\bF} $ are defined by replacing $ c^\sharp$ by $\tilde{c}^\sharp $ in \cref{eq:ABC} and \cref{eq:DEF}.
\begin{lemma}
\label{lem:shoelace}
Under the replacements of $ a^\sharp$ by $\tilde{a}^\sharp $ and $ c^\sharp $ by $\tilde{c}^\sharp $, the multi-commutator $ \ad^n_S (a_q^* a_q) $ becomes
\begin{equation}
\label{eq:shoelace}
	\ad^n_{\tilde{S}} (\tilde{a}_q^* \tilde{a}_q)
	= [\tilde{S}, \ldots, [\tilde{S}, \tilde{a}_q^* \tilde{a}_q] \ldots ]
	= \widetilde{\ad}^n_{q, (\b)}  \;.
\end{equation}
\end{lemma}
This motivates the definition of $ \ad^n_{q, (\b)} $ we used in \cref{subsec:bosonizationapprox}.
\begin{proof}
We use induction in $ n $. The case $ n = 0 $ is trivial, since $ \ad^0_{\tilde{S}} (\tilde{a}_q^* \tilde{a}_q) = \tilde{a}_q^* \tilde{a}_q = \widetilde{\ad}^0_{q, (\b)}$.

\smallskip

\noindent \underline{Step from $ n-1$ to $ n $ for $ n = 1 $}: Here, in \cref{eq:adnqbtilde} we have $ s = 0 $, so $ \widetilde{\ad}^1_{q, (\b)} = \tilde{\bE}_1 + \tilde{\bE}_1^* $. On the other side, using \cref{eq:exactcaarelation} and $ K(k) = K(k)^T$, we get
\begin{align*}
	\ad^1_{\tilde{S}} (\tilde{a}_q^* \tilde{a}_q)
	& = [\tilde{S}, \tilde{a}_q^* \tilde{a}_q]
	= -\frac{1}{2} \sum_{k \in \Gamma^{\nor}} \sum_{\alpha, \beta \in \cI_{k}} K(k)_{\alpha, \beta} [\tilde{c}^*_\alpha(k) \tilde{c}^*_\beta(k) - \mathrm{h.c.}, \tilde{a}_q^* \tilde{a}_q]\\
	& 	= \sum_{k \in \tilde{\cC}^q \cap \ZZZ^3} \sum_{\alpha \in \cI_{k}} K(k)_{\alpha, \alpha_q} \tilde{c}^*_\alpha(k) \tilde{c}^*(g_{q, k}) + \mathrm{h.c.} \qquad = \tilde{\bE}_1 + \tilde{\bE}_1^* \;.
\end{align*}
Here, we were able to restrict to $ \tilde{\cC}^q \cap \ZZZ^3 $ since $ g_{q, k} = 0 $ otherwise.

\smallskip

For the rest of this proof, we adopt the convention that  $ \sum_k $ runs over $ k \in \tilde{\cC}^q \cap \ZZZ^3 $ while $ \sum_{k'} $ over $ k' \in \Gamma^{\nor} $.

\smallskip

\noindent \underline{Step from $ n-1$ to $ n $ for even $ n $}: On the l.\,h.\,s.\  of \cref{eq:shoelace} we have
\begin{align*}
	&\ad^n_{\tilde{S}} (\tilde{a}_q^* \tilde{a}_q)
	= [\tilde{S}, \ad^{n-1}_{\tilde{S}} (\tilde{a}_q^* \tilde{a}_q)]\\
	& = [\tilde{S}, \tilde{\bE}_{n-1}] + [\tilde{S}, \tilde{\bE}_{n-1}^*] + \sum_{m = 1}^s \binom{n-1}{m} [\tilde{S}, \tilde{\bD}_{n-m-1, m}]  + \sum_{m = 1}^s \binom{n-1}{m} [\tilde{S}, \tilde{\bF}_{m, n-m-1}] \;,
\end{align*}
where $ s = n/2 - 1 $. The CCR \cref{eq:exactCCR} render
\begin{equation}
\begin{aligned}\relax
	[\tilde{S}, \tilde{\bE}_{n-1}]
	& = \frac{1}{2} \sum_{k, k'} \sum_{\substack{\alpha, \beta \in \cI_{k'} \\ \alpha_1 \in \cI_{k}}} K(k')_{\alpha, \beta} (K(k)^{n-1})_{\alpha_q, \alpha_1} [\tilde{c}_\alpha(k') \tilde{c}_\beta(k'), \tilde{c}^*(g_{q, k}) \tilde{c}^*_{\alpha_1}(k)]\\
	& = \tilde{\bC}_{n-1, 1} + \tilde{\bB}_n + \tilde{\bA}_n \;.
\end{aligned}
\end{equation}
It is easy to see that $ \tilde{\bC}_{m, m'}^* = \tilde{\bC}_{m', m} $, $ \tilde{\bA}_n^* = \tilde{\bA}_n = \bA_n $ and $ \tilde{S}^* = -\tilde{S} $, which implies
\begin{equation}
	[\tilde{S}, \tilde{\bE}_{n-1}^*]
	= \tilde{\bC}_{1, n-1} + \tilde{\bB}_n^* + \tilde{\bA}_n \;.
\end{equation}
Likewise, and using $ \tilde{\bD}_{m, m'}^* = \tilde{\bF}_{m', m} $, we compute
\begin{align*}
	[\tilde{S}, \tilde{\bD}_{n-m-1, m}]
	= &\tilde{\bC}_{n-m, m} + \tilde{\bC}_{m+1, n-m-1} + \tilde{\bA}_n \;, \\
	[\tilde{S}, \tilde{\bF}_{m, n-m-1}]
	= &\tilde{\bC}_{m, n-m} + \tilde{\bC}_{n-m-1, m+1} + \tilde{\bA}_n \;.
\end{align*}
We sum all $ 1 + \sum_{m = 1}^s \binom{n-1}{m} + \sum_{m = 1}^s \binom{n-1}{m} + 1 = 2^{n-1} $ commutators and get
\begin{align*}
	\ad^n_{\tilde{S}} (\tilde{a}_q^* \tilde{a}_q)
	& = 2^{n-1} \tilde{\bA}_n + \tilde{\bB}_n + \tilde{\bB}_n^* + \tilde{\bC}_{1, n-1} + \tilde{\bC}_{n-1, 1}\\
	& \quad + \sum_{m = 1}^s \binom{n-1}{m}(\tilde{\bC}_{n-m, m} + \tilde{\bC}_{m+1, n-m-1} + \tilde{\bC}_{m, n-m} + \tilde{\bC}_{n-m-1, m+1})\\
	& = 2^{n-1} \tilde{\bA}_n + \tilde{\bB}_n + \tilde{\bB}_n^* + \sum_{m = 1}^{n-1} \binom{n}{m} \tilde{\bC}_{n-m, m} \qquad = \widetilde{\ad}^n_{q, (\b)}\;.	\tagg{eq:adnSfinalresulteven}
\end{align*}

\smallskip

\noindent \underline{Step from $ n-1$ to $ n $ for odd $ n \ge 3 $}: The l.\,h.\,s.\  of \cref{eq:shoelace} is
\begin{align*}
	\ad^n_{\tilde{S}} (\tilde{a}_q^* \tilde{a}_q)
	 = 2^{n-2} [\tilde{S}, \tilde{\bA}_{n-1}] + [\tilde{S}, \tilde{\bB}_{n-1}] + [\tilde{S}, \tilde{\bB}_{n-1}^*] + \sum_{m = 1}^{n-2} \binom{n-1}{m} [\tilde{S}, \tilde{\bC}_{n-m-1, m}] \;.
\end{align*}
Since $ \tilde{\bA}_{n-1} \in \Cbb$ we have $ [\tilde{S}, \tilde{\bA}_{n-1}] = 0 $. Using the CCR as above, we get
\begin{align*}
	[\tilde{S}, \tilde{\bB}_{n-1}]
	= &\tilde{\bD}_{1, n-1} + \tilde{\bE}_n \;, &
	[\tilde{S}, \tilde{\bC}_{n-m-1, m}]
	= &\tilde{\bD}_{n-m, m} + \tilde{\bF}_{m+1, n-m-1}
\end{align*}
and $ [\tilde{S}, \tilde{\bB}_{n-1}^*] = \tilde{\bF}_{n-1, 1} + \tilde{\bE}_n^* $.
Putting all terms together and using $ \tilde{\bD}_{m, m'} = \tilde{\bD}_{m', m} $ and $ \tilde{\bF}_{m, m'} = \tilde{\bF}_{m', m} $ completes the induction step with $ s = \frac{n-1}{2} $:
\begin{align*}
	\ad^n_{\tilde{S}} (\tilde{a}_q^* \tilde{a}_q)
	& = \tilde{\bE}_n + \tilde{\bE}_n^* + \tilde{\bD}_{1, n-1} + \tilde{\bF}_{n-1, 1} + \sum_{m = 1}^{n - 2} \binom{n-1}{m} (\tilde{\bD}_{n-m, m} + \tilde{\bF}_{m+1, n-m-1})\\
	& = \tilde{\bE}_n + \tilde{\bE}_n^* + \sum_{m = 1}^s \binom{n}{m} (\tilde{\bD}_{n-m, m} + \tilde{\bF}_{m, n-m}) \qquad \qquad = \widetilde{\ad}^n_{q, (\b)}\;.	\qedhere
\end{align*}
\end{proof}

\section{Formal Infinite Volume Limit}
\label{subsec:infvolapp}

In this appendix we take the limit of the formula describing the momentum distribution as the size of the torus $L \to \infty$. This is formal in the sense that we do not control the error terms in the derivation of the momentum distribution uniformly in $L$. For simplicity we assume that $ V $ and thus also $ \hat{V} $ are radial. As long as the side length $ L $ of the torus is fixed, the proof of \cref{thm:main} carries through unchanged, so
\begin{equation}
\begin{aligned}
n_q(L)
	\approx &\sum_{k \in \tilde{\cC}^q\cap L^{-1} \ZZZ^3} \frac{1}{\pi} \int_0^\infty \frac{g_k (\mu^2 - \lambda_{q,k}^2)(\mu^2 + \lambda_{q,k}^2)^{-2}}{1 + Q_k^{(0)}(\mu)} \; \d \mu\;.
\end{aligned}
\end{equation}
Still $ g_k = \frac{\hat{V}_k}{2 \hbar \kappa N |k|} $, and \cref{eq:rho} renders $ \frac{N}{L^3} = 8 \pi^3 \rho = \frac{4 \pi k_{\F}^3}{3}(1 + \cO(N^{-\frac{1}{3}})) $, so
\[
	n_q(L)
	\approx L^{-3} \sum_{k \in \tilde{\cC}^q\cap L^{-1} \ZZZ^3} \frac{3 \hat{V}_k}{8 \pi^2 \hbar k_{\F}^3 |k| \kappa} (1 + \cO(k_{\F}^{-1})) \int_0^\infty \frac{ (\mu^2 - \lambda_{q,k}^2)(\mu^2 + \lambda_{q,k}^2)^{-2}}{1 + Q_k^{(0)}(\mu)} \; \d \mu\;.
\]
Note that $ \lambda_{q,k} = |\hat{k} \cdot \hat{q}| $ and $ Q_k^{(0)}(\mu) $ both depend on $ k $, but not on $ L $. So we are able to take the limit $ L \to \infty $, in which the Riemann sum $ L^{-3} \sum_k $ becomes an integral
\begin{equation}
\label{eq:integral}
\begin{split}
	n_q
	&	\approx \int_{\tilde{\cC}^q} \d k \frac{3 \hat{V}_k}{8 \pi^2 \hbar k_{\F}^3 |k| \kappa} (1 + \cO(k_{\F}^{-1})) \int_0^\infty \frac{ (\mu^2 - \lambda_{q,k}^2)(\mu^2 + \lambda_{q,k}^2)^{-2}}{1 + Q_k^{(0)}(\mu)} \; \d \mu\;.
	\end{split}
\end{equation}
For the approximate evaluation of this integral, we assume that the Fermi surface is locally flat and that $ q $ keeps sufficient distance to the boundary of its patch, so $ \tilde{\cC}^q = \cC^q $. The integral is then evaluated in spherical coordinates, as shown in \cref{fig:Cqreflection}. We consider only the case $ q \in B_{\F}^c $, as $ q \in B_{\F} $ can be treated analogously. The integrand is symmetric under reflection $ k \mapsto -k $, so we can replace $ \cC^q $ by $ \cD^q $, see \eqref{eq:abbreviations1}. The radial integral over $ |k| $ starts where the sphere of radius $ |k| $ touches the Fermi surface, which is at
\begin{equation}
 \label{eq:Rq}
 R_q := ||q| - k_{\F}| \;.
\end{equation}
The integral over $|k|$ runs up to the maximal momentum transfer $ R $ given by the diameter of $\supp \hat{V}$. The integration over $ \theta $ runs from 0 to $ \theta_{\max} $ with $ \cos \theta_{\max} \approx \frac{R_q}{|k|} =: \lambda_{\min} $. Thus
\begin{equation}
\begin{aligned}
	n_q
	\approx &\int_{R_q}^R \d |k| |k|^2 \frac{3 \hat{V}_k}{8 \pi^2 \hbar k_{\F}^3 |k| \kappa} \int_0^{\theta_{\max}} \d \theta \sin \theta \; 2 \pi \int_0^\infty \frac{\mu^2 - \cos^2 \theta}{(\mu^2 + \cos^2 \theta)^2} \frac{\d \mu}{1 + Q_k^{(0)}(\mu)}\\
	= &\int_{R_q}^R \d |k| |k| \frac{3 \hat{V}_k}{4 \pi \hbar k_{\F}^3 \kappa} \int_{\lambda_{\min}}^{1} \d \lambda\int_0^\infty \frac{\mu^2 - \lambda^2}{(\mu^2 + \lambda^2)^2} \frac{\d \mu}{1 + Q_k^{(0)}(\mu)}\;.
\end{aligned}
\end{equation}
\begin{figure}
	\centering
	\scalebox{0.7}{\begin{tikzpicture}

\fill[opacity = .1, blue] (-3.5,3.2) -- ({-3.2/sqrt(3)-0.75},3.2) -- ({3/sqrt(3)-0.75},-3) -- (-3.5,-3);
\draw[thick, blue] ({-3.2/sqrt(3)-0.75},3.2) -- ({3/sqrt(3)-0.75},-3);
\draw[blue] ({-2.6/sqrt(3)-0.75},2.6) -- ++(0.4,0.2) node[anchor = west]{$ \partial B_{\F} $};
\node[blue] at (-3,2.5) {$ B_{\F} $};
\node[gray] at (3,2.5) {$ B_{\F}^c $};

\draw[dashed] (-3,0) -- (3,0);
\filldraw[thick, blue!50!red, fill opacity = .1] (0,0) circle (2.5);
\draw[blue!50!red] (1.5,-2) -- ++(0.4,-0.2) node[anchor = west]{$ B_R(q) $};

\draw[dashed, blue] ({-3.2/sqrt(3)+0.75},3.2) -- ({3/sqrt(3)+0.75},-3);
\fill[pattern = north east lines, pattern color = blue!75!red, opacity = 0.2] (0.75,0) -- (2.5,0) arc(0 : {104.94135262633} : 2.5);
\fill[blue!75!red, opacity = 0.4] (-0.75,0) -- ({2.5*cos(135.0586473736610)},{2.5*sin(135.0586473736610)}) arc(135.0586473736610 : 180 : 2.5);

\fill[blue!75!red, opacity = 0.4] (-0.75,0) -- (-2.5,0) arc(180 : {180+104.94135262633} : 2.5);

\fill [red] (0,0) circle (0.08) node[anchor = north west]{$ q $};
\draw[line width = 2, blue!75!red, ->] (1.2,0.8) .. controls ++(0,-1) and ++(0.8,0.2) .. (-0.6,-1.5);

\end{tikzpicture}}
\hspace{4em}
	\scalebox{0.7}{\begin{tikzpicture}

\fill[opacity = .1, blue] (-3.5,-1) rectangle ++(7,-2);
\draw[thick, blue] (-3.5,-1) -- ++(7,0);
\draw[blue] (-3.3,-1) -- ++(0.2,0.4) node[anchor = south]{$ \partial B_{\F} $};
\node[blue] at (-3,-2.5) {$ B_{\F} $};
\node[gray] at (-3,2.5) {$ B_{\F}^c $};

\filldraw[thick, blue!50!red, fill opacity = .1] (0,0) circle (2.5);
\draw[blue!50!red] (1.5,2) -- ++(0.4,0.2) node[anchor = west]{$ B_R(q) $};

\draw[dashed] (0,0) circle (2);
\draw[line width = 2, opacity = .6, blue!75!red] (-{sqrt(3)},-1) arc ({270 - acos(0.5)} : {270 + acos(0.5)} : 2);
\draw[blue!75!red] (1.2,-1.6) -- ++(0.4,-0.8) node[anchor = north]{integration domain};

\draw[-, thick] (-{0.25*sqrt(3)},-0.25) arc ({270 - acos(0.5)} : {270} : 0.5);
\node at (-0.4, -0.6) {\footnotesize $ \theta_{\max} $};

\draw (0,0) -- (0,-2.5);
\draw[thick] (0,0) -- (0,-2);
\draw[thick] (0,0) -- (-{sqrt(3)},-1);
\draw (0,0) -- ++(-0.14,{0.14*sqrt(3)});
\draw (-{sqrt(3)},-1) -- ++(-0.14,{0.14*sqrt(3)});
\draw[<->] (-0.07,{0.07*sqrt(3)}) -- (-{sqrt(3)-0.07},{-1+0.07*sqrt(3)});
\node at (-1.1, -0.1) {\footnotesize $ |k| $};

\draw (0,-2.5) -- ++(0.24,0);
\draw (0,0) -- ++(0.4,0);
\draw[<->] (0.12,0) -- ++(0,-2.5);
\node at (0.3,-1.5) {\footnotesize $ R $};
\draw[<->] (0.3,0) -- ++(0,-1);
\node at (0.5,-0.6) {\footnotesize $ R_q $};

\fill [red] (0,0) circle (0.08) node[anchor = south west]{$ q $};

\end{tikzpicture}}
\caption{\textbf{Left:} Reflecting part of $ \cC^q $ renders $ \cD^q $. \textbf{Right:} The integration range for a fixed $ |k| $ in spherical coordinates.}
\label{fig:Cqreflection}
\end{figure}
Since the potential is radial, $ Q_k^{(0)}(\mu) $ depends only on $ |k| $ and not on $ \lambda $, so we may compute the integral over $ \lambda $ explicitly to be
\begin{equation}
	\int_{\lambda_{\min}}^{1}
	\frac{\mu^2 - \lambda^2}{(\mu^2 + \lambda^2)^2} \; \d \lambda
	= \left[ \frac{\lambda}{\mu^2 + \lambda^2} \right]_{\lambda_{\min}}^1
	= \frac{1}{1 + \mu^2} - \frac{R_q |k|^{-1}}{R_q^2 |k|^{-2} + \mu^2}\;.
\end{equation}
The final result is
\begin{equation}
\label{eq:nqbfinalresult}
n_q \approx \int_{R_q}^R \!\!\!\d |k| |k| \int_0^\infty \!\!\!\frac{3 \hat{V}_k}{4 \pi \hbar k_{\F}^3 \kappa} \left( \frac{1}{1 + \mu^2} - \frac{R_q |k|^{-1}}{R_q^2 |k|^{-2} + \mu^2} \right) \frac{\d \mu}{1 + Q_k^{(0)}(\mu)}\,.
\end{equation}

\section{Comparison with Daniel and Vosko}
\label{subsec:DV60comp}
The standard reference for the momentum distribution in the random phase approximation is \cite{DV60}. Daniel and Vosko use a Hellmann-Feynman argument to obtain the momentum distribution from a derivative of the ground state energy with respect to an artificial parameter in a modified Hamiltonian, where the energy is computed by the perturbative resummation of \cite{GB57}. This approach may not very reliable because the change in occupation numbers that the Hellmann-Feynman argument tests for corresponds to changes in the energy of order $\hbar^2$, which is beyond the energy resolution that the rigorous results provide. But at least formally we may compare \cref{eq:nqbfinalresult} to Daniel and Vosko's momentum distribution, given for the Coulomb potential and in the thermodynamic limit in \cite[Eq.~(8)]{DV60} for $ q \in B_{\F}^c $ as\footnote{The variables $ q, k $ and $ u $ from \cite{DV60} correspond to $ \frac{|k|}{k_{\F}}$, $\frac{|q|}{k_{\F}} $ and $ \mu $ in  our notation.}
\begin{equation}
\label{eq:nqDV}
\begin{aligned}
	n_q^{(\DV, \mathrm{out})}
	& = \frac{\alpha}{|q|} \int_{|q| - k_{\F}}^{|q| + k_{\F}} \d |k| |k| \int_0^\infty \Bigg[ \frac{|q| - \frac{|k|}{2}}{\left( |q| - \frac{|k|}{2} \right)^2 + k_{\F}^2 \mu^2} - \frac{\frac{|q|^2 - k_{\F}^2}{2 |k|}}{\left( \frac{|q|^2 - k_{\F}^2}{2 |k|} \right)^2 + k_{\F}^2 \mu^2} \Bigg] \\
	& \hspace{11em}\times
	\left(|k|^2 k_{\F}^{-2} + \alpha Q^{(\DV)}_k(\mu)\right)^{-1} \d \mu\;,
\end{aligned}
\end{equation}
with coupling constant $ \alpha = \frac{e_{\Coul}^2}{\pi^2 k_{\F}} $ and
\begin{equation}
\label{eq:QDV}
\begin{aligned}
	Q^{(\DV)}_k(\mu) & = 2 \pi \Bigg[ 1
	+ \frac{k_{\F}^2 (1 + \mu^2) - \frac{|k|^2}{4}}{2 |k| k_{\F}} \log\left( \frac{ \left( k_{\F} + \frac{|k|}{2} \right)^2 + k_{\F}^2 \mu^2}{\left( k_{\F} - \frac{|k|}{2} \right)^2 + k_{\F}^2 \mu^2} \right) \\
	& \qquad\qquad  - \mu \arctan \left( \frac{1 + \frac{|k|}{2 k_{\F}}}{\mu} \right)
	- \mu \arctan \left( \frac{1 - \frac{|k|}{2 k_{\F}}}{\mu} \right) \Bigg]\;.
\end{aligned}
\end{equation}
Due to the long range of the Coulomb potential, there is a separate formula \cite[Eq.~(9)]{DV60} for momenta inside the Fermi ball:
\begin{align*}
	&n_q^{(\DV, \mathrm{in})} \tagg{eq:nqDVin}\\
	& = \frac{\alpha}{|q|} \int_{k_{\F} - |q|}^{k_{\F} + |q|} \!\! \d |k| |k| \int_0^\infty \! \Bigg( \! \frac{|q| + \frac{|k|}{2}}{\left( \! |q| + \frac{|k|}{2} \! \right)^2 \!\! + k_{\F}^2 \mu^2} - \frac{\frac{k_{\F}^2 - |q|^2}{2 |k|}}{\left( \frac{k_{\F}^2 - |q|^2}{2 |k|} \right)^2 \!\! + k_{\F}^2 \mu^2} \! \Bigg) \!
	\frac{\d \mu}{|k|^2 k_{\F}^{-2} + \alpha Q^{(\DV)}_k(\mu)}\\
	&\quad + \frac{\alpha}{|q|} \int_{k_{\F} + |q|}^\infty \!\!\!\! \d |k| |k| \int_0^\infty \! \Bigg( \! \frac{|q| + \frac{|k|}{2}}{\left( \! |q| + \frac{|k|}{2} \! \right)^2 \!\! + k_{\F}^2 \mu^2} - \frac{\frac{|k|}{2} - |q|}{\left( \! \frac{|k|}{2} - |q| \! \right)^2 \!\! + k_{\F}^2 \mu^2} \! \Bigg) \!
	\frac{\d \mu}{|k|^2 k_{\F}^{-2} + \alpha Q^{(\DV)}_k(\mu)}.
\end{align*}
We take a short-range approximation of \cref{eq:nqDV} and \cref{eq:nqDVin} by cutting off the interaction at some $R$ independent of $N$, so that in particular $ |k| \leq R \ll k_{\F} $. This allows for simplifying $ Q^{(\DV)}_k(\mu) $; in fact, its contributions can be approximated as
\begin{equation}
	\frac{k_{\F}^2 (1 + \mu^2) - \frac{|k|^2}{4}}{2 |k| k_{\F}}
	= \frac{k_{\F}^2 (1 + \mu^2) + \cO(1)}{2 |k| k_{\F}}
	= \frac{k_{\F} (1 + \mu^2)}{2 |k|} (1 + \cO(k_{\F}^{-2}))
\end{equation}
and
\begin{align*}
	&\log\left( \frac{ \left( k_{\F} + \frac{|k|}{2} \right)^2 + k_{\F}^2 \mu^2}{\left( k_{\F} - \frac{|k|}{2} \right)^2 + k_{\F}^2 \mu^2} \right)
	= \log\left( \frac{ k_{\F}^2(1 + \mu^2) + k_{\F} |k| + \cO(1) }{ k_{\F}^2(1 + \mu^2) - k_{\F} |k| + \cO(1) } \right)
	= \frac{2 |k| (1 + \cO(k_{\F}^{-1}))}{ k_{\F}(1 + \mu^2)}
\end{align*}
and
\begin{equation}
	\mu \arctan \left( \frac{1 \pm \frac{|k|}{2 k_{\F}}}{\mu} \right)
	= \mu \arctan \left( \frac{1}{\mu} (1 + \cO(k_{\F}^{-1})) \right)
	= \mu \arctan \left( \frac{1}{\mu} \right) + \cO(k_{\F}^{-1})\;.
\end{equation}
Thus, with (SR) indicating the short-range approximation,
\begin{equation}
\label{eq:QSR}
	Q^{(\DV)}_k(\mu) = Q^{(\SR)}_k(\mu) + \cO(k_{\F}^{-1}) \quad \text{with} \quad
	Q^{(\SR)}_k(\mu) := 4 \pi \left( 1 - \mu \arctan \left( \frac{1}{\mu} \right) \right) \;.
\end{equation}
Then outside the Fermi ball and with $R_q$ as in \cref{eq:Rq}, \cref{eq:nqDV} becomes
\[
\label{eq:nqDVSR}
	n_q^{(\DV, \SR)}
	=  \int_{R_q}^{R} \!\! \d |k| \frac{\alpha k_{\F}}{|q| |k|} \int_0^\infty \! \left( \frac{1 + \cO(k_{\F}^{-1}) }{1 + \mu^2} - \frac{ R_q |k|^{-1} + \cO(k_{\F}^{-1})}{R_q^2 |k|^{-2} + \mu^2 } \right)
	\frac{\d \mu}{1 + \alpha |k|^{-2} k_{\F}^2 Q^{(\SR)}_k(\mu)}.
\]
A comparison with \cref{eq:nqbfinalresult} and \cref{eq:abbreviations1} shows that $ Q_k^{(0)}(\mu) $ should be identified with the quantity $ \alpha |k|^{-2} k_{\F}^2 Q^{(\SR)}_k(\mu) $, which corresponds to the following choice of the potential:
\begin{equation}
	\frac{3 \hat{V}_k}{2 \kappa \hbar k_{\F}} = \alpha |k|^{-2} k_{\F}^2 4 \pi \quad \Leftrightarrow \quad
	\hat{V}_k
	= \frac{8 \pi \kappa \hbar k_{\F}^3}{3} \alpha |k|^{-2} 
	= \frac{8 \kappa e_{\Coul}^2 \hbar k_{\F}^2}{3 \pi |k|^2}\;.
\end{equation}
With this identification, the $ \hat{V}_k $--dependent factor in \cref{eq:nqbfinalresult} amounts to
\begin{equation}
	 \frac{3 \hat{V}_k}{4 \pi \hbar k_{\F}^3 \kappa}
	 = \frac{2 e_{\Coul}^2}{\pi^2 k_{\F} |k|^2} = \frac{2 \alpha}{|k|^2}\;.
\end{equation}
As $ \frac{k_{\F}}{|q|} = 1 + \cO(k_{\F}^{-1}) $, we can equivalently write
\begin{align*}
	n_q^{(\DV, \SR)} \tagg{eq:nqDVSRfinal}
	& =\int_{R_q}^{R} \d |k| \frac{|k|  3 \hat{V}_k}{2 \pi \hbar k_{\F}^3 \kappa} \int_0^\infty \left( \frac{1 + \cO(k_{\F}^{-1}) }{1 + \mu^2} - \frac{ R_q |k|^{-1} + \cO(k_{\F}^{-1})}{R_q^2 |k|^{-2} + \mu^2 } \right)
	\frac{(1 + \cO(k_{\F}^{-1})) \d \mu}{1 + Q_k^{(0)}(\mu)} \;.
\end{align*}

Inside the Fermi ball, considering $ n_q^{(\DV, \mathrm{in})} $ in \cref{eq:nqDVin}, the second of the two integrals over $ |k| $ vanishes in the short-range approximation $ |k| \le R $ as soon as $ k_{\F} $ large enough. The first term is identical to $ n_q^{(\DV, \mathrm{in})} $ up to a replacement of $ |q| - \frac{|k|}{2} $ by $ |q| + \frac{|k|}{2} $ in two places,
of $ |q| - k_{\F} $ by $ k_{\F} - |q| $ in the integral limits,
and of $ |q|^2 - k_{\F}^2 $ by $ k_{\F}^2 - |q|^2 $ in two other places.
Thus, the expansion of $ n_q^{(\DV, \mathrm{in})} $ in the short-range approximation is identical to $ n_q^{(\DV, \SR)} $, again agreeing with half our result \cref{eq:nqbfinalresult} as $ k_{\F} \to \infty $.

\section*{Acknowledgments}
The authors were supported by the European Union (ERC \textsc{FermiMath} nr.~101040991). Views and opinions expressed are those of the authors and do not necessarily reflect those of the European Union or the European Research Council Executive Agency. Neither the European Union nor the granting authority can be held responsible for them. The authors were partially supported by Gruppo Nazionale per la Fisica Matematica in Italy.

\section*{Statements and Declarations}
The authors have no competing interests to declare.

\section*{Data Availability}
As purely mathematical research, there are no datasets related to the article.

\end{document}